\documentclass{amsart}
\usepackage{yfonts}
\usepackage{amsmath}%
\usepackage{amsfonts}%
\usepackage{dsfont}
\usepackage{amssymb}%
\usepackage{graphicx}
\usepackage{enumitem}
\usepackage{setspace}
\usepackage[dvipsnames]{xcolor}
\usepackage{comment}
\usepackage[utf8]{inputenc}
\usepackage{physics}
\usepackage{hyperref}
\usepackage{subcaption}
\usepackage{tikz}
\usetikzlibrary{patterns}

\theoremstyle{plain}
\newtheorem{theorem}{Theorem}[section]
\newtheorem{definition}[theorem]{Definition}

\newtheorem{corollary}[theorem]{Corollary}

\newtheorem{lemma}[theorem]{Lemma}
\newtheorem{notation}[theorem]{Notation}

\newtheorem{proposition}[theorem]{Proposition}
\newtheorem{remark}[theorem]{Remark}

\numberwithin{equation}{section}

\DeclareFontFamily{U}{mathx}{\hyphenchar\font45}
\DeclareFontShape{U}{mathx}{m}{n}{
      <5> <6> <7> <8> <9> <10>
      <10.95> <12> <14.4> <17.28> <20.74> <24.88>
      mathx10
      }{}
\DeclareSymbolFont{mathx}{U}{mathx}{m}{n}
\DeclareFontSubstitution{U}{mathx}{m}{n}
\DeclareMathAccent{\widecheck}{0}{mathx}{"71}
\DeclareMathAccent{\wideparen}{0}{mathx}{"75}

\newcommand{\cg}[6]{C{\mathstrut}^{#1,}_{#4,}{\mathstrut}^{#2,}_{#5,}{\mathstrut}^{#3}_{#6}}

\begin{document}
\title[On symbol correspondences for quark systems I]{On symbol correspondences for\\ quark systems I: Characterizations}
\author[P. A. S. Alc\^antara]{P. A. S. Alc\^antara}
\author[P. de M. Rios]{P. de M. Rios}
\thanks{This work was supported in part by Coordenação de Aperfeiçoamento de Pessoal de Nível Superior (CAPES), Brasil – Finance Code 001.}
\address{Instituto de Ci\^encias Matem\'aticas e de Computa\c{c}\~ao, Universidade de S\~ao Paulo. \newline
S\~ao Carlos, SP, Brazil.}
\email{pedro.antonio.alcantara@usp.br}
\email{prios@icmc.usp.br}



\begin{abstract}
We present the characterizations of symbol correspondences for mechanical systems that are symmetric by $SU(3)$, which we refer to as \emph{quark systems}. The quantum quark systems are the unitary irreducible representations of $SU(3)$ of class $(p,q)$, $p,q\in\mathbb N_0$, together with their operator algebras. 
We study symbol correspondences from quantum operators to smooth functions on the phase space of a classical quark system, when such a phase space is a (co)adjoint orbit: either the complex projective plane $\mathbb CP^2$ or the flag manifold that is the total space of a fiber bundle $\mathbb CP^1\hookrightarrow \mathcal E\to \mathbb CP^2$. In the first case, we refer to pure-quark systems and the characterization of their correspondences is given in terms of  characteristic numbers, similarly to the case of spin systems, cf. \cite{RS}. In the second case, we refer to general quark systems, particularly mixed-quark systems, and the characterization of their correspondences is given in terms of characteristic matrices, which introduces various novel features.  Furthermore, we present the $SU(3)$ decomposition of the product of quantum operators and their corresponding twisted products of classical functions, for both pure and mixed quark systems.  
\end{abstract}
\maketitle

\tableofcontents

\section{Introduction}\label{sec:intro}
Inspired by the treatment of symbol correspondences between quantum and classical mechanical systems which are symmetric by $SU(2)$, the so-called spin systems, cf. \cite{RS}, in this paper we start a generalization of this treatment to the study of symbol correspondences between quantum and classical mechanical systems which are symmetric under a compact Lie group $G$ by focusing on the case of $G=SU(3)$.\footnote{Initial efforts in this direction can be found in \cite{klim1, klim2, martins}.} Since $SU(3)$ is the symmetry group of the strong force, in Physics, here we call such systems \emph{quark systems}.

The first remarkable difference between spin and quark systems is that in the former case there is just one classical system, namely the Poisson algebra of smooth functions on $\mathbb CP^1 \simeq \mathcal S^2$, whereas in the latter case there are two types of symplectic phase space: the complex projective plane $\mathbb CP^2$, seen as a flag manifold $SU(3)/H$, with $H\simeq U(2)$, and the total flag manifold $\mathcal E\simeq SU(3)/T$, with $T\simeq U(1)\times U(1)\times U(1)$, that can be seen as a $\mathbb CP^1$-fiber bundle over $\mathbb CP^2$, denoted $\mathbb CP^1\hookrightarrow\mathcal E\rightarrow \mathbb CP^2$. These two kinds of symplectic manifolds are the two kinds of (co)adjoint orbits foliating the Lie algebra $\mathfrak{su}(3)$ or its unitary sphere $\mathcal S^7\subset \mathfrak{su}(3)$, seen as Poisson manifolds, with $\mathcal E$ being isomorphic to the leaves of the regular part of the foliation. Hence, for quark systems there are three kinds of classical system: (i) the one defined w.r.t.~symplectic manifold $\mathbb C P^2$; (ii) the one defined w.r.t.~symplectic manifold $\mathcal E$; and also (iii) a total one defined w.r.t.~the Poisson unitary sphere $\mathcal S^7\subset \mathfrak{su}(3)$.\footnote{For spin systems, the unitary sphere $\mathcal S^2\subset \mathfrak{su}(2)$ is isomorphic to all symplectic orbits, except for the trivial one, making the distinction between symplectic and Poisson manifolds irrelevant.} On the other hand, the quantum systems to be studied here are the ones defined on Hilbert spaces $\mathcal H_{(p,q)}$ with an irreducible representation of $SU(3)$ of class $(p,q)$ for $p,q\in\mathbb N_0$,\footnote{We use the convention which identifies the set of natural numbers as $\mathbb N=\{1,2,3,\cdots\}$ and denote $\mathbb N_0=\{0,1,2,3,\cdots\}$.} together with their operator algebras. 

We defer to our subsequent paper \cite{ARpreprint2} (referred to as Paper II) the study of any issues pertaining to the relations between quantum and classical systems when the latter is the total one defined on  Poisson manifold $\mathcal S^7\subset \mathfrak{su}(3)$. In this paper (referred to as Paper I), we only address symbol correspondences between quantum and classical systems when the latter is the Poisson algebra of smooth functions on a symplectic (co)adjoint orbit $\mathcal O\simeq\mathbb C P^2$ or $\mathcal O\simeq\mathcal E$. 

Then, the main question posed in this paper can be addressed in the following way: when/how is it possible to injectively map the space of  operators on $\mathcal H_{(p,q)}$ to the space of smooth functions on $\mathbb CP^2$, or $\mathcal E$, in an  $SU(3)$-equivariant way which also ensures that quantum observables give rise to classical observables?

To answer the question above we define \emph{symbol correspondences} in the spirit of what is already done in literature, especially in \cite{RS}, as linear injective maps
\begin{equation}\label{mapW}
 W:\mathcal B(\mathcal H_{(p,q)})\to C^\infty_{\mathbb C}(\mathcal O)\, : \ A\mapsto W_A \ ,   
\end{equation}
where $\mathcal B(\mathcal H_{(p,q)})$ is the space of operators on $\mathcal H_{(p,q)}$ and $\mathcal O\simeq\mathbb C P^2$, or $\mathcal O\simeq\mathcal E$, satisfying a few extra properties: \  (i) $W$ is $SU(3)$-equivariant, (ii) the image of any Hermitian operator is a real function and (iii) the  normalization condition 
\begin{equation}
    \int_{\mathcal O} W_{A}(\vb*\varsigma)d\vb*\varsigma = \dfrac{1}{\dim (p,q)} \tr(A)
\end{equation}
applies to every operator $A\in \mathcal B(\mathcal H_{(p,q)})$, with respect to a normalized left invariant integral on $\mathcal O$. Condition (ii) encodes that $W$ maps observables to observables and condition (iii) means it preserves expected values.

It turns out that for $\mathcal O\simeq \mathbb CP^2$ we can only define symbol correspondences for irreducible representations of type $(p,0)$ or $(0,q)$. We refer to the classical and quantum systems associated to $\mathbb CP^2$ and Hilbert spaces $\mathcal H_{(p,q)}$ with $pq=0$ as \emph{pure-quark systems}, since the pertinent irreducible representations of $SU(3)$ emerge from systems of $p$ quarks only, or $q$ antiquarks only. Characterization of symbol correspondences for pure-quark systems is very similar to what is known for spin systems. In particular, the correspondences for $\mathcal B(\mathcal H_{(p,0)})$, or $\mathcal B(\mathcal H_{(0,p)})$, are unequivocally determined by an ordered set of nonzero real numbers
\begin{equation}
    c_n \in \mathbb R^\times \ , \ \ 1\le n \le p \ ,
\end{equation}
called \emph{characteristic numbers}, so that the moduli space of symbol correspondences for pure-quark systems is $(\mathbb R^\times)^p$, cf. Theorem \ref{symb_corresp-ker} and Corollary \ref{moduli-space-pq}.

However, when $\mathcal O$ is the total flag manifold $\mathcal E$, we can define symbol correspondences for any $\mathcal H_{(p,q)}$, i.e.~any irreducible representation $(p,q)$ of $SU(3)$, so quantum systems with correspondences to smooth functions on $\mathcal E$ are \emph{general quantum quark systems}\footnote{Note that quantum pure-quark systems are special cases of  general quantum quark systems.} and, in particular, the quantum systems associated to $\mathcal H_{(p,q)}$ with $pq\neq 0$ shall be called  \emph{quantum mixed-quark systems}. Here we see some novel features and the characterization of symbol correspondences for general quark systems is now presented via full-rank complex matrices, called \emph{characteristic matrices}, cf. Theorem \ref{op-ker-gen}, so that the moduli space of symbol correspondences for a general quark system is a product of noncompact Stiefel manifolds, cf. Corollary \ref{moduli-space-gq}.

Just as it happens for spin systems, for both pure-quark and mixed-quark systems, Theorems \ref{symb_corresp-ker} and \ref{op-ker-gen} show that a symbol correspondence $W$ can be realized in terms of expectations  over a Hermitian operator with unitary trace $K$, called the \emph{operator kernel}, via
\begin{equation}
    W_A(g\vb*\varsigma_0) = \tr(AK^g) \ , 
\end{equation}
for any $A\in \mathcal B(\mathcal H_{(p,q)})$ and $g\in SU(3)$, where $\vb*\varsigma_0\in \mathcal O$ is a suitable point related to a choice of basis for $\mathcal H_{(p,q)}$, with $g\vb*\varsigma_0$ denoting the (co)adjoint action of $g$ on $\vb*\varsigma_0\in\mathcal O$ and $K^g$ denoting the action of $g$ on $K\in \mathcal B(\mathcal H_{(p,q)})$ by conjugation, 
\begin{equation*}
    SU(3)\ni g: K\mapsto \rho(g)K\rho(g)^{-1}=K^g \ , 
\end{equation*}
for $\rho$ being the irreducible $SU(3)$-representation on the respective quantum system. 

Thus, one can also interpret symbol correspondences for quark systems as expectation values over pseudo-states (Hermitian operators with unit trace which are not necessarily positive). When the operator kernel is also a positive operator, i.e. when $K$ is also a state\footnote{It must be emphasized that, whether $K$ is a state or just a pseudo-state, $K$ also has other defining properties for being the operator kernel of a symbol correspondence, in other words, not every state, or pseudo-state, is an operator kernel of a symbol correspondence.}, the correspondence maps positive(-definite) operators to (strictly-)positive functions and is called a \emph{mapping-positive correspondence}. 

On the other hand, if  correspondence $W$ induces an isometry between $\mathcal B(\mathcal H_{(p,q)})$ and the image set of $W$, for appropriately normalized invariant inner products in $\mathcal B(\mathcal H_{(p,q)})$ and $C^\infty_{\mathbb C}(\mathcal O)$, then $W$ is called a \emph{Stratonovich-Weyl correspondence}. As shown in \cite{AR} for spin systems, Theorem \ref{prop:mp_disjoint_sw} shows that the set of mapping-positive correspondences is disjoint from the set of Stratonovich-Weyl correspondence.

Adaptations of proofs from \cite{RS} show that, for pure-quark systems, the projector onto the highest weight space of $\mathcal H_{(p,0)}$ is an operator kernel, as well as the projector onto lowest weight space of $\mathcal H_{(0,q)}$, cf. Propositions \ref{ber-prop} and \ref{ber-prop2}. Then, Theorem \ref{ber-gen-prop} states that, for any $\mathcal H_{(p,q)}$, both the projector onto the highest weight space and the projector onto the lowest weight space, each is an operator kernel for a general quark system correspondence -- the failure of one of these projectors to be an operator kernel for a pure-quark system correspondence is due to a lack of $H$ symmetry required of an operator kernel for the map (\ref{mapW}) when $\mathcal O\simeq\mathbb CP^2\simeq SU(3)/H$. The symbol correspondences that these projectors define are called \emph{Berezin correspondences}, cf. Definitions \ref{st-al-ber-def} and \ref{ber-def-gen}. Besides being examples of mapping-positive correspondences, Berezin correspondences also exemplify a natural relation between correspondences for dual quark systems, which are defined as \emph{antipodal correspondences}, cf. Definition \ref{def:antipodal} and Theorem \ref{pq_ber_sw_ant}.

In this paper, we also start the study of \emph{twisted products of symbols}. These are defined for each symbol correspondence $W$ as the noncommutative product on the finite dimensional subspace  $W\big(\mathcal B(\mathcal H_{(p,q)})\big)\subset C^\infty_{\mathbb C}(\mathcal O)$ which is induced by the operator product in $\mathcal B(\mathcal H_{(p,q)})$ via $W$. To study these products, first we develop the $SU(3)$-invariant decomposition of the operator product in $\mathcal B(\mathcal H_{(p,q)})$, cf. Lemma \ref{op-prod-gen}, Definition \ref{Wigner-product} and Theorem \ref{op-prod-gen-s}. From this, in Theorems \ref{twist-prod-hs}, \ref{trik-hs-theo}, \ref{twist-prod-gen} and \ref{trik-gen-prop} we are able to present some explicit expressions for twisted products, and in Proposition \ref{prop:ant_rev_dyn} we show that antipodal correspondences can be seen as inducing a ``reverse symbolic dynamics'' via its twisted product, cf. Remark \ref{rmk:ant_rev_dyn}.

This paper is organized as follows.

In Section \ref{sec:SU(3)}, first we present some general facts concerning the Lie group $SU(3)$ and proceed in subsection \ref{irrepsu3} with the characterization of its irreducible unitary representations using the Gelfand-Tsetlin method for defining a standard basis of $\mathcal H_{(p,q)}$. Then, in subsection \ref{ss-cgs} we present the Clebsch-Gordan series of $SU(3)$ and proceed  with the $SU(3)$-invariant decomposition of the operator product, using various Wigner symbols in order to highlight the symmetries of this invariant decomposition. Finally, in subsection \ref{co-sec} we present the description of $\mathbb CP^2$ and $\mathcal E$ as (co)adjoint orbits of $SU(3)$, along with some of their properties.  

In Section \ref{sec:symb_c}, we present symbol correspondences for quark systems and their twisted products in a general setting, highlighting the common definitions and features that apply to both pure-quark and mixed-quark systems. 

In Sections \ref{sec:pq-sys} and \ref{sec:gen-sys} we work out the description of symbol correspondences for pure-quark systems and mixed-quark systems, respectively. The first subsection of each section presents the construction of harmonic functions on each respective symplectic manifold that constitutes the classical phase space (from which the pertinent irreducible representations of $SU(3)$ are identified as possible quantum quark systems to be studied in the subsequent subsection), and then proceeds with the $SU(3)$-invariant decomposition of the pointwise product of each of these harmonic functions. Then, the last subsection of each of these two sections is devoted to study the respective symbol correspondences and their twisted products.

Finally, in Section \ref{sec:conc} we briefly discuss  some results obtained in this paper, with  indication of some topics for future investigations, particularly the main topic of asymptotic analysis of twisted products of symbols to be studied in Paper II. 

To complete this Paper I, in Appendix \ref{sec:def_GT} we explain the Gelfand-Tsetlin method for $SU(3)$ used in Definition \ref{gt-basis}; in Appendix \ref{sec:mixed-casimirs} we explain a method based on mixed Casimir operators, used to prove the symmetries presented in Theorem  \ref{cg-sym-theo}; in Appendix \ref{sec:wig_symb} we detail the presentation of various Wigner symbols for $SU(3)$ which provide expressions for the decomposition of the operator product, similar to the ones for $SU(2)$; in Appendix \ref{sec:def_PQ} we justify the name in Definition \ref{QPQsystem}; in Appendix \ref{sec:ber_ker} we reproduce in greater detail the proof of Theorem \ref{ber-gen-prop} from \cite{figueroa, wild}.

\

\emph{Acknowledgements:} We thank Eldar Straume for stimulating initial discussions and interesting later comments. We also thank Igor Mencattini and Luiz San Martin for some pertinent comments.

\section{On the Lie group $SU(3)$ and its representations}\label{sec:SU(3)}
Let $SU(3)$ denote the special unitary subgroup of $GL_3(\mathbb C)$, satisfying  $\det g = 1$ and $gg^\dagger = g^\dagger g = \mathds 1$, for all $g\in SU(3)$. As a manifold, $SU(3)$ can be seen as a fiber bundle over $\mathcal S^5$ whose fibers are $\mathcal S^3\simeq SU(2)$ (see discussion in subsection \ref{co-sec}), hence it is a simply connected compact Lie group of real dimension $8$. The Lie algebra of $SU(3)$, denoted by $\mathfrak{su}(3)$, is generated by $i\lambda_k$, for $k=1,...,8$, where
\begin{equation}
    \begin{aligned}
        \lambda_1 & = \small{\begin{pmatrix}
        0 & 1 & 0 \\
        1 & 0 & 0 \\
        0 & 0 & 0
        \end{pmatrix}}  , \ 
        \lambda_2 = \small{\begin{pmatrix}
        0 & -i & 0 \\
        i & 0 & 0 \\
        0 & 0 & 0
        \end{pmatrix}}  , \ 
        \lambda_3 = \small{\begin{pmatrix}
        1 & 0 & 0 \\
        0 & -1 & 0 \\
        0 & 0 & 0
        \end{pmatrix}} , \\
        \lambda_4 & = \small{\begin{pmatrix}
        0 & 0 & 1 \\
        0 & 0 & 0 \\
        1 & 0 & 0
        \end{pmatrix}}  , \  
        \lambda_5 = \small{\begin{pmatrix}
        0 & 0 & -i \\
        0 & 0 & 0 \\
        i & 0 & 0
        \end{pmatrix}}  , \ 
        \lambda_6 = \small{\begin{pmatrix}
        0 & 0 & 0 \\
        0 & 0 & 1 \\
        0 & 1 & 0
        \end{pmatrix}}  ,\\
        &\lambda_7 = \small{\begin{pmatrix}
        0 & 0 & 0 \\
        0 & 0 & -i \\
        0 & i & 0
        \end{pmatrix}} , \ 
        \lambda_8 = \dfrac{1}{\sqrt{3}}\small{\begin{pmatrix}
        1 & 0 & 0 \\
        0 & 1 & 0 \\
        0 & 0 & -2
        \end{pmatrix}} 
    \end{aligned}
\end{equation}
\begin{table}[b]
    \centering
    \begin{tabular}{|c|c|c|c|c|c|c|c|c|c|}
    \hline
        $abc$ & $123$ & $147$ & $156$ & $246$ & $257$ & $345$ & $367$ & $458$ & $678$\\
        \hline
        $f^{abc}$ & $1$ & $1/2$ & $-1/2$ & $1/2$ & $1/2$ & $1/2$ & $-1/2$ & $\sqrt{3}/{2}$ & $\sqrt{3}/{2}$\\
    \hline
    \end{tabular}
    \caption{Structure constants for Gell-Mann matrices.}
    \label{tab:sc-gell-mann}
\end{table}
are Hermitian matrices, known as \emph{Gell-Mann matrices}, satisfying
\begin{equation}\label{commgen}
    \tr(\lambda_a\lambda_b) = 2\delta_{ab}  \ \ , \ \ \ 
    [\lambda_a,\lambda_b] = 2i\sum_{c=1}^8 f^{abc}\lambda_c \ \ , 
\end{equation}
with $f^{abc}$ totally antisymmetric and  determined by  Table \ref{tab:sc-gell-mann} -- see e.g. \cite{grein}.

Thus, $SU(3)$ is a simple Lie group. In order to describe its irreducible representations, we take the complexification $\mathfrak{sl}(3)=\mathfrak{su}(3)\oplus i\,\mathfrak{su}(3)$ and define, from
\begin{equation}\label{f-spin}
    F_k = \lambda_k/2 \ , 
\end{equation}
the following operators:
\begin{equation}\label{TUV}
    T_\pm  = F_1 \pm i F_2 \, , \ V_\pm = F_4 \pm i F_5\, , \ U_\pm = F_6 \pm i F_7 \, , \ T_3= F_3 \, , \ Y = \dfrac{2}{\sqrt{3}}F_8 \, . 
\end{equation}
Then, one can easily verify that
\begin{eqnarray}
    &T_{\pm}^\dagger = T_\mp \, , \ U_\pm^\dagger = U_\mp \, ,  \ V_\pm^\dagger = V_\mp \, ,  \ T_3^\dagger = T_3 \, , \ Y^\dagger=Y \, ,&  \label{ladder-dagger} \\
    &[T_3,Y]=[T_{\pm},Y]=0 \ ,& \label{rank2} 
\end{eqnarray}
and furthermore, 
\begin{equation}
    \tr(T_3T_3)  = \dfrac{1}{2} \ , \ \ \tr(YY)  = \dfrac{2}{3} \ , \ \ \tr(T_3Y) = 0 \ ,
\end{equation}
thus $\mathfrak{su}(3)$ is of rank $2$ and the set $\{iT_3, iY\}$ forms an orthogonal, but not normal basis of the Cartan subalgebra of $\mathfrak{su}(3)$, w.r.t.~the standard inner product on $\mathfrak{sl}(3)$,
\begin{equation}
    \langle X_1|X_2\rangle=\tr(X_1^\dagger X_2) \ . 
\end{equation}

Then, by defining 
\begin{equation}\label{u0v0}
    U_3 = \dfrac{3}{4}Y - \dfrac{1}{2}T_3 \ \ \ , \ \ \
    V_3 = \dfrac{3}{4}Y + \dfrac{1}{2}T_3 \ ,
\end{equation}
among the various commutation relations that follow from (\ref{commgen})-(\ref{TUV}), we have  
\begin{equation}\label{uv-com}
\begin{aligned}
&[T_3, T_\pm] = \pm T_\pm  \ , \ [U_3, U_\pm] = \pm U_\pm  \ , \ [V_3, V_\pm] = \pm V_\pm \ ,  \\
&[T_+, T_-] = 2T_3 \  , \ \  [U_+, U_-] = 2U_3 \  , \ \ [V_+, V_-] = 2V_3  \ , \\
&[U_3, T_\pm] = \mp T_\pm/2 \ ,  \  [V_3, T_{\pm}]= \pm T_\pm/2 \ , \ [U_3, V_\pm] = \pm V_\pm/2 \ ,  \\
&[V_3, U_\pm] = \pm U_\pm/2 \ , \ [T_3, V_\pm] = \pm V_\pm/2 \ , \ [T_3, U_\pm] = \mp U_\pm/2 \ , 
\end{aligned}
\end{equation}
and also
\begin{equation}\label{0-com}
[T_3, U_3] = [U_3, V_3] = [V_3, T_3] = 0 \ ,
\end{equation}
hence the root system of $SU(3)$ is composed by three root systems of $SU(2)$, with the same length, framing a regular hexagon as in Figure \ref{fig:root-su(3)}. Among the infinitely many $SU(2)\subset SU(3)$, the ones associated to $\{T_3, T_\pm\}$, $\{U_3, U_\pm\}$ and $\{V_3, V_\pm\}$ are singled out, respectively, as the $t$-, $u$- and $v$-\textit{standard} $SU(2)$ subgroups of $SU(3)$.

\begin{figure}[h]
    \centering
    \begin{tikzpicture}[scale = 0.7]
    \foreach\ang in {60,120,...,360}{
     \draw[->,black!80!black,thick] (0,0) -- (\ang:2cm);
    }
    \node[anchor=south west,scale=0.8] at (2,0) {$\alpha_1$};
    \node[anchor=south west,scale=0.8] at (1,1.7) {$\alpha_3$};
    \node[anchor=south west,scale=0.8] at (-1.5,1.7) {$\alpha_2$};
    
    \draw[dashed,pattern=north east lines,pattern color=blue] (30:2cm) -- (0,0) --(90:2cm);
  \end{tikzpicture}
    \caption{\centering Root diagram of $\mathfrak{su}(3)$.}
    \label{fig:root-su(3)}
\end{figure}
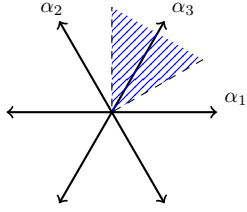

 The roots $\alpha_1$, $\alpha_2$, $\alpha_3$ are associated to the ladder operators $T_+$, $U_+$, $V_+$, respectively. We choose the fundamental Weyl chamber as the blue hatched one, enclosed by the dashed lines, so that $\{\alpha_1,\alpha_2,\alpha_3\}$ is the set of positive roots and $\{\alpha_1,\alpha_2\}$ is the set of simple roots. Let $\varpi_1$ and $\varpi_2$ be the fundamental weights satisfying
\begin{equation}\label{f-w-ort}
    2\dfrac{\ip{\varpi_j}{\alpha_k}}{\norm{\alpha_k}^2} = \delta_{j,k} \ , \ \ j,k\in \{1,2\} \ ,
\end{equation}
where $\ip{\cdot}{\cdot}$ is the canonical Euclidean inner product on the root space. Writing the fundamental weights as linear combination of the simple roots $\{\alpha_1,\alpha_2\}$, and using 
\begin{equation}
    \ip{\alpha_1}{\alpha_2} = -\norm{\alpha_1}\norm{\alpha_2}/2 \ \  , \ \ \ \norm{\alpha_1} = \norm{\alpha_2} \ \ ,
\end{equation}
the relations given by (\ref{f-w-ort}) imply 
\begin{equation}\label{fund-wei}
    \varpi_1 = \dfrac{1}{3}(2\alpha_1+\alpha_2) \ \   , \ \ \  \varpi_2 = \dfrac{1}{3}(\alpha_1+2\alpha_2) \ \ .
\end{equation}

\subsection{Irreducible representations and their GT basis}\label{irrepsu3}

We identify the lattice of weights of $\mathfrak{su}(3)$ with $\mathbb Z^2$ by
\begin{equation}
	(m,n) \equiv m\,\varpi_1 + n\,\varpi_2\, .
\end{equation}
Since $SU(3)$ is a simply connected compact group, its irreducible unitary representations are determined by the irreducible representations of its Lie algebra, and the classes of irreps of $\mathfrak{su}(3)$, as a simple Lie algebra, are determined by ordered pairs of nonnegative integers $(p,q)$, where $(p,q) \equiv p\,\varpi_1+q\,\varpi_2$ is the highest weight of the representation w.r.t.~the partial order
\begin{equation}\label{weight-order}
	\omega\ge\tau \ \ \mbox{if} \ \ \omega-\tau = c_1\alpha_1+c_2\alpha_2 \ \ \mbox{for nonnegative integers} \ (c_1,c_2)\, ,
\end{equation}
cf. \cite{hump}. We shall often refer to an unitary irreducible representation with highest weight $(p,q)$ in a less specific way simply as an irreducible representation $(p,q)$, or just as a representation $(p,q)$. Accordingly, for a representation $(p,q)$, $p$ and $q$ are the maximal integers such that $(T_-)^p$ and $(U_-)^q$ can be applied to a highest weight vector $\vb*e_{>}$ before vanishing, and an orthonormal basis of weight vectors for the representation $(p,q)$ on a complex Hilbert space $\mathcal H$ of dimension\footnote{Cf. Weyl dimensionality formula \cite{hump}.}
\begin{equation}\label{dim}
\dim (p,q) = \dfrac{(p+1)(q+1)(p+q+2)}{2}
\end{equation}
is obtained via linear combinations of the action of $(T_-)^a(U_-)^b(T_-)^c$ on  $\vb*e_{>}$. Then, any weight $\omega$ of a representation $(p,q)$ can be expressed as a linear combination of the fundamental weights, $\omega = m\,\varpi_1+n\,\varpi_2$, for some $m,n\in\mathbb Z$ such that $m/2,n/2$ are the eigenvalues of $T_3,U_3$, respectively, according to the actions
\begin{equation}\label{action_Tpm_Upm}
	\begin{cases}
		T_{\pm}:\omega\mapsto \omega\pm\alpha_1\\
		U_{\pm}:\omega\mapsto \omega\pm\alpha_2
	\end{cases}\, ,
\end{equation}
and $(p,q) \ge (m,n)$. Although the highest weight of an irrep occurs only once, some other weights may have multiplicity, as illustrated in Figure \ref{fig:basis-diag}(B).  

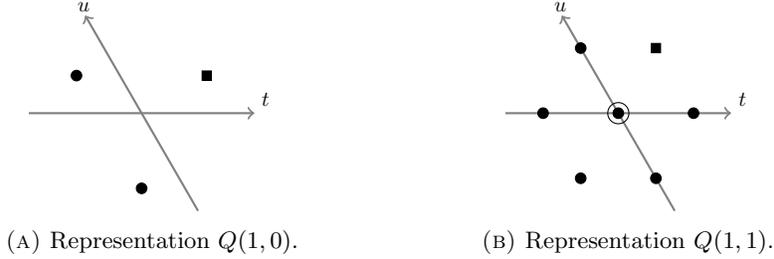
\begin{figure}[h!]
    \begin{subfigure}[b]{0.5\textwidth}
    \centering
    \begin{tikzpicture}[scale = 0.5]
    \foreach\ang in {0,120}{
     \draw[->,gray!80!gray,thick] (0,0) -- (\ang:3cm);
    }
    \foreach\ang in {180,300}{
     \draw[-,gray!80!gray,thick] (0,0) -- (\ang:3cm);
    }
    
    \node[anchor=south west,scale=0.8] at (0:3cm) {$t$};
    \node[anchor=south west,scale=0.8] at (127:3.15cm) {$u$};
    \foreach \ang in {150, 270}{
    \filldraw[black] (\ang:2cm) circle (4pt);
    }
    
    \filldraw[black] (30:2cm) +(-3.5pt,-3.5pt) rectangle +(3.5pt,3.5pt);
    \end{tikzpicture}
    \subcaption{Representation $(1,0)$.}
    \end{subfigure}%
    \begin{subfigure}[b]{0.5\textwidth}
    \hspace{4 em} \begin{tikzpicture}[scale = 0.5]
    \foreach\ang in {0,120}{
     \draw[->,gray!80!gray,thick] (0,0) -- (\ang:3cm);
    }
    \foreach\ang in {180,300}{
     \draw[-,gray!80!gray,thick] (0,0) -- (\ang:3cm);
    }
    \node[anchor=south west,scale=0.8] at (0:3cm) {$t$};

    \node[anchor=south west,scale=0.8] at (127:3.15cm) {$u$};

    \foreach \ang in {120,180,...,360}{
    \filldraw[black] (\ang:2cm) circle (4pt);
    }
    
    \filldraw[black] (60:2cm) +(-3.5pt,-3.5pt) rectangle +(3.5pt,3.5pt);

    \filldraw[black] (0,0) circle (4pt);
    \draw (0,0) circle (8pt);
    \end{tikzpicture}
    \subcaption{Representation $(1,1)$.}
    \end{subfigure}
    \caption{\centering Examples of weight diagrams for $SU(3)$. Each highest weight is highlighted as a square dot and the multiplicities of a weight are represented by rings around the weight.}
    \label{fig:basis-diag}
\end{figure}

In what follows we need to resolve the multiplicities of weights, so we resort to labeling each weight $\omega$ of $(p,q)$ by a triple $(\nu_1,\nu_2,\nu_3)$ of nonnegative integers, \emph{using an extra nonnegative half-integer index $J$ to distinguish weights with multiplicity}, such that $(\nu_1,\nu_2,\nu_3, J)$ satisfy the \emph{Gelfand-Tsetlin pattern}
\begin{equation}\label{const}
\begin{aligned}0\le r_- \le q \le r_+ & \le p + q \ , \ r_- \le \nu_3 \le r_+ \ , \\
\nu_1 = p+2q-(r_++r_-) \ , & \ \nu_2 = (r_++r_-)-\nu_3 \ , \ J = \dfrac{1}{2}(r_+-r_-) \ ,
\end{aligned}
\end{equation}
for $r_+$ and $r_-$ integers. In Appendix \ref{sec:def_GT}, we present an explanation of the Gelfand-Tsetlin method for $SU(3)$. In this Gelfand-Tsetlin labeling of weights,
\begin{equation}\label{tuv-nu}
    t = (\nu_1-\nu_2)/2 \ , \ \ u = (\nu_2-\nu_3)/2 \ , \ \ v = (\nu_1-\nu_3)/2
\end{equation}
are the eigenvalues\footnote{Henceforth the weights from which we are taking the eigenvalues will be specified or will be clear from the presence or absence of subscript and superscript in $t$, $u$ and $v$.} of the operators $T_3$, $U_3$, $V_3$, respectively, so that 
\begin{equation}\label{stepactions}
\begin{cases}
  T_\pm : (\nu_1,\nu_2,\nu_3)\mapsto (\nu_1\pm 1,\nu_2\mp 1, \nu_3) \ , \\
  U_\pm : (\nu_1,\nu_2,\nu_3)\mapsto (\nu_1,\nu_2\pm 1, \nu_3 \mp1) \ , \\ 
  V_\pm : (\nu_1, \nu_2, \nu_3)\mapsto (\nu_1\pm 1, \nu_2, \nu_3\mp 1) \ ,
\end{cases}
\end{equation}
and the index $J$ is the total spin number of the subrepresentation of the $u$-standard\footnote{For our purposes, this choice of standard $SU(2)$ subgroup is more convenient, but in the physics literature, it is more common to use the total spin number of the standard $SU(2)$ subgroup generated by $\{T_3,T_\pm\}$, which is often called the isospin \cite{grein}.} $SU(2)$. In particular, the highest weight is labelled by
\begin{equation}\label{GThighestweight}
   r_+ = q \ ,  \ r_- = \nu_3 = 0 \implies  (\nu_1,\nu_2,\nu_3)=(p+q,q,0)  \ ,  \ J = q/2 \ .
\end{equation} 
Because the Gelfand-Tsetlin labeling distinguishes multiplicities, we use it to define a standard basis for any irreducible representation of class $(p,q)$ as follows. 

\begin{definition}[cf. e.g. \cite{baird}] \label{gt-basis}
A \emph{Gelfand-Tsetlin basis}, or simply a \emph{GT basis} of an irreducible representation of $SU(3)$ of class 
\begin{equation}\label{bfp}
    \vb* p=(p,q)
\end{equation}
is an orthonormal basis $\{\vb*e(\vb* p;\vb*\nu, J)\}$ indexed by the triple
\begin{equation}
\vb* \nu = (\nu_1,\nu_2,\nu_3)     
\end{equation}
and the total spin number $J$ of $\{U_3,U_{\pm}\}$, as specified above, cf. (\ref{const})-(\ref{stepactions}), which is constructed by fixing a highest weight vector
\begin{equation}\label{hwvector}
    \vb*e_{>}= \ \vb*e((p,q);(p+q,q,0),q/2) \ \in \ \mathcal H \ , 
\end{equation}
cf. (\ref{GThighestweight}), and determining the other basis vectors by 
\begin{equation}\label{d-1}
\begin{aligned}
& U_-\vb*e((p,q);(\nu_1,\nu_2,\nu_3), J) =\\
& \hspace{2 em}\sqrt{(J+u)(J-u+1)}\,\vb*e((p,q);(\nu_1,\nu_2-1,\nu_3+1), J) \, , \\
    & T_-\vb*e((p,q);(\nu_1,\nu_2,\nu_3), J) = \\
    &\hspace{2 em}\sqrt{\dfrac{(J-u)(p+q+J-u-\nu_3+1)(q+J-u-\nu_3)(u-J+\nu_3+1)}{2J(2J+1)}}\\
    & \hspace{3 em}\times \vb*e((p,q);(\nu_1-1,\nu_2+1,\nu_3),J-1/2)\\
    & + \sqrt{\dfrac{(J+u+1)(p+q-J-u-\nu_3)(J+u+1+\nu_3-q)(J+u+2+\nu_3)}{2(J+1)(2J+1)}}\\
    & \hspace{3 em}\times \vb*e((p,q);(\nu_1-1,\nu_2+1,\nu_3),J+1/2) \ .
\end{aligned}
\end{equation}
\end{definition}
\begin{remark}\label{minimalindet}
Thus, for any irreducible representation of class $\vb* p=(p,q)$, its GT basis is uniquely determined up to a choice of phase for $\vb*e_{>}$ as in (\ref{hwvector}). This ``minimal indeterminacy'' in the definition of a standard orthonormal  basis for any irreducible representation of class $(p,q)$ is fundamental for all that follows.  
\end{remark}

Let $\rho:SU(3)\to \mathcal B(\mathcal H)$ be an unitary irrep on  $\mathcal H$ with highest weight $(p,q)$. Considering the anti-isomorphism $\mathcal H^\ast\leftrightarrow \mathcal H$ via inner product, for the dual representation $\widecheck{\rho}\leftrightarrow \rho$, we get that $\widecheck T_3\leftrightarrow - T_3$, $\widecheck U_3 \leftrightarrow -U_3$, $\widecheck V_3 \leftrightarrow - V_3$, $\widecheck T_\pm \leftrightarrow - T_\mp$, $\widecheck U_\pm \leftrightarrow - U_\mp$ and $\widecheck V_\pm \leftrightarrow - V_\mp$. Thus, the states of $\widecheck \rho$ are related to the states of $\rho$ by 
\begin{equation}\label{I-dual}
\widecheck J = J \ ; \ \  \widecheck t = -t \ , \ \widecheck u = -u \ , \ \widecheck v = -v \ , 
\end{equation}
which implies
\begin{equation}\label{nu-dual}
    \widecheck \nu_1 = p+q - \nu_1 \ \ \ \ , \ \ \ \ \widecheck \nu_2 = p+q - \nu_2 \ \ \ \ , \ \ \ \ \widecheck \nu_3 = p+q - \nu_3 \, .
\end{equation}
Therefore, from the above and (\ref{const})-(\ref{stepactions}), we have that $(q,p) = q\,\varpi_1+p\,\varpi_2$ is the highest weight of $\widecheck \rho$.

\begin{notation}
In light of this dualization symmetry, we introduce the  notation: 
\begin{equation}
  \mathbb N_0\times \mathbb N_0\ni \vb* p = (p,q) \leftrightarrow \ \widecheck{\vb* p} = (q,p) \ , \ \ \mbox{with} \ \ |\vb*p|=|\widecheck{\vb* p}| = p+q \, .
\end{equation}
Then, to better state the relations in (\ref{nu-dual}), we define
\begin{equation}\label{Delta}
    \Delta_{\vb*\nu, \vb*\mu}^{|\vb*p|}\equiv\Delta_{\vb*\nu, \vb*\mu}^{p+q} := \begin{cases}
   \ 1 \ , \ \ \mbox{if} \ \ \vb*\nu + \vb*\mu = (p+q,p+q,p+q)=(|\vb*p|,|\vb*p|,|\vb*p|) \\
    \ 0 \ , \ \ \mbox{otherwise}
    \end{cases} .
\end{equation}
\end{notation}

\begin{definition}\label{gt-dual}
For a Gelfand-Tsetlin basis $\{\vb* e(\vb*p; \vb*\nu, J)\}$ of a representation $\vb*p = (p,q)$, the \emph{induced Gelfand-Tsetlin basis of the dual representation} $\widecheck{\vb*p} = (q,p)$ is the basis comprised by the vectors
\begin{equation}\label{dual-eq}
     \widecheck{\vb*e}(\widecheck{\vb* p};\widecheck{\vb*\nu}, J) = (-1)^{2(t_{\vb*\nu}+u_{\vb*\nu})}\vb*e^\ast(\vb*p;\vb*\nu, J) \, ,
\end{equation}
where $\vb*e^\ast(\vb*p;\vb*\nu,J)$ is the dual of $\vb*e(\vb*p;\vb*\nu,J)$ via the Hermitian inner product on the space of $\vb*p$, and $t_{\vb*\nu}$ and $u_{\vb*\nu}$ stand as in (\ref{tuv-nu}), with $\vb*\nu$ and $\widecheck{\vb*\nu}$ satisfying the duality relations (\ref{nu-dual}), that is, using (\ref{Delta}), 
\begin{equation}\label{Deltanu}
\mbox{duality:} \ \  \vb*\nu\leftrightarrow\widecheck{\vb*\nu} \ \iff \  \Delta^{|\vb*p|}_{\vb*\nu,\widecheck{\vb*\nu}}=1 \, .
\end{equation}
\end{definition}

\begin{remark}\label{ind-dual-inv}
A GT basis $\{\vb*e(\vb*p;\vb*\nu, J)\}$ and the induced GT basis $\widecheck{\vb*e}(\widecheck{\vb*p};\widecheck{\vb*\nu}, J)$ are related by an involution. Considering the natural isomorphism between a finite dimensional vector space and its double dual, the dual GT basis induced by $\widecheck{\vb*e}(\widecheck{\vb*p};\widecheck{\vb*\nu}, J)$ is precisely $\{\vb*e(\vb*p;\vb*\nu, J)\}$, that is,
\begin{equation}
    \vb*e(\vb*p;\vb*\nu, J) = (-1)^{2(t+u)}\widecheck{\vb*e}^\ast(\widecheck{\vb*p};\widecheck{\vb*\nu}, J) \ .
\end{equation}
This contrasts with standard convention for irreducible representations of $SU(2)$,  since for an $SU(2)$-representation with spin number $j$, there is a phase $(-1)^{2j}$ between a standard basis and the basis of the double dual space induced by the basis of the dual space, c.f. \cite{RS}.
\end{remark}

\begin{definition}\label{wign-fun}
The \emph{Wigner $D$-functions} (in the GT basis) of an irreducible unitary $SU(3)$-representation $\rho$ of class $\vb*p$ are the functions
\begin{equation}
    D^{\vb*p}_{\vb*\nu J, \vb*\mu L}(g) = \ip{\,\vb*e(\vb*p;\vb*\nu, J)\,}{\,\rho(g)\vb*e(\vb*p; \vb* \mu, L)\,} \, .
\end{equation}
\end{definition}

Using the conjugate symmetry of the inner product and the relation $\ip{v}{w} = \ip{w^\ast}{v^\ast}$ between inner products of $\mathcal H$ and $\mathcal H^\ast$, we get, for $\Delta^{|\vb*p|}_{\vb*\nu,\widecheck{\vb*\nu}} = \Delta^{|\vb*p|}_{\vb*\mu,\widecheck{\vb*\mu}} = 1$, 
\begin{equation}\label{conj-wfun}
\begin{aligned}
    \overline{D^{\vb*p}_{\vb*\nu J, \vb*\mu L}}(g) & = \ip{\,\rho(g)\vb*e(\vb*p; \vb* \mu, L)\,}{\,\vb*e(\vb*p;\vb*\nu, J)\,} = \ip{\vb*e^\ast(\vb*p;\vb*\nu, J)}{\widecheck\rho(g)\vb*e^\ast(\vb*p; \vb* \mu, L)}\\
    & =(-1)^{2(t_{\vb*\nu}+u_{\vb*\nu}+t_{\vb*\mu}+u_{\vb*\mu})}\ip{\,\widecheck{\vb*e}(\widecheck{\vb*p};\widecheck{\vb*\nu}, J)\,}{\,\widecheck\rho(g)\widecheck{\vb*e}(\widecheck{\vb*p}; \widecheck{\vb* \mu}, L)\,}\\
    & = (-1)^{2(t_{\vb*\nu}+u_{\vb*\nu}+t_{\vb*\mu}+u_{\vb*\mu})}D^{\widecheck{\vb*p}}_{\widecheck{\vb*\nu} J, \widecheck{\vb*\mu} L}(g) \, .
\end{aligned}
\end{equation}

\subsection{Clebsch-Gordan series and the operator algebra}\label{ss-cgs}

An irreducible unitary representation $\rho$ on $\mathcal H$ with highest weight $\vb*p$ extends to a unitary representation (with respect to the trace inner product) on $\mathcal B(\mathcal H)$ via the conjugation by $\rho$,
\begin{equation}\label{act-op}
\mathcal B(\mathcal H)\to\mathcal B(\mathcal H) :      A\mapsto A^g = \rho(g)A\rho(g)^{-1} \hspace{2 em} \forall g \in SU(3) \, .
\end{equation}
Now, we recall that $\mathcal B(\mathcal H)$ is naturally isomorphic to $\mathcal H\otimes \mathcal H^\ast$ in a manner that (\ref{act-op}) matches the representation $\rho\otimes \widecheck \rho$ on $\mathcal H\otimes \mathcal H^\ast$.

The decomposition of a tensor product of irreducible unitary $SU(3)$-represen\-ta\-tions into a direct sum of irreducible unitary $SU(3)$-representations is known as the \emph{Clebsch-Gordan series} of $SU(3)$.

\begin{theorem}[cf. e.g. \cite{cole}]\label{cg-s}
The Clebsch-Gordan series of $SU(3)$ is given by
\begin{equation}\label{cg-s-eq1}
  (p_1,q_1)\otimes (p_2,q_2) = \bigoplus_{n=0}^{\min\{p_1,q_2\}}\bigoplus_{m=0}^{\min\{p_2,q_1\}} (p_1-n,p_2-m;q_1-m,q_2-n) \ ,
\end{equation}
where
\begin{equation}\label{cg-s-eq2}
\begin{aligned}
    (r_1,r_2;s_1,s_2) = &\, (r_1+r_2,s_1+s_2)\oplus\left(\bigoplus_{k=1}^{\min\{r_1,r_2\}}(r_1+r_2-2k,s_1+s_2+k)\right)\\
    & \ \ \oplus\left(\bigoplus_{k=1}^{\min\{s_1,s_2\}}(r_1+r_2+k,s_1+s_2-2k)\right) \ .
\end{aligned}
\end{equation}
\end{theorem}

\begin{corollary}\label{cg-series-corol}
For $p_1=q_2=p$ and $q_1=p_2=q$, the Clebsch-Gordan series assumes the form
\begin{equation}
\begin{aligned}
   (p,q)\otimes (q,p) = &\, \bigoplus_{n=0}^{p}\bigoplus_{m=0}^{q}\Bigg\{(p+q-n-m,p+q-n-m)\\
    & \ \ \oplus \Bigg[\bigoplus_{k=1}^{\min\{p-n,q-m\}}\Bigg((p+q-n-m-2k,p+q-n-m+k)\\
    & \ \ \oplus (p+q-n-m+k,p+q-n-m-2k)\Bigg)\Bigg]\Bigg\} \ .
\end{aligned}
\end{equation}
\end{corollary}

Note that an irreducible unitary representation of class $\vb*a$ may appear more than once in the CG series of $\vb*p_1\otimes \vb*p_2$, even if $\vb*p_2 = \widecheck{\vb*p}_1$.

\begin{notation}\label{mult-ind-not}
We shall denote the multiplicity of $\vb*a=(a,b)$ in the Clebsch-Gordan series of $\vb*p_1\otimes \vb*p_2$ by 
\begin{equation}
    \textswab{m}(\vb*p_1,\vb*p_2;\vb*a) \, .
\end{equation}
To distinguish multiple appearances of the same class of representation in a Clebsch-Gordan series, we will write
\begin{equation}
   (\vb*a;\sigma)\equiv (a,b;\sigma) \ , 
\end{equation}
where the index $\sigma$ counts the multiplicity starting from $1$ to $\textswab{m}(\vb*p_1,\vb*p_2;\vb*a)$.
\end{notation}

We thus provide two basis for $\vb*p_1\otimes \vb*p_2$.

\begin{definition}\label{unc-gt-basis}
An \emph{uncoupled GT basis} of the tensor product representation $\vb*p_1\otimes \vb*p_2$ is a basis comprised by the tensor product of GT basis $\{\vb*e(\vb*p_1;\vb*\nu_1,J_1)\}$ of $\vb*p_1$ and $\{\vb*e(\vb*p_2;\vb*\nu_2,J_2)\}$ of $\vb*p_2$. In particular, when $\vb* p_1 = \widecheck{\vb* p}_2 \equiv \mathbf p$, we take a GT basis of $\vb* p$ and use the induced GT basis for $\widecheck{\vb* p}$.
\end{definition}

\begin{definition}
A \emph{coupled GT basis} of the tensor product representation $\vb*p_1\otimes \vb*p_2$ is the union of GT basis $\{\vb*e_{\vb*p_1,\vb*p_2}((\vb*a;\sigma);\vb*\nu,J)\}$ of each $(\vb*a;\sigma)\equiv (a,b;\sigma)$ in the Clebsch-Gordan series of $\vb*p_1\otimes \vb*p_2$.
\end{definition}

\begin{notation}
To simplify notation, whenever clear from context, we write 
\begin{equation}\label{e=e}
    \vb*e((\vb*a;\sigma);\vb*\nu,J)\equiv \vb*e_{\vb*p_1,\vb*p_2}((\vb*a;\sigma);\vb*\nu,J) \, . 
\end{equation}
Also, unless specified otherwise, from now on we shall always refer to the uncoupled and coupled basis of the tensor product as meaning their respective GT basis, and likewise for the Clebsch-Gordan coefficients defined below.
\end{notation}

Both basis are orthonormal, so they are related by a unitary transformation.

\begin{definition}
The \emph{Clebsch-Gordan coefficients} (in the GT basis) are the entries of the unitary transformation that relates a coupled and an uncoupled GT basis of $\vb*p_1\otimes \vb*p_2$:
\begin{equation}
    \cg{\vb*p_1}{\vb*p_2}{(\vb*a;\sigma)}{\vb*\nu_1 J_1}{\vb*\nu_2 J_2}{\,\,\vb*\nu J} =
    \ip{\, \vb*e((\vb*a;\sigma);\vb*\nu,J)\, }{\,\vb*e(\vb*p_1;\vb*\nu_1,J_1)\otimes \vb* e(\vb*p_2;\vb*\nu_2,J_2) \,} \, ,
\end{equation}
where $\ip{\cdot}{\cdot}$ is the $SU(3)$-invariant inner product induced by the ones on each representation of the tensor product.
\end{definition}

From the definition,
\begin{equation}\label{unc-coup}
\vb*e(\vb*p_1;\vb*\nu_1,J_1)\otimes \vb* e(\vb*p_2;\vb*\nu_2,J_2) = \sum_{\substack{(\vb*a;\sigma)\\\vb*\nu, J}} \cg{\vb*p_1}{\vb*p_2}{(\vb*a;\sigma)}{\vb*\nu_1 J_1}{\vb*\nu_2 J_2}{\,\,\vb*\nu J} \vb*e((\vb*a;\sigma);\vb*\nu,J)\, ,
\end{equation}
\begin{equation}\label{coup-unc}
\vb*e((\vb*a;\sigma);\vb*\nu,J) = \sum_{\substack{\vb*\nu_1, J_1\\\vb*\nu_2, J_2}}
    \cg{\vb*p_1}{\vb*p_2}{(\vb*a;\sigma)}{\vb*\nu_1 J_1}{\vb*\nu_2 J_2}{\,\,\vb*\nu J}\vb*e(\vb*p_1;\vb*\nu_1,J_1)\otimes \vb* e(\vb*p_2;\vb*\nu_2,J_2) \ .
\end{equation}
Moreover, we are able to fix a relative phase between a coupled and an uncoupled basis so that all Clebsch-Gordan coefficients are real. Usually, one chooses some set of Clebsch-Gordan coefficients to be positive and the remaining coefficients are completely determined via the action of the step operators on the basis vectors, cf. e.g. \cite{desw} -- we shall return to this problem later in this section. What is important now is that we take Clebsch-Gordan coefficients as real numbers, so that we have
\begin{equation}\label{ortho-cg}
\begin{aligned}
    &\sum_{\substack{(\vb*a;\sigma)\\\vb*\nu, J}}\cg{\vb*p_1}{\vb*p_2}{(\vb*a;\sigma)}{\vb*\nu_1 J_1}{\vb*\nu_2 J_2}{\,\,\vb*\nu J}\cg{\vb*p_1}{\vb*p_2}{(\vb*a;\sigma)}{\vb*\nu_1' J_1'}{\vb*\nu_2' J_2'}{\,\,\vb*\nu J} = \delta_{\vb*\nu_1,\vb*\nu_1'}\delta_{\vb*\nu_2,\vb*\nu_2'}\delta_{J_1,J_1'}\delta_{J_2,J_2'} \ , \\
    &\sum_{\substack{\vb*\nu_1,J_1\\ \vb*\nu_2, J_2}}\cg{\vb*p_1}{\vb*p_2}{(\vb*a;\sigma_1)}{\vb*\nu_1 J_1}{\vb*\nu_2 J_2}{\,\,\vb*\nu J}\cg{\vb*p_1}{\vb*p_2}{(\vb*b;\sigma_2)}{\vb*\nu_1 J_1}{\vb*\nu_2 J_2}{\,\,\vb*\nu' J'} = \delta_{\vb*a,\vb*b}\delta_{\sigma_1,\sigma_2}\delta_{\vb*\nu,\vb*\nu'}\delta_{J,J'} \ .
\end{aligned}
\end{equation}

\begin{remark}
From the way the GT basis for $SU(3)$-representations were constructed using $SU(2)$-subrepresentations, the $SU(3)$ Clebsh-Gordan coefficientes in the GT basis are related to the $SU(2)$ Clebsh-Gordan coefficientes by 
\begin{equation}
  \cg{\vb*p_1}{\vb*p_2}{(\vb*a;\sigma)}{\vb*\nu_1 J_1}{\vb*\nu_2 J_2}{\,\,\vb*\nu J} = \cg{J_1}{J_2}{J}{u_{\vb*\nu_1}}{u_{\vb*\nu_2}}{u_{\vb*\nu}}\cg{\vb*p_1}{\vb*p_2}{(\vb*a;\sigma)}{t_{\vb*\nu_1} J_1}{t_{\vb*\nu_2} J_2}{\,\,t_{\vb*\nu} J} \ ,
\end{equation}
where the second coefficient on the r.h.s. is called \emph{isoscalar factor}, and this provides explicit equations for the $SU(3)$ Clebsh-Gordan coefficients in the GT basis in terms of explicit equations for the $SU(2)$ Clebsh-Gordan coefficientes, as found in \cite{RS} and \cite{varsh}, for instance.
\end{remark}

From decompositions (\ref{unc-coup})-(\ref{coup-unc}), we obtain some sufficient conditions for the Clebsch-Gordan coefficients to be zero. Since $\vb*e(\vb*p_1;\vb*\nu_1,J_1)$ and $\vb*e(\vb*p_2;\vb*\nu_2,J_2)$ are basis vectors of $SU(2)$-representations with spin numbers $J_1$ and $J_2$, their tensor product is a vector of the tensor product of the $SU(2)$-representations they belong to, that is, the Clebsch-Gordan coefficients are zero if $J_1$, $J_2$ and $J$ do not satisfy the triangle inequality or their sum is not integer. Also, using superscripts to identify the $SU(3)$-representations, the operators $T_{3}$ and $U_3$ in $\vb*p_1\otimes \vb*p_2$ have the form
\begin{equation*}
    \bigoplus_{(\vb*a;\sigma)} T^{(\vb*a;\sigma)}_3  = T^{\vb*p_1}_3\otimes \mathds{1}+\mathds 1\otimes T^{\vb*p_2}_3 \ , \ \ \ 
    \bigoplus_{(\vb*a;\sigma)} U^{(\vb*a;\sigma)}_3  = U^{\vb*p_1}_3\otimes \mathds{1}+\mathds 1\otimes U^{\vb*p_2}_3 \ ,
\end{equation*}
where $\mathds 1$ is the identity operator. Thus, the Clebsch-Gordan coefficients are zero if $t\neq t_1+t_2$ or $u \neq u_1+u_2$, for $t$, $t_1$, $t_2$, $u$, $u_1$ and $u_2$ being the eigenvalues of $T_3$ and $U_3$ related to the weights $\vb*\nu$, $\vb*\nu_1$ and $\vb*\nu_2$.

To summarize, let $\delta(x,y,z)$ be equal to $1$ if $x$, $y$ and $z$ satisfy the triangle inequality and $x+y+z\in\mathbb Z$, or $0$ otherwise, and let
\begin{equation}
    \nabla_{\vb*\nu,\vb*\mu} := \delta_{t_{\vb*\nu},t_{\vb*\mu}}\delta_{u_{\vb*\nu},u_{\vb*\mu}} \ ,
\end{equation}
where $\delta_{m,n}$ is the Kronecker delta. Then,
\begin{equation}\label{cg-ne-0}
    \cg{\vb*p_1}{\vb*p_2}{(\vb*a;\sigma)}{\vb*\nu_1 J_1}{\vb*\nu_2 J_2}{\,\,\vb*\nu J}\neq 0  \ \ \ \ \implies \ \ \ \ \begin{cases}
    \ \nabla_{\vb*\nu_1+\vb*\nu_2, \vb*\nu} = 1 \\
    \ \delta(J_1,J_2,J) = 1
    \end{cases} \ .
\end{equation}

We have avoided until now the problem of specifying a decomposition for degenerate representations in general Clebsch-Gordan series. If $\vb*p_1$, $\vb*p_2$ and $\vb*a$ are representations such that $\textswab{m}(\vb*p_1,\vb*p_2;\vb*a)>1$ (cf. Notation \ref{mult-ind-not}), there is no canonical way to decompose
\begin{equation}\label{3tensorprod}
\bigoplus_{\sigma=1}^{\textswab{m}(\vb*p_1,\vb*p_2;\vb*a)} (\vb*a; \sigma) \  \subset \ \vb*p_1\otimes \vb*p_2
\end{equation}
into irreducible representations of class $\vb*a$. A systematic method for solving this is worked out by Chew \cite{chew1, chew2} and Pluhar \cite{pluh1}, in a way that the CG coefficients satisfy a large set of symmetries very similar to the ones of $SU(2)$ \cite{RS, varsh}. In Appendix \ref{sec:mixed-casimirs}, we provide a short presentation of their method, which is based on mixed Casimir operators and leads to Theorem \ref{cg-sym-theo} below.  

\begin{notation}\label{PST-count-not}
	The following involution in the set of multiplicity indices will be relevant:
	\begin{equation}\label{sigma-check}
		\{1,...,\textswab{m}(\vb*p_1, \vb*p_2;\vb*a)\}\ni\sigma\mapsto\widecheck \sigma = \textswab{m}(\vb*p_1, \vb*p_2;\vb*a)-\sigma + 1 \ .
	\end{equation}
\end{notation}

\begin{theorem}[\cite{chew1,chew2,pluh1}]\label{cg-sym-theo}
	Given any representations $\vb*p_1$, $\vb*p_2$ and $\vb*a$, there is a choice of orthogonal decomposition\footnote{That is to say, \eqref{ortho-cg} holds} for Clebsch-Gordan series such that the Clebsch-Gordan coefficients are all real and satisfy
	\begin{equation}\label{cg-sym}
		\begin{aligned}
			\cg{\vb*p_1}{\vb*p_2}{(\vb*a;\sigma)}{\vb*\nu_1 J_1}{\vb*\nu_2 J_2}{\,\,\vb*\nu J} & =  (-1)^{|\vb*p_1|+|\vb*p_2|+|\vb*a|}\cg{\vb*p_2}{\vb*p_1}{(\vb*a; \widecheck\sigma)}{\vb*\nu_2 J_2}{\vb*\nu_1 J_1}{\,\,\vb*\nu J}\\
			& = (-1)^{|\vb*p_1|-2(t_{\vb*\nu_1}+u_{\vb*\nu_1})}\sqrt{\dfrac{\dim (\vb*a)}{\dim (\vb*p_2)}}\,\cg{\vb*p_1}{\widecheck{\vb*a}}{(\widecheck{\vb*p}_2;\widecheck\sigma)}{\vb*\nu_1 J_1}{\widecheck{\vb*\nu} J}{\,\,\widecheck{\vb*\nu}_2 J_2}\\
			& = (-1)^{|\vb*p_1|+|\vb*p_2|+|\vb*a|}\cg{\widecheck{\vb*p}_1}{\widecheck{\vb*p}_2}{(\widecheck{\vb*a};\widecheck\sigma)}{\widecheck{\vb*\nu}_1 J_1}{\widecheck{\vb*\nu}_2 J_2}{\,\,\widecheck{\vb*\nu} J} \ .
		\end{aligned}
	\end{equation}
\end{theorem}

In particular, for the convention of Theorem \ref{cg-sym-theo}, the Hermitian conjugate $\dagger$ of operators in $\vb*p\otimes \widecheck{\vb*p}$ satisfies
\begin{equation}\label{hce}
	\vb*e^\dagger((\vb*a;\sigma);\vb*\nu, J) = (-1)^{2(t+u)}\vb*e((\widecheck{\vb*a};\sigma);\widecheck{\vb*\nu}, J)\, .
\end{equation}
And the adjoint is given by
\begin{equation}\label{adj-c-b}
	\begin{aligned}
		\ast & :  \vb*e((\vb*a;\sigma);\vb*\nu, J) \\
		& \hspace{3 em}\mapsto (-1)^{|\vb*a|}\sum_{\substack{\vb*\nu_1,J_1\\\vb*\nu_2,J_2}}\cg{\widecheck{\vb*p}}{\vb*p}{(\vb*a;\widecheck\sigma)}{\vb*\nu_2 J_2}{\vb*\nu_1 J_1}{\,\,\vb*\nu J} \widecheck{\vb*e}(\widecheck{\vb*p};\vb*\nu_2, J_2)\otimes \vb*e(\vb*p;\vb*\nu_1, J_1) \ .
	\end{aligned}
\end{equation}
In the light of Remark \ref{ind-dual-inv}, we identify
\begin{eqnarray}
	\widecheck{\vb*e}((\vb*a;\widecheck\sigma);\vb*\nu, J) &:=&  \sum_{\substack{\vb*\nu_1,J_1\\\vb*\nu_2,J_2}}\cg{\widecheck{\vb*p}}{\vb*p}{(\vb*a;\widecheck\sigma)}{\vb*\nu_2 J_2}{\vb*\nu_1 J_1}{\,\,\vb*\nu J} \widecheck{\vb*e}(\widecheck{\vb*p};\vb*\nu_2, J_2)\otimes \vb*e(\vb*p;\vb*\nu_1, J_1) \ , \label{adj-c-b-d} \\
	\quad\quad\quad {\vb*e}^\ast((\vb*a;\sigma);\vb*\nu, J) &=&  (-1)^{|\vb*a|}\widecheck{\vb*e}((\vb*a;\widecheck\sigma);\vb*\nu, J) \ . \label{astcoupled}
\end{eqnarray}

\begin{definition}
	Given a coupled basis $\{\vb*e((\vb*a;\sigma);\vb*\nu, J\}$ of $\vb*p\otimes \widecheck{\vb*p}$, the \emph{induced coupled basis} of $\widecheck{\vb*p}\otimes \vb*p$ is the basis $\{\widecheck{\vb*e}((\vb*a;\widecheck\sigma);\vb*\nu, J)\}$ satisfying (\ref{adj-c-b-d})-(\ref{astcoupled}).
\end{definition}

Even under this convention, there is still a free phase, that we do not fix in general, except for setting
\begin{equation}\label{id-gen-case}
	\vb*e((0,0);(0,0,0),0) = \dfrac{1}{\sqrt{\dim (\vb*p)}}\mathds{1}
\end{equation}
for reasons that will become clear later (see Theorems \ref{symb_corresp-ker} and \ref{op-ker-gen}).

For the operator product in uncoupled basis, we have
\begin{equation}\label{prod_unc}
	\begin{aligned}
		& \big(\vb*e(\vb*p;\vb*\nu,J) \otimes \widecheck{\vb*e}(\widecheck{\vb*p};\vb*\nu',J')\big) \big(\vb*e(\vb*p;\vb*\mu',I') \otimes \widecheck{\vb*e}(\widecheck{\vb*p};\vb*\mu,I)\big)\\
		& \hspace{5 em} = \delta_{J', I'}\Delta^{|\vb*p|}_{\vb*\nu', \vb*\mu'}(-1)^{2(t_{\vb*\nu'}+u_{\vb*\nu'})}\vb*e(\vb*p;\vb*\nu,J) \otimes \widecheck{\vb*e}(\widecheck{\vb*p};\vb*\mu,I) \, .
	\end{aligned}
\end{equation}

\begin{lemma}\label{op-prod-gen}
For representations $\vb*a_1$, $\vb*a_2$, $\vb*a$ and $\vb*p$, there exist real coefficients $M[\vb*p]^{(\vb*a_1;\sigma_1)(\vb*a_2;\sigma_2)}_{(\vb*a;\sigma,\sigma')}$ such that the operator product of elements of a coupled basis of $\vb*p\otimes \widecheck{\vb*p}$ can be decomposed as 
\begin{equation}\label{op-prod-gen-eq}
    \begin{aligned}
    	& \vb*e((\vb*a_1;\sigma_1); \vb*\nu_1, J_1)\vb*e((\vb*a_2;\sigma_2); \vb*\nu_2,J_2)=\\
    	& \hspace{3 em}\sum_{\substack{(\vb* a;\sigma,\sigma')\\\vb*\nu, J}} \cg{\vb*a_1}{\vb*a_2}{(\vb*a;\sigma')}{\vb*\nu_1 J_1}{\vb*\nu_2 J_2}{\vb*\nu J}M[\vb*p]^{(\vb*a_1;\sigma_1)(\vb*a_2;\sigma_2)}_{(\vb*a;\sigma,\sigma')}\vb*e((\vb*a;\sigma); \vb*\nu, J) \, ,
    \end{aligned}
\end{equation}
where $\nabla_{\vb*\nu_1+\vb*\nu_2,\vb*\nu}=\delta(J_1,J_2,J)=1$, and $\sigma$ and $\sigma'$ count the multiplicity of $\vb*a$ in $\vb*p\otimes \widecheck{\vb*p}$ and in $\vb*a_1\otimes \vb*a_2$, respectively.
\end{lemma}
\begin{proof}
The product of operators is an equivariant bilinear map from the representation $\vb*p\otimes \widecheck{\vb*p}$ to itself. The claim follows straightforwardly from Schur's Lemma.
\end{proof}

To write the product of operators in a more explicit way, we use the following:

\begin{definition}\label{Wigner-product}
	The \emph{Wigner product symbol} is the coefficient denoted by the square brackets below:
	\begin{equation}
		\begin{aligned}
			&\begin{bmatrix}
				(\vb*a_1;\sigma_1) & (\vb*a_2;\sigma_2) & (\vb*a;\sigma)\\
				\vb*\nu_1, J_1 & \vb*\nu_2, J_2 & \vb*\nu, J
			\end{bmatrix}\!\![\vb*p] \\
			& \hspace{7 em} = \sum_{\substack{\vb*\mu_1,\vb*\mu_2,\vb*\mu_3\\I_1,I_2,I_3}}(-1)^{|\vb*p|+2(t_{\vb*\nu}+t_{\vb*\mu_2}+u_{\vb*\nu}+u_{\vb*\mu_2})}\cg{\vb*p}{\widecheck{\vb*p}}{(\vb*a_1;\sigma_1)}{\vb*\mu_1 I_1}{\vb*\mu_2 I_2}{\,\,\vb*\nu_1 J_1} \\
			& \hspace{12 em}\times\cg{\vb*p}{\widecheck{\vb*p}}{(\vb*a_2;\sigma_2)}{\widecheck{\vb*\mu}_2 I_2}{\vb*\mu_3 I_3}{\,\,\vb*\nu_2 J_2}\cg{\vb*p}{\widecheck{\vb*p}}{(\vb*a;\sigma)}{\vb*\mu_1 I_1}{\vb*\mu_3 I_3}{\,\,\vb*\nu J}\, .
		\end{aligned}
	\end{equation}
\end{definition}

\begin{theorem}\label{op-prod-gen-s}
	The operator product of elements of a coupled basis of $\vb*p\otimes \widecheck{\vb*p}$ can be decomposed as 
	\begin{equation}\label{opwp}
		\begin{aligned}
			&\vb*e((\vb* a_1;\sigma_1); \vb*\nu_1,\, J_1)\vb*e((\vb* a_2;\sigma_2); \vb*\nu_2, J_2) = \\ &\hspace{3 em}\sum_{\substack{(\vb* a;\sigma)\\\vb*\nu, J}}(-1)^{|\vb*p|+2(t_{\vb*\nu}+u_{\vb*\nu})}\!\! \begin{bmatrix}
				(\vb*a_1;\sigma_1) & (\vb*a_2;\sigma_2) & (\vb*a;\sigma)\\
				\vb*\nu_1, J_1 & \vb*\nu_2, J_2 & \widecheck{\vb*\nu}, J
			\end{bmatrix}\!\![\vb*p]\vb*e((\vb* a;\sigma); \vb*\nu, J)
		\end{aligned}
	\end{equation}
	where summation over $(\vb* a;\sigma)$ is restricted to the Clebsch-Gordan series of $\vb*p\otimes \widecheck{\vb*p}$, and summations over $\vb*\nu$ and $J$ effectively restricted by $\nabla_{\vb*\nu_1+\vb*\nu_2,\vb*\nu}=\delta(J_1,J_2,J)=1$.
\end{theorem}
\begin{proof}
	Using (\ref{coup-unc}), we write
	\begin{equation}
		\begin{aligned}
			&\vb*e((\vb*a_1;\sigma_1); \vb*\nu_1, J_1) = \sum_{\substack{\vb*\mu_1, I_1 \\ \vb*\mu_2, I_2}}\cg{\vb*p}{\widecheck{\vb*p}}{(\vb*a_1;\sigma_1)}{\vb*\mu_1 I_1}{\vb*\mu_2 I_2}{\,\,\vb*\nu_1 J_1}\vb*e(\vb*p;\vb*\mu_1,I_1)\otimes \widecheck{\vb* e}(\widecheck{\vb*p};\vb*\mu_2,I_2) \, , \\
			&\vb*e((\vb* a_2;\sigma_2); \vb*\nu_2, J_2) = \sum_{\substack{\vb*\mu_3, I_3\\ \vb*\mu_4, I_4}}\cg{\vb*p}{\widecheck{\vb*p}}{(\vb*a_2;\sigma_2)}{\vb*\mu_4 I_4}{\vb*\mu_3 I_3}{\,\,\vb*\nu_2 J_2}\vb*e(\vb*p;\vb*\mu_4,I_4)\otimes \widecheck{\vb* e}(\widecheck{\vb*p};\vb*\mu_3,I_3) \ .
		\end{aligned}     
	\end{equation}
	From (\ref{prod_unc}), we have
	\begin{equation}
		\begin{aligned}
			&\vb*e((\vb*a_1;\sigma_1); \vb*\nu_1, J_1)\vb*e((\vb*a_2;\sigma_2); \vb*\nu_2,J_2) \ =  \\
			&\hspace{5 em}\sum_{\substack{\vb*\mu_1,\vb*\mu_2,\vb*\mu_3\\I_1,I_2,I_3}}\!\!(-1)^{2(t_{\vb*\mu_2}+u_{\vb*\mu_2})}\cg{\vb*p}{\widecheck{\vb*p}}{(\vb*a_1;\sigma_1)}{\vb*\mu_1 I_1}{\vb*\mu_2 I_2}{\,\,\vb*\nu_1 J_1} \cg{\vb*p}{\widecheck{\vb*p}}{(\vb*a_2;\sigma_2)}{\widecheck{\vb*\mu}_2 I_2}{\vb*\mu_3 I_3}{\,\,\vb*\nu_2 J_2}\\
			& \hspace{10 em}\times\vb*e(\vb*p;\vb*\mu_1,I_1)\otimes \widecheck{\vb* e}(\widecheck{\vb*p};\vb*\mu_3,I_3)  
		\end{aligned}
	\end{equation}
	then, using (\ref{unc-coup}) and (\ref{cg-ne-0}),
	\begin{equation}
		\begin{aligned}
			\vb*e((\vb*a_1;\sigma_1); \vb*\nu_1, J_1)\vb*e((\vb*a_2;\sigma_2); \vb*\nu_2,J_2) \ = &\\ \sum_{\substack{(\vb*a;\sigma) \\\vb*\nu, J}}\sum_{\substack{\vb*\mu_1,\vb*\mu_2,\vb*\mu_3\\I_1,I_2,I_3}}(-1)^{2(t_{\vb*\mu_2}+u_{\vb*\mu_2})}&\cg{\vb*p}{\widecheck{\vb*p}}{(\vb*a_1;\sigma_1)}{\vb*\mu_1 I_1}{\vb*\mu_2 I_2}{\,\,\vb*\nu_1 J_1} \cg{\vb*p}{\widecheck{\vb*p}}{(\vb*a_2;\sigma_2)}{\widecheck{\vb*\mu}_2 I_2}{\vb*\mu_3 I_3}{\,\,\vb*\nu_2 J_2} \\ & \times\cg{\vb*p}{\widecheck{\vb*p}}{(\vb*a;\sigma)}{\vb*\mu_1 I_1}{\vb*\mu_3 I_3}{\,\,\vb*\nu J}\,\vb*e((\vb*a;\sigma); \vb*\nu, J) \,.
		\end{aligned}
	\end{equation}
	The restriction on the sum follows from Lemma \ref{op-prod-gen}.
\end{proof}

In Appendix \ref{sec:wig_symb}, we derive another expression for Wigner product symbols in terms of \emph{Wigner (re)coupling symbols} by means of the \emph{Wigner identity} under the validity of the symmetries of Theorem \ref{cg-sym-theo}.

\subsection{(Co)Adjoint orbits as invariant phase spaces}\label{co-sec}

$SU(3)$ being a simple compact Lie group, the coadjoint and adjoint orbits of $SU(3)$ are isomorphic,\footnote{A general discussion of coadjoint orbits of semisimple Lie groups can be found in \cite{bernatska}.} so here we focus on the adjoint action of $SU(3)$ on its Lie algebra, which provides a real representation whose complexification is the irrep $(1,1)$. We identify the root diagram of $\mathfrak{su}(3)$ with the Cartan subalgebra generated by $iT_3$ and $iU_3$ by making $\alpha_1\equiv 2i\, T_3$ and $\alpha_2 \equiv 2i\, U_3$. Then, we obtain
\begin{equation}\label{weights12}
    \varpi_1 \equiv \dfrac{i}{2}\,\lambda_3+\dfrac{i}{2\sqrt{3}}\,\lambda_8 = \dfrac{i}{3}\small{\begin{pmatrix}
    2 & 0 & 0\\
    0 & -1 & 0\\
    0 & 0 & -1
    \end{pmatrix}} \ , \ \ \varpi_2 = \dfrac{i}{\sqrt{3}}\,\lambda_8 = \dfrac{i}{3}\small{\begin{pmatrix}
    1 & 0 & 0\\
    0 & 1 & 0\\
    0 & 0 & -2
    \end{pmatrix}}\, ,
\end{equation}
so that
\begin{equation}
	\xi_{(x,y)} = \sqrt{\dfrac{3}{2}}\,(x\, \varpi_1+y\,\varpi_2)
\end{equation}
has squared norm 
\begin{equation}\label{normxi}
    \norm{\xi_{(x,y)}}^2=x^2+xy+y^2 \ . 
\end{equation}

It is well known that each orbit intersects the closed positive Weyl chamber 
\begin{equation}
\overline C = \{\xi_{(x,y)}: x,y\ge 0\}
\end{equation} 
in precisely one single point (see e.g.~ \cite{bott}). Let $\mathcal O_{(x,y)}$ be the orbit passing through the point $\xi_{(x,y)}\in \overline C\setminus \{0\}$.

It is clear from (\ref{weights12}) that, for $x>0$,  the isotropy subgroup of $\xi_{(x,0)}$ is
\begin{equation}\label{u2-ut}
    H := \left\{ \small{\begin{pmatrix}
    \det(U)^{-1} & 0 \\
    0 & U
    \end{pmatrix}}\ , \ \ U\in U(2) \right\} \simeq S\big(U(1)\times U(2)\big)\simeq U(2) \, ,
\end{equation}
whereas, for $y>0$, the isotropy subgroup of $\xi_{(0,y)}$ is
\begin{equation}
	\begin{aligned}
		\widecheck{H}  &:= \widecheck\delta H\widecheck\delta^{-1} = \widecheck\delta H \widecheck\delta = \left\{\small{\begin{pmatrix}
				U & 0 \\
				0 & \det(U)^{-1}
		\end{pmatrix}} \ , \ \ U\in U(2) \right\} \label{u2-u} \\ 
		&\hspace{2 em}\simeq S\big(U(2)\times U(1)\big) \simeq U(2) \ , \ \ \mbox{where}
	\end{aligned}
\end{equation}
\begin{equation}\label{g0}
	\widecheck\delta = {\small\begin{pmatrix}
			0 & 0 & -1\\
			0 & -1 & 0\\
			-1 & 0 & 0
	\end{pmatrix}}\in SU(3) \, .
\end{equation}
On the other hand, the isotropy subgroup of $\xi_{(x,y)}$ is the maximal torus\footnote{It is a matter of simple calculation to verify that $\widecheck{T} = \widecheck\delta T\widecheck\delta = T$.}
\begin{equation}\label{torus}
\begin{aligned}
    T := & \left\{\small{\begin{pmatrix}
        e^{i\theta_1} & 0 & 0\\
        0 & e^{i\theta_2} & 0\\
        0 & 0 & e^{i\theta_3}
    \end{pmatrix}}: \theta_1+\theta_2+\theta_3 = 0\right\}\\
    \simeq \ & S\big(U(1)\times U(1)\times U(1)\big) \simeq U(1)\times U(1) \, .
\end{aligned}
\end{equation}
Therefore, we have two types of non trivial (co)adjoint orbits:
\begin{equation}
\begin{aligned}
    \mathcal O_{(x,0)}&\simeq SU(3)/H\simeq SU(3)/\widecheck H\simeq \mathcal O_{(0,y)} \ , \\
   \mathcal O_{(x,y)}&\simeq SU(3)/T\simeq
    \mathcal O_{(y,x)} \ , \ \ \mbox{for} \  x,y>0 \ . 
\end{aligned}    
\end{equation}

For a better realization of such orbits, we recall the complex projective space $\mathbb CP^2$, the quotient of $\mathbb C^3\backslash\{0\}$ by the equivalence relation
\begin{equation}
      z\sim z'\iff z = a\, z' \ , \ \ a\in \mathbb C^\times \ .
\end{equation}
To construct $\mathbb CP^2$ using this equivalence relation,\footnote{This construction is presented in \cite{alex}, for instance.} we can look only to the unitary vectors of $\mathbb C^3$, reducing our analysis to the $SU(3)$-homogeneous space $\mathcal S^5 = \{z\in \mathbb C^3: ||z|| = 1\}$. Since the point $(1,0,0)\in \mathcal S^5$ has the $u$-standard $SU(2)$
\begin{equation}\label{su2}
    \left\{\small{\begin{pmatrix}
    1 & 0 \\
    0 & U 
    \end{pmatrix}}: \ U \in SU(2)\right\} \subset SU(3)
\end{equation}
as isotropy subgroup, we have\footnote{$SU(2)\simeq \mathcal S^3$ hence $SU(3)$ is a $3$-sphere bundle over $\mathcal S^5$.} $\mathcal S^5 \simeq SU(3)/SU(2)$. Also note that, $\forall z \in \mathcal S^5$, $e^{i\theta}z \sim z$. So $\mathbb CP^2 \simeq \mathcal S^5/\mathcal S^1$ and the isotropy subgroup of the equivalence class $[1:0:0]\in \mathbb CP^2$ is $H$, i.e., $SU(3)/H \simeq \mathbb CP^2$. By similar argument we get $U(2)/(U(1)\times U(1))\simeq \mathbb CP^1$, so $SU(3)/T\simeq \mathcal E$, where the flag manifold $\mathcal E$ is the total space of a fiber bundle $\mathbb CP^1\hookrightarrow\mathcal E\to \mathbb CP^2$. Thus,
\begin{equation}\label{orbits-iso1}
    \mathcal O_{(x,y)} \simeq \begin{cases}
        \mathbb CP^2 \ , \ \ \mbox{if} \ \ xy= 0\\
        \mathcal E \ , \ \ \mbox{if} \ \ xy\ne 0
    \end{cases} .
\end{equation}

The orbits $\mathcal O_{(x,y)}$ and $\mathcal O_{(y,x)}$ are related by the involution $\iota = -id$ on $\mathfrak{su}(3)$. Indeed, $\iota\circ Ad_g = Ad_g\circ \iota$ trivially holds for every $g\in SU(3)$ and
\begin{equation}\label{iota-def}
     \iota(x\,\varpi_1+y\,\varpi_2) = -x\,\varpi_1-y\,\varpi_2 = Ad_{\widecheck\delta}(y\,\varpi_1+x\,\varpi_2) \ ,
\end{equation}
so $\iota(\mathcal O_{(x,y)}) = \mathcal O_{(y,x)}$. Thus, $\iota$ is an involution on $\mathcal O_{(x,x)}$.

Let $\vb x_0 = [1:0:0] \in \mathbb CP^2$, whose isotropy subgroup is $H$, so that the isotropy subgroup of $\widecheck\delta\vb x_0 = [0:0:1]$ is $\widecheck{H}$, and let $\vb z_0\in \pi^{-1}(\vb x_0)\subset \mathcal E$ be a point with $T$ as isotropy subgroup. Then 
consider the equivariant diffeomorphisms 
\begin{equation}\label{psixy}
	\begin{aligned}
		& \ \, \psi_{(x,0)}:\mathcal O_{(x,0)}\to \mathbb CP^2 : Ad_g \xi_{(x,0)}\mapsto g\,\vb x_0 \ , \\
		& \psi_{(0,y)}:\mathcal O_{(0,y)}\to \mathbb CP^2: Ad_g \xi_{(0,y)}\mapsto g\widecheck\delta\,\vb x_0 \ , \ \ \\
		& \ \ \ \psi_{(x,y)}: \mathcal O_{(x,y)} \to \mathcal E : Ad_g \xi_{(x,y)}\mapsto g\, \vb z_0 \ ,
	\end{aligned}
\end{equation}
for $x,y>0$ still holding, and
\begin{equation}\label{checkpsixy}
	\widecheck\psi_{(x,y)} = \psi_{(x,y)}\circ \iota\, .
\end{equation}
Therefore,
\begin{equation}\label{id-x0}
	\alpha_{{}_{\mathbb CP^2}}:=\psi_{(x,0)}\circ \widecheck\psi_{(0,x)}^{-1} = \psi_{(0,y)}\circ \check\psi_{(y,0)}^{-1}:\mathbb CP^2\to \mathbb CP^2
\end{equation}
is the identity map, and
\begin{equation}\label{iso-comp-orbit1}
	\alpha_{{}_{\mathcal E}} := \psi_{(x,y)} \circ \widecheck\psi_{(y,x)}^{-1}:\mathcal E\to \mathcal E
\end{equation}
is an $SU(3)$-equivariant involution. Of course, there is an equivalent involution on each $\mathcal O_{(x,y)}\simeq \mathcal E$, namely 
\begin{equation}\label{iso-comp-orbit2}
	\alpha_{(x,y)}:=\psi_{(x,y)}^{-1}\circ \widecheck\psi_{(y,x)} : \mathcal O_{(x,y)}\to \mathcal O_{(x,y)} \ ,  
\end{equation}
which reduces to $\iota$ for $x = y$.

We also have a $SU(3)$-invariant symplectic form on every (co)adjoint orbit: each $\mathcal O_{(x,y)}$ is a symplectic leaf of $\mathfrak  g =\mathfrak{su}(3)$ equipped with a Kirillov-Arnold-Kostant-Souriau (KAKS) Poisson bivector field\footnote{This Poisson structure was actually first identified for some Lie groups by Sophus Lie himself \cite{Wein}.} 
\begin{equation}\label{pois_bivec}
	\Pi_{\mathfrak g} =  \Pi_{\mathfrak{su}(3)}=\sum_{j,k,l=1}^8c_{kl}^lx_l\partial_j\otimes \partial_k\, ,
\end{equation}
where $c^l_{jk}$ are structure constants of $\mathfrak{su}(3)$ with respect to some basis, and $(x_1,...,x_8)$ are the coordinates with respect to the same basis (cf. \cite{kir} for more details, including a natural definition for such symplectic forms). Furthermore, the $SU(3)$-invariant symplectic form on $\mathcal O_{(x,y)}$ induces a normalized left invariant integral on the orbit $\mathcal O_{(x,y)}$ such that any other left invariant integral differs from it by a scalar factor. Then, we can fix this factor for $\mathcal O_{(x,y)}$ so that the lift $\tilde f\in \mathcal C(SU(3))$ of any $f\in \mathcal C(\mathcal O_{(x,y)})$ satisfies
\begin{equation}\label{haar-leftinvint}
    \int_{SU(3)}\tilde f(g)\,dg = \int_{\mathcal O_{(x,y)}} f(\vb*\varsigma)\,d\vb*\varsigma
\end{equation}
for the Haar integral on $SU(3)$ (cf. eg. \cite{foll}). With no danger of causing confusion, we may denote the $SU(3)$-invariant inner product in $L^2(\mathcal O_{(x,y)})$ with respect to such integral simply by $\ip{\,}{\,}$, that is, for any $f_1,f_2\in L^2(\mathcal O_{(x,y)})$,
\begin{equation}\label{ip-orb}
    \ip{f_1}{f_2} = \int_{\mathcal O_{(x,y)}}\overline{f_1}(\vb*\varsigma)f_2(\vb*\varsigma)\,d\vb*\varsigma \, .
\end{equation}

Throughout this paper, we consider $\mathbb CP^2$ and $\mathcal E$ as homogeneous spaces (by the adjoint action of $SU(3)$) equipped with the the aforementioned symplectic forms and normalized left invariant integrals, so that $\mathbb CP^2$ and $\mathcal E$ are $SU(3)$-invariant symplectic phase spaces for classical quark systems. Therefore, in what follows, we shall identify as a classical quark system, the Poisson algebra of smooth functions on $\mathcal O$, for $\mathcal O$ isomorphic to either $\mathbb CP^2$ or $\mathcal E$. When $\mathcal O\simeq\mathbb CP^2$, we shall refer to a \emph{classical pure-quark system}. When $\mathcal O\simeq\mathcal E$, we shall refer to a \emph{classical mixed-quark system}. Later, in Paper II, we shall also consider the Poisson algebra of smooth functions on the unitary sphere $\mathcal S^7\subset \mathfrak{su}(3)$,  as the \emph{classic total-quark system}.

We finish this section by establishing the following relevant result, which leads to the decomposition of the product of harmonic functions, cf. Theorems \ref{harm-point-prod} and \ref{e-harm-point-prod}, and can be used in asymptotic analysis of twisted product, as depicted in Paper II.

\begin{lemma}\label{wig-prod}
The pointwise product of Wigner $D$-functions can be decomposed into a sum of the form
\begin{equation}
    D^{\vb*p_1}_{\vb* \nu_1 J_1, \vb*\mu_1 I_1}D^{\vb*p_2}_{\vb* \nu_2 J_2, \vb*\mu_2 I_2} = \sum_{\substack{(\vb*a;\sigma)}}\sum_{\substack{\vb*\nu,J\\\vb*\mu, I}} \cg{\vb*p_1}{\vb*p_2}{(\vb*a;\sigma)}{\vb*\nu_1 J_1}{\vb*\nu_2 J_2}{\,\,\vb*\nu J}\cg{\vb*p_1}{\vb*p_2}{(\vb*a;\sigma)}{\vb*\mu_1 I_1}{\vb*\mu_2 I_2}{\,\,\vb*\mu I}D^{\vb*a}_{\vb*\nu J, \vb*\mu I} 
\end{equation}
where the summations are restricted to $\nabla_{\vb*\nu_1+\vb*\nu_2, \vb*\nu} = \nabla_{\vb*\mu_1+\vb*\mu_2,\vb*\mu} = 1$, $\delta(J_1,J_2,J)=\delta(I_1,I_2,I) = 1$ and $(\vb*a;\sigma)$ in the CG series of $\vb*p_1\otimes \vb*p_2$.
\end{lemma}
\begin{proof} 
Let $g\in SU(3)$. From (\ref{unc-coup}) and (\ref{cg-ne-0}), we have
\begin{equation}
\begin{aligned}
    \sum_{\substack{\vb*\nu_1, J_1\\ \vb*\nu_2, I_2}}& D^{\vb*p_1}_{\vb* \nu_1 J_1, \vb*\mu_1, I_1}(g)D^{\vb*p_2}_{\vb* \nu_2 J_2, \vb*\mu_2, I_2}(g)\,\vb*e(\vb*p_1; \vb*\nu_1, J_1)\otimes\vb*e(\vb*p_2; \vb*\nu_2, J_2) = \\ &\sum_{\substack{(\vb*a;\sigma)\\ \vb*\mu, I}}\cg{\vb*p_1}{\vb*p_2}{(\vb*a;\sigma)}{\vb*\mu_1I_1}{\vb*\mu_2 I_2}{\,\,\vb*\mu I}\sum_{\vb* \nu, J}D^{\vb*a}_{\vb*\nu J, \vb*\mu I}(g)\vb*e((\vb*a;\sigma), \vb*\nu, J) \ ,
\end{aligned}
\end{equation}
where the sum over $(\vb*a;\sigma)$, $\vb*\mu$ and $I$ satisfies the statement. From (\ref{coup-unc}) and (\ref{cg-ne-0}), 
\begin{equation}
\begin{aligned}
    &\sum_{\substack{\vb*\nu_1, J_1\\ \vb*\nu_2, J_2}} D^{\vb*p_1}_{\vb* \nu_1 J_1, \vb* \mu_1 I_1}(g)D^{\vb*p_2}_{\vb* \nu_2 J_2, \vb* \mu_2 I_2}(g)\vb*e(\vb*p_1; \vb*\nu_1, J_1)\otimes\vb*e(\vb*p_2; \vb*\nu_2, J_2) \ = \\
    &\hspace{5 em}\sum_{\substack{\vb*\nu_1, J_1\\ \vb*\nu_2,J_2}}\sum_{\substack{(\vb*a;\sigma)}}\sum_{\substack{\vb*\nu, J\\\vb*\mu, I}}\!\cg{\vb*p_1}{\vb*p_2}{(\vb*a;\sigma)}{\vb*\mu_1 I_1}{\vb*\mu_2 I_2}{\,\,\vb*\mu I}\cg{\vb*p_1}{\vb*p_2}{(\vb*a;\sigma)}{\vb*\nu_1 J_1}{\vb*\nu_2 J_2}{\,\,\vb*\nu J}D^{\vb*a}_{\vb*\nu J, \vb*\nu J}(g)\\
    &\hspace{7 em} \times \vb*e(\vb*p_1; \vb*\nu_1, J_1)\otimes \vb*e(\vb*p_2; \vb*\nu_2, J_2) \, ,
\end{aligned}
\end{equation}
where $\vb*\nu$, $\vb*\nu_1$, $\vb*\nu_2$, $J$, $J_1$ and $J_2$ are related as in the statement. The decomposition in a basis is unique, so this finishes the proof.
\end{proof}

\section{On symbol correspondences and twisted products}\label{sec:symb_c}

Let $\mathcal O$ be either $\mathbb CP^2$ or $\mathcal E$. The basic object of interest to us is the following.

\begin{definition}\label{def:symb_corresp}
Given a non trivial irrep $\vb*p$ of $SU(3)$ on $\mathcal H_{\vb*p}$, a \emph{symbol correspondence for $\vb*p$ over $\mathcal O$}, referred to simply as a \emph{symbol correspondence} or just a \emph{correspondence}, is an injective linear map $$W:\mathcal B(\mathcal H_{\vb*p})\to C_{\mathbb C}^\infty(\mathcal O)\ , \ \ A\mapsto W_{\!A} \ , \ \ \ \mbox{satisfying}$$ 
\begin{enumerate}
	\item Equivariance: $\forall g \in SU(3)$, $W_{\!A^g} = (W_{\!A})^g$;
	\item Reality: $W_{\!A^\dagger} = \overline{W_{\!A}}$ ;
	\item Normalization: $\int_{\mathcal O}W_{\!A}(\vb*\varsigma)d\vb*\varsigma = \dfrac{1}{\dim(\vb*p)}\tr(A)$ .
\end{enumerate}
The function $W_{\!A}$ is called the \emph{symbol} of the operator $A$ by $W$.
\end{definition}

\begin{notation}\label{ImW}
    We denote by $\mathcal S_{\vb*p}(W)$ the image of a correspondence $W:\mathcal B(\mathcal H_{\vb*p})\to C_{\mathbb C}^\infty(\mathcal O)$, so that $W:\mathcal B(\mathcal H_{\vb*p})\to \mathcal S_{\vb*p}(W)$ is a bijection. Although not the case for pure-quark systems, for mixed-quark systems in general we may have two symbol correspondences $W_1,W_2:\mathcal B(\mathcal H_{\vb*p})\to C_{\mathbb C}^\infty(\mathcal E)$ with different images, cf. Lemma \ref{WneqW} further below.
\end{notation}

\begin{remark}\label{rmk:sc_conc_orbit}
	$\mathcal O$ can be substituted in the above definition of symbol correspondence by a concrete orbit $\mathcal O_{(x,y)}\subset\mathfrak{su}(3)$ by means of the isomorphism $\psi_{(x,y)}$, where $\mathcal O_{(x,y)}\simeq\mathcal E$ for $xy\neq 0$, and $\mathcal O_{(x,y)}\simeq\mathbb CP^2$, for $xy=0$ and $x+y\neq 0$.   
\end{remark}

\begin{remark}
	In the next sections we shall characterize and classify all symbol correspondences, as defined above, but these depend on which orbit $\mathcal O$ is considered. In this section, we just highlight some general properties that hold for both orbits. Indeed, although we restrict the context to $SU(3)$ for textual coherence, an attentive reader may note that the material in this section holds for every compact Lie group, with obvious adaptations.
\end{remark}

Let $G_0$ be either $H$ or $T$ so that $\mathcal O \simeq SU(3)/G_0$, and let $\vb*\varsigma_0 \in \mathcal O$ be either $\vb x_0$ or $\vb z_0$, respectively, so that its isotropy subgroup is $G_0$. Given $K\in \mathcal B(\mathcal H_{\vb*p})$ fixed by $G_0$, we have a smooth operator-valued function
\begin{equation}\label{op-val_f}
	\mathcal O\to \mathcal B(\mathcal H_{\vb*p}):\vb*\varsigma=g\vb*\varsigma_0\mapsto K(\vb*\varsigma)= K(g\vb*\varsigma_0)=K^g\, .
\end{equation}

\begin{proposition}\label{prop:W_op_k}
If $W:\mathcal B(\mathcal H_{\vb*p})\to C_{\mathbb C}^\infty(\mathcal O)$ is a symbol correspondence, then there is a unique Hermitian operator with unitary trace $K\in\mathcal B(\mathcal H_{\vb*p})$ fixed by $G_0$ such that
\begin{equation}\label{op_ker_general}
	W_{\!A}(\vb*\varsigma) = \tr(AK(\vb*\varsigma)) \ \ \ \forall \vb* \varsigma\in \mathcal O \, .
\end{equation}
\end{proposition}
\begin{proof}
The map $A\mapsto W_{\!A}(\vb*\varsigma_0)$ is a linear functional, so there exists $K\in \mathcal B(\mathcal H_{\vb*p})$ such that $W_{\!A}(\vb*\varsigma_0) = \tr(AK)$. Since
\begin{equation}\label{tr_equiv}
	\tr(A^gK) = \tr(AK^{g^{-1}}) \ , 
\end{equation}
$\forall g \in SU(3)$, and $\vb*\varsigma_0$ is fixed by $G_0$, we have that $K$ is fixed by $G_0$. In addition, 
\begin{equation}\label{tr_real}
	\overline{\tr(AK^g)} = \tr(A^\dagger(K^\dagger)^g)\, ,
\end{equation}
so reality condition implies $K$ is a Hermitian operator. To finish, from equivariance, we get that $W_{\mathds 1}$ is a constant function, and normalization condition gives that this constant is $1$, providing $\tr(K)=1$.
\end{proof}

\begin{definition}\label{def:op_k}
	An operator $K\in \mathcal B(\mathcal H_{\vb*p})$ for which \eqref{op_ker_general} gives a symbol correspondence is the \emph{operator kernel} of the correspondence. The \emph{moduli space} of symbol correspondences for $\vb*p$ over $\mathcal O$ is the set $\mathcal M_{\vb*p}(\mathcal O)\subset \mathcal B(\mathcal H_{\vb*p})$ such that $K\in \mathcal M_{\vb*p}(\mathcal O)$ if and only if $K$ is an operator kernel of a correspondence for a symbol correspondence $\vb*p$ over $\mathcal O$.
\end{definition}

\begin{proposition}
    If non empty, $\mathcal M_{\vb*p}(\mathcal O)$ is a non compact embedded submanifold of $\mathcal B(\mathcal H_{\vb*p})$.
\end{proposition}
\begin{proof}
    Let $\mathcal R_{\vb*p}\subset\mathcal B(\mathcal H_{\vb*p})$ be the subset of hermitian matrices, which is a real vector space invariant by the action of $SU(3)$, and denote by $\mathcal R_{\vb*p}^{G_0}$ the subspace of $\mathcal R_{\vb*p}$ of all elements fixed by $G_0$. Then, 
    \begin{equation}
        \mathcal V = \mathcal R_{\vb*p}^{G_0}\cap \ker(\tr)
    \end{equation}
    is a real vector space that contains $K'-K$ for every $K,K'\in \mathcal M_{\vb*p}(\mathcal O)$. We will prove that, for any $K\in \mathcal M_{\vb*p}(\mathcal O)$, the intersection of some open ball in $\mathcal B(\mathcal H_{\vb*p})$ centered at $K$ with $\mathcal M_{\vb*p}(\mathcal O)$ is an open ball in the affine space $K+\mathcal V$. 
    
    Indeed, if $K\in \mathcal M_{\vb*p}(\mathcal O)$ and $X \in \mathcal V$, the only obstacle for $K+X$ to be an operator kernel is the injectivity hypothesis. Then, let $W:\mathcal B(\mathcal H_{\vb*p})\to C_{\mathbb C}^\infty(\mathcal O)$ be the symbol correspondence with operator kernel $K$, and take
    \begin{equation}
        \epsilon = \min\{\norm{W_{\!A}}: \norm{A}=1 \} > 0 \ ,
    \end{equation}
    with strict inequality holding because $\mathcal B(\mathcal H_{\vb*p})$ is finite dimensional and $W$ is injective.
    
    Now, note that since $X\in \mathcal V$ is fixed by $G_0$, we can define the linear map
    \begin{equation}
    \begin{aligned}
        F&:\mathcal B(\mathcal H_{\vb*p})\to C_{\mathbb C}^\infty(\mathcal O): A\mapsto F_A\, , \\
        & \hspace{2 em} F_A(\vb*\varsigma) = \tr(AX(\vb*\varsigma))\, .
    \end{aligned}
    \end{equation}
    By the Cauchy-Schwarz inequality, $|F_A(\vb*\varsigma)|\le \norm{A}\norm{X(\vb*\varsigma)} = \norm{A}\norm{X}$. Since the $L^2$-norm of functions on $\mathcal O$ is bounded by the supremum norm, we get that
    \begin{equation}
        \norm{F_A}\le \norm{A}\norm{X}
    \end{equation}
    for every $A \in \mathcal B(\mathcal H_{\vb*p})$. In particular, if $A\in \mathcal B(\mathcal H_{\vb*p})$ satisfies $\norm{A} = 1$, then
    \begin{equation}
        \norm{W_{\!A}+F_A}\ge \norm{W_A}-\norm{F_A} \ge \epsilon-\norm{X}\, .
    \end{equation}
    Therefore, for every $X\in \mathcal V$ with $\norm{X}< \epsilon$, \ $K+X$ is an operator kernel, that is, $K+X\in \mathcal M_{\vb*p}(\mathcal O)$. Thus,  the open ball in $\mathcal B(\mathcal H_{\vb*p})$ with radius $\epsilon$ centered at $K$ intersects $\mathcal M_{\vb*p}(\mathcal O)$  precisely in the open ball of radius $\epsilon$ centered at $K$ in $K+\mathcal V$.

To finish, we need to prove that $\mathcal M_{\vb*p}(P)$ is non compact. For any $K \in \mathcal M_{\vb*p}(P)$, the operator $K-\mathds 1$ is its projection on the subspace orthogonal to $\mathds 1$. Hence, for any real $\lambda \in \mathbb R$, it is straightforward to verify that $\mathds 1+\lambda(K-\mathds 1)$ gives an operator kernel, so $\mathcal M_{\vb*p}(P)$ is unbounded in $\mathcal B(\mathcal H_{\vb*p})$.
\end{proof}

\subsection{Types of symbol correspondences}
Here we define some types of correspondences and explore some of their properties -- specific constructions are given in sections \ref{sec:pq-sys} and \ref{sec:gen-sys}. To begin with, consider the $\vb*p$-normalized inner product
\begin{equation}\label{normalized_ip}
	\ip{A_1}{A_2}_{\vb*p} = \dfrac{1}{\dim(\vb*p)}\tr(A_1^\dagger A_2) \ , 
\end{equation}
$\forall A_1, A_2 \in \mathcal B(\mathcal H_{\vb*p})$, which induces the inner-product norm $\norm{\,}_{\vb*p}$. The normalization condition in Definition \ref{def:symb_corresp} raises the question of whether it is possible to have a symbol correspondence that is an isometry with respect to this inner product on $\mathcal B(\mathcal H_{\vb*p})$ and the inner product (\ref{ip-orb}) on $C_{\mathbb C}^\infty(\mathcal O)$.

\begin{definition}
	Symbol correspondences $W,\widetilde W:\mathcal B(\mathcal H_{\vb*p})\to C_{\mathbb C}^\infty(\mathcal O)$ are said to be \emph{dual} to each other, just as their operators kernel, if
	\begin{equation}\label{dualW-eq}
		\ip{\widetilde W_{A_1}}{W_{A_2}} = \ip{A_1}{A_2}_{\vb*p} = \ip{W_{A_1}}{\widetilde W_{A_2}}
	\end{equation}
	for every $A_1, A_2 \in \mathcal B(\mathcal H_{\vb*p})$. A symbol correspondence is \emph{Stratonovich-Weyl} if it is self-dual, or equivalently if it is an isometry with respect to the inner product \eqref{normalized_ip} on $\mathcal B(\mathcal H_{\vb*p})$ and the inner product (\ref{ip-orb}) on $C_{\mathbb C}^\infty(\mathcal O)$.
\end{definition}

As we shall see in greater detail later, dual correspondences are not unique in general due to possible existence of irreps within $C_{\mathbb C}^\infty(\mathcal O)$ with higher multiplicity than within $\mathcal B(\mathcal H_{\vb*p})$, when $\mathcal O\simeq \mathcal E$. The following proposition can be used to identify a more canonical choice of duality. For that, we consider the map
\begin{equation}\label{Spmap}
\begin{aligned} 
	C_{\mathbb C}^\infty(\mathcal O)\supset\mathcal S_{\vb*p}(W)  &\to \mathcal B(\mathcal H_{\vb*p}) \, , \\
    f&\mapsto A_f = \dim(\vb*p)\int_{\mathcal O}f(\vb*\varsigma)K(\vb*\varsigma)d\vb*\varsigma\, ,
\end{aligned}     
\end{equation}
where $\mathcal S_{\vb*p}(W)\subset C_{\mathbb C}^\infty(\mathcal O)$ is the image of $W$, cf. Notation \ref{ImW}. 

\begin{proposition}\label{prop:dual_int}
If $W:\mathcal B(\mathcal H_{\vb*p})\to C_{\mathbb C}^\infty(\mathcal O)$ is a symbol correspondence with operator kernel $K$, then a symbol correspondence $\widetilde W:\mathcal B(\mathcal H_{\vb*p})\to C_{\mathbb C}^\infty(\mathcal O)$ is dual to $W$ if and only if
\begin{equation}
	A = \dim(\vb*p)\int_{\mathcal O}\widetilde W_{\!A}(\vb*\varsigma)K(\vb*\varsigma)d\vb*\varsigma
\end{equation}
for every $A\in \mathcal B(\mathcal H_{\vb*p})$. In particular, $W$ has exactly one dual correspondence satisfying $\mathcal S_{\vb*p}(\widetilde W) = \mathcal S_{\vb*p}(W)$.
\end{proposition}

\begin{proof}
Given $A_1,A_2 \in \mathcal B(\mathcal H_{\vb*p})$, consider $\widetilde A_2\in B(\mathcal H_{\vb*p})$ given by
\begin{equation}
	\widetilde A_2 = \dim(\vb*p)\int_{\mathcal O}\widetilde W_{A_2}(\vb*\varsigma)K(\vb*\varsigma)d\vb*\varsigma\, .
\end{equation}
Then, we have
\begin{equation}\label{ip_A_tildeA}
\begin{aligned}
	\ip{A_1}{\widetilde A_2}_{\vb*p} & = \tr(\int_{\mathcal O}A_1^\dagger\widetilde W_{A_2}(\vb*\varsigma)K(\vb*\varsigma)d\vb*\varsigma) = \int_{\mathcal O}\tr(A_1^\dagger K(\vb*\varsigma))\widetilde W_{A_2}(\vb*\varsigma)d\vb*\varsigma\\
	& = \int_{\mathcal O}\overline{W_{A_1}}(\vb*\varsigma)\widetilde W_{A_2}(\vb*\varsigma)d\vb*\varsigma = \ip{W_{A_1}}{\widetilde W_{A_2}} \ , \ \forall A_1 \in \mathcal B(\mathcal H_{\vb*p}) \, .
\end{aligned}
\end{equation}
But from (\ref{dualW-eq}),  $\ip{A_1}{A_2}_{\vb*p} =\ip{A_1}{\widetilde A_2}_{\vb*p}, \ \forall A_1 \in \mathcal B(\mathcal H_{\vb*p})$, if and only if  $\widetilde A_2 = A_2$, proving the first claim.

Now, consider the map (\ref{Spmap}). It trivially satisfies linearity, equivariance and reality. By the Schur's Orthogonality Relations, it sends $1$ to $\mathds 1$. In order to prove the second claim, we just need to prove such map is an isomorphism, the desired correspondence will be its inverse. Given $f\in \mathcal S_{\vb*p}(W)$, by definition there is $A\in \mathcal B(\mathcal H_{\vb*p})$ such that $f = W_A$ and, by \eqref{ip_A_tildeA},
\begin{equation}
	\ip{A}{A_f}_{\vb*p} = \ip{f}{f}\, ,
\end{equation}
so the map is injective and (\ref{Spmap}) defines $f=\widetilde W(A_f)$ s.t. $S_{\vb*p}(\widetilde W)=S_{\vb*p}(W)$.
\end{proof}

\begin{definition}\label{canndual}
	Given a symbol correspondence $W:\mathcal B(\mathcal H_{\vb*p})\to C_{\mathbb C}^\infty(\mathcal O)$ with operator kernel $K$, the \emph{canonical dual correspondence} of $W$ is the unique symbol correspondence $\widetilde W:\mathcal B(\mathcal H_{\vb*p})\to C_{\mathbb C}^\infty(\mathcal O)$ such that $W$ and $\widetilde W$ are dual to each other and $\mathcal S_{\vb*p}(\widetilde W) = \mathcal S_{\vb*p}(W)$. The operator kernel $\widetilde K$ of $\widetilde W$ is the \emph{canonical dual operator kernel} of $K$.
\end{definition}

By Proposition \ref{prop:W_op_k} and Definition \ref{def:op_k}, we obtain symbol correspondences as expected values with respect to Hermitian operators with unitary trace. From Physics, an operator on a complex Hilbert space is a \emph{state} if it is a positive operator with unitary trace. Since a general operator kernel might have negative eigenvalues, these are \emph{pseudo-states}, and we can see them as providing pseudo-probabilities just as a state provides actual probabilities. With this in mind, we have:

\begin{definition}
	A symbol correspondence is \emph{mapping-positive} if it maps positive(-definite) operators to (strictly-)positive functions. 
\end{definition}

\begin{proposition}\label{prop:mp_corresp_state}
	A symbol correspondence is mapping-positive if and only if its operator kernel is a state, that is, if and only if it is a positive operator.
\end{proposition}
\begin{proof}
	Suppose the operator kernel $K$ is a positive operator, so $K = R^\dagger R$ for some $R\in \mathcal B(\mathcal H_{\vb*p})$, and denote by $\rho(g)$ the representation of $g\in SU(3)$ on $\mathcal H_{\vb*p}$. If $A = M^\dagger M\in \mathcal B(\mathcal H_{\vb*p})$ is a positive operator, then
	\begin{equation*}
			W_A(g\vb* \varsigma_0) = \tr(M^\dagger M\rho(g)K\rho(g)^\dagger) = \tr(M\rho(g)K\rho(g)^\dagger M^\dagger) = \tr(\widetilde MR^\dagger R\widetilde M^\dagger) \ge 0
	\end{equation*}
	for every $g \in SU(3)$, where $\widetilde M = M\rho(g)$ and $\widetilde MR^\dagger R\widetilde M^\dagger$ is a positive operator. Since $K$ is non null, $R$ is also non null, so there exists $w_0\in \mathcal H_{\vb*p}$ such that $\norm{R(w_0)}^2 > 0$. If $A$ is positive-definite, then $\widetilde M^\dagger$ is an automorphism and the vector $w = (\widetilde M^\dagger)^{-1}(w_0)$ satisfies $\norm{w}>0$, so
	\begin{equation}
		\begin{aligned}
			W_A(g\vb* \varsigma_0) =  \tr(\widetilde MR^\dagger R\widetilde M^\dagger)  \ge \dfrac{\ip{w}{\widetilde MR^\dagger R\widetilde M^\dagger (w)}}{\norm{w}^2} &= \dfrac{\ip{R\widetilde M^\dagger(w)}{R\widetilde M^\dagger(w)}}{\norm{w}^2}\\
			 = \dfrac{\norm{R\widetilde M^\dagger(w)}^2}{\norm{w}^2} &= \dfrac{\norm{R(w_0)}^2}{\norm{w}^2} > 0 \ .
		\end{aligned}
	\end{equation}
	
	Now, suppose $K$ is not positive. Then, $K$ has a negative eigenvalue. Let $\Pi$ be the projection onto an eigenspace of $K$ associated to a negative eigenvalue. We have that $\tr(\Pi K) < 0$.
\end{proof}

\begin{proposition}\label{prop:mp_disjoint_sw}
	No mapping-positive symbol correspondence is Stratonovich-Weyl.
\end{proposition}
\begin{proof}
	Suppose $W:\mathcal B(\mathcal H_{\vb*p})\to C_{\mathbb C}^\infty(\mathcal O)$ is a mapping-positive correspondence. Take a non null projection $A\ne \mathds 1$. In particular, $A^\dagger = A = A^2$, and $\mathds 1 - A\ne 0$ is also a positive operator, so $0\le W_A\le 1$ is a non constant function and the set
	\begin{equation}
		\mathcal U = \{\vb*\varsigma\in \mathcal O:0<W_A(\vb*\varsigma)<1\}
	\end{equation}
	is open and non empty. By construction,
	\begin{equation}
		\vb*\varsigma\in \mathcal U \iff 0 < W_A(\vb*\varsigma)^2<W_A(\vb*\varsigma)<1\, ,
	\end{equation}
	whereas
	\begin{equation}
		\vb*\varsigma\in \mathcal O\setminus \mathcal U \iff W_A(\vb*\varsigma)^2 = W_A(\vb*\varsigma)\, .
	\end{equation}
	Then
	\begin{equation}
		\begin{aligned}
			\ip{W_A} & = \int_{\mathcal O}W_A(\vb*\varsigma)^2d\vb*\varsigma = \int_{\mathcal U}W_A(\vb*\varsigma)^2d\vb*\varsigma \ +\int_{\mathcal O\setminus \mathcal U}W_A(\vb*\varsigma)^2d\vb*\varsigma\\
			& < \int_{\mathcal U}W_A(\vb*\varsigma)d\vb*\varsigma \ +\int_{\mathcal O\setminus \mathcal U}W_A(\vb*\varsigma)d\vb*\varsigma \ = \int_{\mathcal O} W_A(\vb*\varsigma) d\vb*\varsigma \ . 
        	\end{aligned}
	\end{equation}    
   But since $A$ is a projector,         
       \begin{equation}      
            \int_{\mathcal O} W_A(\vb*\varsigma) d\vb*\varsigma  = \dfrac{1}{\dim\mathcal(\vb*p)}\tr(A) = \dfrac{1}{\dim(\vb*p)}\tr(A^2) = \ip{A}_{\vb*p}\, .
	\end{equation}
	Therefore, $W$ is not an isometry.
\end{proof}

Recall that $\mathcal H_{\widecheck{\vb*p}}=\mathcal H_{\vb*p}^\ast$ carries an irrep $\widecheck{\vb*p}$, and $\mathcal B(\mathcal H_{\vb*p})\to\mathcal B(\mathcal H_{\widecheck{\vb*p}}):A\mapsto A^\ast$ is an equivariant anti-isomorphism of algebras, so a symbol correspondence for $\vb*p$ induces a symbol correspondence for $\widecheck{\vb*p}$.

\begin{definition}\label{def:antipodal}
	For a symbol correspondence $W:\mathcal B(\mathcal H_{\vb*p})\to C_{\mathbb C}^\infty(\mathcal O)$, its \emph{antipodal correspondence} $\widecheck W:\mathcal B(\mathcal H_{\widecheck{\vb*p}})\to C_{\mathbb C}^\infty(\mathcal O)$ is the one given by
	\begin{equation}\label{antipW-def}
		\widecheck W_{A^\ast} = W_A\, .
	\end{equation}
\end{definition}

\begin{remark}
	Recalling Remark \ref{rmk:sc_conc_orbit}, if one defines a symbol correspondence as a map $W:\mathcal B(\mathcal H_{\vb*p})\to C_{\mathbb C}^\infty(\mathcal O_{(x,y)})$, then its antipodal correspondence can be defined on the antipodal orbit, $\widecheck{W}: B(\mathcal H_{\widecheck{\vb*p}})\to C_{\mathbb C}^\infty(\mathcal O_{(y,x)})$, by means of $\widecheck\psi_{(x,y)}$ so that
	\begin{equation}
		\widecheck{W}_{A^\ast} = W_A\circ \iota \, ,
	\end{equation} 
    where $\iota(\mathcal O_{(x,y)}) = \mathcal O_{(y,x)}$, cf. (\ref{iota-def}), thus the name. In Remark \ref{rmk:ant_rev_dyn} we justify a little further why it makes sense to take the antipodal orbit.
\end{remark}

For every operator $A$, we have
\begin{eqnarray}
	&\tr(A) = \tr(A^\ast) \, ,& \label{tr_adjoint} \\
    &(A^\dagger)^\ast = (A^\ast)^\dagger\, .& \label{adjoint_herm}
\end{eqnarray}

\begin{proposition}\label{prop:k_ant}
	Two symbol correspondences are antipodal to each other if and only if their operator kernels are adjoint of each other.
\end{proposition}
\begin{proof}
	For $K$ being an operator kernel, \eqref{tr_adjoint} implies
	\begin{equation}
		\tr(AK(\vb*\varsigma)) = \tr(K(\vb*\varsigma)^\ast A^\ast) = \tr(A^\ast K^\ast(\vb*\varsigma))\, .
	\end{equation}
	So the symbol correspondence with $K$ as operator kernel is antipodal to the one with $K^\ast$ as operator kernel.
\end{proof}

\begin{proposition}
	A symbol correspondence is mapping-positive if and only if its antipodal is mapping-positive.
\end{proposition}
\begin{proof}
	From \eqref{adjoint_herm}, the adjoint of a positive operator is a positive operator, so the result is a consequence of Proposition \ref{prop:mp_corresp_state}.
\end{proof}

\begin{proposition}\label{prop:ant_dual}
	The antipodal of correspondences dual to each other are also dual to each other; in particular, a symbol correspondence is Stratonovich-Weyl if and only if its antipodal is Stratonovich-Weyl.
\end{proposition}
\begin{proof}
	Equation \eqref{tr_adjoint} implies that $\ip{A_1}{A_2}_{\vb*p} = \ip{A_1^\ast}{A_2^\ast}_{\widecheck{\vb*p}}$ for every pair of operators $A_1,A_2\in \mathcal B(\mathcal H_{\vb*p})$.
\end{proof}

And recalling Definition \ref{def:op_k}, since the adjoint map is continuous, we also have:

\begin{proposition}\label{prop:conex_antip}
Two symbol correspondences are in the same connected component of $\mathcal M_{\vb*p}(\mathcal O)$ if and only if their antipodal are in the same connected component of $\mathcal M_{\widecheck{\vb*p}}(\mathcal O)$.
\end{proposition}

As mentioned before,  two symbol correspondences for the same $\vb*p$ and $\mathcal O$ may have different images when $\mathcal O\simeq\mathcal E$. Thus, in general there is nothing to assure that two symbol correspondences for different representations would have the same image, but we have the special case:

\begin{lemma}\label{lemma:antip_im_coinc}
	The image of a symbol correspondence coincides with the image of its antipodal.
\end{lemma}
\begin{proof}
	This is immediate from the definition, cf.~(\ref{antipW-def}).
\end{proof}

\subsection{Twisted algebras (``fuzzy spaces'')}

The pushforward of the algebraic structure of $\mathcal B(\mathcal H_{\vb*p})$ by a symbol correspondence $W$ induces a noncommutative algebra on $\mathcal S_{\vb*p}(W)\subset C_{\mathbb C}^\infty(\mathcal O)$.

\begin{definition}
The \emph{$\vb*p$-twisted product}, or simply the \emph{twisted product}, induced by the correspondence $W:\mathcal B(\mathcal H_{\vb*p})\to C^\infty_{\mathbb C}(\mathcal O)$ on its image $\mathcal S_{\vb*p}(W)$ is the product $\star$ given by
\begin{equation}
	W_{A_1}\star W_{A_2} = W_{A_1A_2}\, .
\end{equation}
The algebra $(\mathcal S_{\vb*p}(W),\star)$ is a \emph{twisted $\vb*p$-algebra}.\footnote{Because $\mathcal S_{\vb*p}(W)\subset C^\infty_{\mathbb C}(\mathcal O)$, some authors refer to $\mathcal O$ with such an algebra as a ``fuzzy space''.}
\end{definition}

\begin{proposition}
	Every twisted $\vb*p$-algebra is an equivariant unital $\ast$-algebra\footnote{It is a $C^\ast$-algebra if one consider also the norm induced from the operator norm.} with identity being the constant function $1$ and involution being the complex conjugation. In addition, for fixed $\vb*p$, any two twisted $\vb*p$-algebra are naturally isomorphic and any twisted $\vb*p$-algebra is naturally anti-isomorphic to any twisted $\widecheck{\vb*p}$-algebra, even for correspondences over different orbits.
\end{proposition}
\begin{proof}
	The first part is true because every symbol correspondence is an equivariant isomorphism satisfying reality and mapping the identity operator to the constant function $1$. For the second part, if $W_1$ and $W_2$ are correspondences for $\vb*p$, then we have an isomorphism of algebras $W_1\circ W_2^{-1}:\mathcal S_{\vb*p}(W_2)\to \mathcal S_{\vb*p}(W_1)$ because $W_1$ and $W_2$ are isomorphisms; but, if $W_1$ is a correspondence for $\vb*p$ and $W_2$ is a correspondence for $\widecheck{\vb*p}$, then the composition $W_1\circ \ast\circ W_2^{-1}:\mathcal S_{\widecheck{\vb*p}}(W_2)\to \mathcal S_{\vb*p}(W_1)$ is an anti-isomorphism because $W_1$ and $W_2$ are isomorphisms and $\ast$ is an anti-isomorphism.
\end{proof}

The proof of the above proposition leads us to a more strong relation between the twisted algebras induced by antipodal correspondences.

\begin{proposition}\label{prop:ant_rev_dyn}
	If $W$ and $\widecheck W$ are antipodal correspondences, for $\vb*p$ and $\widecheck{\vb*p}$, respectively, then $\mathcal S_{\vb*p}(W) = \mathcal S_{\widecheck{\vb*p}}(\widecheck W)$ and their induced twisted products $\star$ and $\widecheck\star$ satisfy
	\begin{equation}
		f_1\star f_2 = f_2\ \widecheck\star\ f_1
	\end{equation}
	for every $f_1,f_2 \in \mathcal S_{\vb*p}(W) = \mathcal S_{\widecheck{\vb*p}}(\widecheck W)$.
\end{proposition}
\begin{proof}
	Coincidence of images was established in Lemma \ref{lemma:antip_im_coinc}. Also, by definition,
	\begin{equation}
		W_{A_1A_2} = \widecheck W_{(A_1A_2)^\ast} = \widecheck W_{A_2^\ast A_1^\ast}\, 
	\end{equation}
	which finishes the proof.
\end{proof}

\begin{remark}\label{rmk:ant_rev_dyn}
	Antipodal correspondences induce reverse symbolic dinamics since the commutators $[\cdot,\cdot]_\star$ and $[\cdot,\cdot]_{\widecheck\star}$ of their twisted products satisfy
	\begin{equation}
		[f_1,f_2]_\star = -[f_1,f_2]_{\widecheck\star}\, .
	\end{equation}
	The semiclassical analysis we will pursue in Paper II involves comparing such commutator with a Poisson bracket on $\mathcal O$. If we use the isomorphisms $\psi_{(x,y)}$ and $\widecheck{\psi}_{(x,y)}$ to induce Poisson brackets on $\mathcal O$ from the restriction of the bivector field \eqref{pois_bivec} to $\mathcal O_{(x,y)}$ and $\mathcal O_{(y,x)}$, respectively, we also get reverse classical dynamics, so the symbol dynamics and the classical dynamics have equal changes of orientation. 
\end{remark}

Now, an integral formulation of twisted product can be obtained by means of Proposition \ref{prop:dual_int}, as follows.

\begin{proposition}\label{prop:int_trik}
	If $W,\widetilde{W}:\mathcal B(\mathcal H_{\vb*p})\to C_{\mathbb C}^\infty(\mathcal O)$ are canonically dual correspondences with operators kernel $K, \widetilde{K}$, respectively, then the twisted product induced by $W$ is given by
	\begin{equation}\label{tp_int}
		f_1\star f_2(\vb* \varsigma) = \int_{\mathcal O\times \mathcal O} f_1(\vb* \varsigma_1)f_2(\vb* \varsigma_2)\mathbb L^{\!W}_{\vb*p}(\vb* \varsigma_1, \vb* \varsigma_2, \vb* \varsigma)d\vb* \varsigma_1 d\vb* \varsigma_2
	\end{equation}
	for any $f_1, f_2 \in \mathcal S_{\vb*p}(W)$, where
	\begin{equation}\label{int_trik}
		\mathbb L_{\vb*p}^{\!W}(\vb* \varsigma_1, \vb* \varsigma_2, \vb* \varsigma) = \dim(\vb*p)^2 \tr(\widetilde K(\vb* \varsigma_1)\widetilde K(\vb* \varsigma_2) K(\vb* \varsigma))\, .
	\end{equation}
\end{proposition}
\begin{proof}
	By writing $f_k = W_{A_k}$ for $k=1,2$, we have
	\begin{equation}
		f_1\star f_2(\vb* \varsigma) = \tr(A_1A_2 K(\vb* \varsigma))\, .
	\end{equation}
	Using Proposition \ref{prop:dual_int} to write $A_1$ and $A_2$, we get the statement.
\end{proof}

\begin{definition}
	The \emph{integral trikernel} $\mathbb L_{\vb*p}^{\!W}\in C_{\mathbb C}^\infty(\mathcal O\times \mathcal O\times \mathcal O)$ of a twisted product induced by a symbol correspondence $W:\mathcal B(\mathcal H_{\vb*p})\to C_{\mathbb C}^\infty(\mathcal O)$ is a function of the form \eqref{int_trik}, for $K$ being the operator kernel of $W$ and $\widetilde K$ being the canonical dual operator kernel of $K$, so that the twisted product is given by \eqref{tp_int}.
\end{definition}

\begin{remark}
	In \cite{RS}, new formulas for integral trikernels of spin systems were obtained using $SU(2)$-invariant $2$-point and $3$-point functions on $\mathbb CP^1\simeq \mathcal S^2$, but a similar exercise of finding new $SU(3)$-invariant formulas for trikernels on $\mathbb CP^2$ or $\mathcal E$
    is much harder and shall not be pursued here.
\end{remark}

\begin{proposition}
Let $\mathbb L^{W}_{\vb*p}$ be an integral trikernel of a $\vb*p$-twisted product $\star$. Then, for every $g \in SU(3)$ and every $\vb* \varsigma_1, \vb* \varsigma_2, \vb* \varsigma_3, \vb* \varsigma_4 \in \mathcal O$,
\begin{enumerate}
	\item $\mathbb L_{\vb*p}^{W}(\vb* \varsigma_1, \vb* \varsigma_2, \vb* \varsigma_3) = \mathbb L_{\vb*p}^{W}(g\vb* \varsigma_1, g\vb* \varsigma_2, g\vb* \varsigma_3)$;
	\item $\int_{\mathcal O}\mathbb L_{\vb*p}^{W}(\vb* \varsigma_1, \vb* \varsigma_2, \vb* \varsigma)\mathbb L_{\vb*p}^{W}(\vb* \varsigma, \vb* \varsigma_3, \vb* \varsigma_4)d\vb* \varsigma = \int_{\mathcal O}\mathbb L_{\vb*p}^{W}(\vb* \varsigma_1, \vb* \varsigma, \vb* \varsigma_4)\mathbb L_{\vb*p}^{W}(\vb* \varsigma_2, \vb* \varsigma_3, \vb* \varsigma)d\vb* \varsigma$;
	\item $\mathcal R_{\vb*p}^{W}(\vb* \varsigma_1, \vb* \varsigma_2) := \int_{\mathcal O}\mathbb L_{\vb*p}^{W}(\vb* \varsigma, \vb* \varsigma_1, \vb* \varsigma_2)d\vb* \varsigma = \int_{\mathcal O}\mathbb L_{\vb*p}^{W}(\vb* \varsigma_1, \vb* \varsigma, \vb* \varsigma_2)d\vb* \varsigma$ is a reproducing kernel for $\mathcal S_{\vb*p}(W)$; and
	\item $\overline{\mathbb L_{\vb*p}^{W}(\vb* \varsigma_1, \vb* \varsigma_2, \vb* \varsigma_3)} = \mathbb L_{\vb*p}^{W}(\vb* \varsigma_2, \vb* \varsigma_1, \vb* \varsigma_3)$.
\end{enumerate}
\end{proposition}
\begin{proof}
First, we note that the operator-valued functions $\widetilde K(\vb* \varsigma)$ and $K(\vb* \varsigma)$ can be taken as elements of $\mathcal S_{\vb*p}(W)\otimes \mathcal B(\mathcal H_{\vb*p})$, so that $\mathbb L^W_{\vb*p}$ is an element of $\mathcal S_{\vb*p}(W)\otimes \mathcal S_{\vb*p}(W)\otimes \mathcal S_{\vb*p}(W)$. Thus, item (i) follows from equivariance of $\star$, item (ii) follows from associativity of $\star$, item (iii) follows from the fact that the constant function $1$ is the identity of the twisted algebra, and item (iv) follows from the fact that conjugation is the involution of the twisted algebra.
\end{proof}

One may use \eqref{tp_int} to expand a twisted product on $S_{\vb*p}(W)$ induced by a symbol correspondence $W$ to all $C_{\mathbb C}^\infty(\mathcal O)$.

\begin{proposition}\label{b-prod-hs}
	Let $\mathbb L_{\vb*p}^W$ be the integral trikernel of a $\vb*p$-twisted product induced by a symbol correspondence $W$. The binary operation $\bullet$ given by 
	\begin{equation}\label{bul-prod}
		f_1\bullet f_2 (\vb* \varsigma) = \int_{\mathcal O\times \mathcal O}f_1(\vb* \varsigma_1)f_2(\vb* \varsigma_2)\mathbb L(\vb* \varsigma_1,\vb* \varsigma_2,\vb* \varsigma)\,d\vb* \varsigma_1d\vb* \varsigma_2
	\end{equation}
	for any $f_1,f_2\in C_{\mathbb C}^\infty(\mathcal O)$, defines an $SU(3)$-equivariant associative $\ast$-algebra on $C_{\mathbb C}^\infty(\mathcal O)$ with respect to complex conjugation.  In particular, if $f_1, f_2\in \mathcal S_{\vb*p}(W)$, we have $f_1\bullet f_2 = f_1\star f_2$. But, if either $f_1$ or $f_2$ is orthogonal to $\mathcal S_{\vb*p}(W)$, then $f_1\bullet f_2 = 0$, and thus $C_{\mathbb C}^\infty(\mathcal O)\to \mathcal S_{\vb*p}(W):f\mapsto 1\bullet f = f\bullet 1$ is an orthogonal projection.
\end{proposition}
\begin{proof}
	Linearity of integral implies the product is bilinear, hence it defines an algebra. By definition, it is clear that $f_1 \bullet f_2 = f_1 \star f_2$ if $f_1, f_2 \in \mathcal S_{\vb*p}(W)$. Now, if $f_k$ is orthogonal to $\mathcal S_{\vb*p}(W)$, by the already made observation that $\mathbb L^W_{\vb*p}$ is an element of $\mathcal S_{\vb*p}(W)\otimes \mathcal S_{\vb*p}(W)\otimes \mathcal S_{\vb*p}(W)$, the integral over $\vb* \varsigma_k$ in \eqref{bul-prod} results in $0$, so $f_1 \bullet f_2 = 0$. Since any $f \in C_{\mathbb C}^\infty(\mathcal O)$ can be decomposed into $f = f_\parallel + f_\perp$, where $f_\parallel \in\mathcal S_{\vb*p}(W)$ and $f_\perp$ is orthogonal to $S_{\vb*p}(W)$, the $SU(3)$-equivariant, associative and $\ast$-algebra properties of $\star$ extends to $\bullet$.
\end{proof}

\section{Pure-quark systems}\label{sec:pq-sys}
Here we focus on the simpler possible phase space for a classical quark system: $\mathbb CP^2$. First, we describe the set of harmonic functions on $\mathbb CP^2$, which imposes a restriction on the classes of  irreducible representations of $SU(3)$ with possible correspondences to smooth functions on $\mathbb CP^2$. Then, we proceed to describe the relevant $SU(3)$-representations for this case as quantum quark systems. Finally, we work out the characterization of all symbol correspondences from such quantum quark systems to the classical quark system of interest and describe the induced twisted products of symbols on $\mathbb CP^2$. The construction and characterization of symbol correspondences in this section is very close to what is done for spin systems in \cite{RS}. Accordingly,  proofs of some propositions are identical to the $SU(2)$ case. The quantum and classical systems in correspondence, in this chapter, are both called ``pure-quark system'' and this name is explained in Appendix \ref{sec:def_PQ}.

\subsection{Classical pure-quark system}

\begin{definition}
The \emph{classical pure-quark system} consists of $\mathbb CP^2$ equipped with its $SU(3)$-invariant  symplectic form,  together with its Poisson algebra on $C^\infty_{\mathbb C}(\mathbb CP^2)$.
\end{definition}

Since $\mathbb CP^2 \simeq SU(3)/H$, where $H\simeq U(2)$, cf. (\ref{u2-ut}), we look for representations $(p,q)$ with weights satisfying $t=u=J=0$ (cf. (\ref{const})-(\ref{tuv-nu})) to determine the harmonic functions on $\mathbb CP^2$.

\begin{proposition}\label{Qnn}
The representations of $SU(3)$ with non null vectors fixed by $H\simeq U(2)$ are the representations $(n,n)$. The space fixed by $H$ is spanned by $\vb* e((n,n); \vb*0_n, 0)$, where $\vb*0_n=(n,n,n)$ represents the null weight.
\end{proposition}
\begin{proof}
Let $\vb*e((p,q);(\nu_1,\nu_2,\nu_3), J)$ be such that $t = u = J = 0$, so $\nu_1=\nu_2=\nu_3=\nu$. From the constraints (\ref{const}), we get $r_+=r_-=\nu = q = p$. Thus, putting $n=p=q$, we finish the proof.
\end{proof}

\begin{definition}\label{harm-def}
The \emph{pure-quark harmonics}, or $\mathbb CP^2$-\emph{harmonics} are the functions $X^n_{\vb*\nu, J}: \mathbb CP^2\to\mathbb C$, such that 
\begin{equation}\label{harm-def-eq1}
\begin{aligned}
    X^n_{\vb*\nu, J}(g\vb x_0)
    & = (n+1)^{3/2}\overline{D^{(n,n)}_{\vb* \nu J, \vb*0_n 0}}(g) \ ,
\end{aligned}
\end{equation}
for $\vb x_0 = [1:0:0]\in \mathbb CP^2$ and $g \in SU(3)$.
\end{definition}

The factor $(n+1)^{3/2}$ in the definition of $\mathbb CP^2$-harmonics is the square root of the dimension of the representation $(n,n)$ and is used to ensure normalization according to Schur's Orthogonality Relations, so that we have
\begin{equation}
\ip{X^n_{\vb*\nu, J}}{X^m_{\vb*\mu, L}} = \delta_{n,m}\delta_{\vb*\nu,\vb*\mu}\delta_{J,L}    
\end{equation}
with respect to the inner product described in section \ref{co-sec}, cf. (\ref{haar-leftinvint})-(\ref{ip-orb}). 

We note that 
\begin{equation}\label{X^0=1}
   X^0_{(0,0,0),0}\equiv 1 
\end{equation}
and, cf. (\ref{conj-wfun}), 
\begin{equation}\label{conj_X^n}
 \overline{X^n_{\vb*\nu, J}} = (-1)^{2(t+u)}X^n_{\widecheck{\vb*\nu}, J} \ , \ \mbox{for} \ \    \Delta^{2n}_{\vb*\nu,\widecheck{\vb*\nu}}=1 \ . 
\end{equation}

\begin{remark}\label{cp2-harm-orbit}
Fixed $x>0$, the diffeomorphism $\psi_{(x,0)}$ can be used to carry $\mathbb CP^2$ harmonics to $\mathcal O_{(x,0)}$ by means of the composition $X^n_{\vb*\nu, J}\circ \psi_{(x,0)}$, cf. (\ref{psixy}). Equivalently, $X^n_{\vb*\nu, J}\circ \psi_{(0,x)}$ are the $\mathbb CP^2$ harmonics carried to the orbit $\mathcal O_{(0,x)}$. Consequently, we have a set of harmonic functions on $\mathcal O_{(x,0)}$ related to a set of harmonic functions on $\mathcal O_{(0,x)}$ by the map $\iota$, cf. (\ref{id-x0}).
\end{remark}

Then, from Lemma \ref{wig-prod}, we have the following.

\begin{theorem}\label{harm-point-prod}
The pointwise product of $\mathbb CP^2$ harmonics decomposes as
\begin{equation}\label{harm-point-prod-eq}
\begin{aligned}
    X^{n_1}_{\vb* \nu_1, J_1}X^{n_2}_{\vb* \nu_2,J_2} = \sum_{\substack{(n,n;\sigma)\\\vb*\nu, J}}&\left(\dfrac{(n_1+1)(n_2+1)}{n+1}\right)^{3/2}\cg{(n_1,n_1)}{(n_2,n_2)}{(n,n;\sigma)}{\,\,\vb*\nu_1 J_1}{\,\,\vb*\nu_2 J_2}{\,\,\vb*\nu J}\\
    &\times \cg{(n_1,n_1)}{(n_2,n_2)}{(n,n;\sigma)}{\,\,\vb*0_{n_1} 0}{\,\,\vb*0_{n_2} 0}{\,\,\vb*0_n 0}X^n_{\vb*\nu, J} \ ,
\end{aligned}
\end{equation}
where $\vb*0_{n_k} = (n_k,n_k,n_k)$ for $k=1,2$ and $\vb*0_n = (n,n,n)$, and summation is restricted to $\nabla_{\vb*\nu_1+\vb*\nu_2, \vb*\nu}=1$, $\delta(J_1, J_2, J) = 1$ and $(n,n;\sigma)$ in the Clebsch-Gordan series of $(n_1,n_1)\otimes (n_2,n_2)$; in particular, $|n_1-n_2|\le n \le n_1+n_2$.
\end{theorem}
\begin{proof}
With a little abuse of notation, we write
\begin{equation}
    X^{n_k}_{\vb*\nu_k,J_k} = (n_k+1)^{3/2}\overline{D^{(n_k,n_k)}_{\vb* \nu_k J_k,\vb*n_k 0}} \ ,
\end{equation}
and apply Lemma \ref{wig-prod} to get
\begin{equation}
\begin{aligned}
    X^{n_1}_{\vb* \nu_1, J_1}X^{n_2}_{\vb* \nu_2,J_2} = \sum_{(\vb*a;\sigma)}\sum_{\substack{\vb*\nu,J\\\vb*\mu,0}}&\left((n_1+1)(n_2+1)\right)^{3/2}\cg{(n_1,n_1)}{(n_2,n_2)}{(\vb*a;\sigma)}{\,\,\vb*\nu_1 J_1}{\,\,\vb*\nu_2 J_2}{\,\,\vb*\nu J}\\
    & \times \cg{(n_1,n_1)}{(n_2,n_2)}{(\vb*a;\sigma)}{\,\,\vb*0_{n_1} 0}{\,\,\vb*0_{n_2} 0}{\,\,\vb*\mu 0}\overline{D^{\vb*a}_{\vb* \nu J,\vb*\mu 0}} \ ,
\end{aligned}
\end{equation}
where $\nabla_{\vb*\nu_1+\vb*\nu_2, \vb*\nu}=\nabla_{\vb*n_1+\vb*n_2,\vb*\mu} =1$ and $\delta(J_1,J_2,J)=1$, so $\vb*\mu = (\mu,\mu,\mu)$. But $\vb*e(\vb*a;(\mu,\mu,\mu), 0)$ only exists if $\vb*a = (\mu,\mu)$. Thus, we set $\vb*a = (n,n)$ and $\vb*\mu = \vb*0_n = (n,n,n)$. The restriction over $n$ follows from Theorem \ref{cg-s}.
\end{proof}

\subsection{Quantum pure-quark systems}

From Proposition \ref{Qnn}, we look for representations $(p,q)$ such that the tensor product $(p,q)\otimes (q,p)$ splits into representations of the form $(n,n)$, without multiplicities. From  Corollary \ref{cg-series-corol}, we have that $(p,0)$ and $(0,p)$ are the only ones that satisfy these requirements. These are special cases of quantum quark systems.

\begin{definition}\label{QPQsystem}
Let $\vb*p\in (\mathbb N\times\{0\})\cup(\{0\}\times \mathbb N)$ with $|\vb*p| = p$. A \emph{quantum pure-quark system}\footnote{The reason for the name \emph{pure-quark} systems is explained in Appendix \ref{sec:def_PQ}.} is a complex Hilbert space $\mathcal H_{\vb*p}\simeq \mathbb C^d$, where
\begin{equation}\label{dimpure}
    d = \dim(\vb*p)= \dfrac{(p+1)(p+2)}{2} \ ,    
\end{equation}
cf. (\ref{dim}), with an irreducible unitary $SU(3)$-representation $\vb*p$ together with its operator algebra $\mathcal B(\mathcal H_{\vb*p})$.
\end{definition}

In the pure-quark case, the Gelfand-Tsetlin pattern (\ref{const}) is reduced to
\begin{equation}\label{const-2}
\vb*p=(p,0)\implies 
\begin{cases}
\begin{aligned}
0\le r   \le p & \ , \\
\nu_1 = p-r \ , \ \nu_2 = r-\nu_3 & \ , \ 0 \le \nu_3 \le r \ ,\\
J = \dfrac{r}{2} \ .&
\end{aligned}
\end{cases}
\end{equation}
\begin{equation}\label{const-3}
\vb*p=(0,p)\implies 
\begin{cases}
\begin{aligned}
0 & \le r \le   p \ , \\
 \nu_1 = p-r \ , \ 
\nu_2 & = p+r-\nu_3 \ , \ r\le \nu_3 \le p \ , \\
& J = \dfrac{p-r}{2} \ .
\end{aligned}
\end{cases} 
\end{equation}
In both cases, $J$ is determined by $\vb*\nu$, so we can simplify the notation as 
\begin{equation}\label{basis-simpl-pq}
\vb* e(\vb*p; \vb* \nu) := \vb* e(\vb*p; \vb* \nu, J) \ , \   \widecheck{\vb*e}(\widecheck{\vb*p}; \vb*\nu) := \widecheck{\vb*e}(\widecheck{\vb*p}; \vb*\nu, J) \ .   
\end{equation}
To clear even more the notation, we will denote the elements of a coupled basis of $\mathcal B(\mathcal H_{\vb*p})$ that lies in the $(n,n)$-invariant subspace by 
\begin{equation}\label{n-simp}
   \vb* e(n; \vb* \nu, J):=\vb* e((n,n); \vb* \nu, J)\equiv \vb* e_{\vb*p,\widecheck{\vb*p}}((n,n); \vb* \nu, J) \ , 
\end{equation}
cf. (\ref{e=e}).
Thus, the notation for the Clebsch-Gordan coefficients can be simplified to
\begin{equation}\label{simplifnotation}
\cg{\vb*{p}}{\widecheck{\vb*{p}}}{n}{\vb*\nu_1 }{\vb*\nu_2}{\vb*\nu J} :=\cg{\vb*{p}}{\widecheck{\vb*{p}}}{n}{\vb*\nu_1 J_1}{\vb*\nu_2 J_2}{\vb*\nu J} :=\cg{\vb*p}{\widecheck{\vb*p}}{(n,n)}{\vb*\nu_1 J_1}{\vb*\nu_2 J_2}{\,\,\vb*\nu J} \ .
\end{equation}
And applying the same simplification to the Wigner product symbol, Theorem \ref{op-prod-gen-s} takes the special form:

\begin{theorem}\label{prod-sym}
For a quantum pure-quark system $\mathcal H_{\vb*p}$, $|\vb*p|=p$, the product of elements of a coupled basis of the space of operators $\mathcal B(\mathcal H_{\vb*p})$ decomposes in the form 
\begin{equation}\label{op-prod-eq}
	\begin{aligned}
		&\vb* e(n_1; \vb*\nu_1, J_1)\vb* e(n_2; \vb*\nu_2, J_2) \\ 
		& \hspace{5 em} = \sum_{n=0}^p\sum_{\vb*\nu,J}(-1)^{p+2(t_{\vb*\nu}+u_{\vb*\nu})}\!\!\begin{bmatrix}
			n_1 & n_2 & n\\
			\vb*\nu_1, J_1 & \vb*\nu_2, J_2 & \widecheck{\vb*\nu}, J
		\end{bmatrix}\!\![\vb* p]\vb* e(n; \vb*\nu, J)\, ,
	\end{aligned}
\end{equation}
with summations over $\vb*\nu$ and $J$ effectively restricted by $\nabla_{\vb*\nu_1+\vb*\nu_2,\vb*\nu}=\delta(J_1,J_2,J)=1$, for $0\le n_1,n_2\le p$. 
\end{theorem}

We also identify the operator algebra $\mathcal B(\mathcal H_{\vb*p})$ with the matrix algebra $M_{\mathbb C}(d)$ by means of an uncoupled basis of $\vb*p\otimes \widecheck{\vb*p}$. So let $\vb*\nu$ and $\widecheck{\vb*\nu}$ be such that $\Delta^{|\vb*p|}_{\vb*\nu,\vb*\nu'}=1$, the operator $\vb*e(\vb*p,\vb*\nu)\otimes \widecheck{\vb*e}(\widecheck{\vb*p}, \widecheck{\vb*\nu})$ is a diagonal matrix and its decomposition in the coupled basis can be written as
\begin{equation}\label{diag-decomp}
    \vb*e(\vb*p;\vb*\nu)\otimes \widecheck{\vb*e}(\widecheck{\vb*p}; \widecheck{\vb*\nu}) = \sum_{n=0}^p\sum_{J = 0}^{n}\cg{\vb*p}{\widecheck{\vb*p}}{n}{\vb*\nu}{\widecheck{\vb*\nu}}{\vb*0_n J}\vb*e(n;\vb*0_n, J) \ .
\end{equation}
That is, any diagonal matrix is a linear combination of $\{\vb*e(n; \vb*0_n, J)\}$. Since the cardinality of this set is $(p+1)(p+2)/2$, it is the set of diagonal matrices of a coupled basis. In accordance to the choice of CG coefficients being real numbers, we can and we do take such operators as real matrices.

\subsection{Correspondences over pure-quark systems}

Until the end of this section, whenever we talk about symbol correspondences, we are referring to correspondences for the classical pure-quark system $\mathbb CP^2$. In this vein, let $\vb*p \in(\mathbb N\times\{0\})\cup(\{0\}\times \mathbb N)$ with $|\vb*p| = p$. For $\vb x_0 = [1:0:0]\in \mathbb CP^2$, given an operator $K\in \mathcal B(\mathcal H_{\vb*p})$ fixed by $H$, the map \eqref{op-val_f} assumes the form
\begin{equation}
	\mathbb CP^2 \to \mathcal B(\mathcal H_{\vb*p}): \vb x = g\vb x_0 \mapsto K(\vb x) = K(g\vb x_0) = K^g\, .
\end{equation}
Then, Proposition \ref{prop:W_op_k} is strengthened as follows.

\begin{theorem}\label{symb_corresp-ker}
A map $W:\mathcal B(\mathcal H_{\vb*p}) \to C^\infty_\mathbb{C}(\mathbb CP^2): A \mapsto W_A$ is a symbol correspondence if and only if
\begin{equation}\label{sc-ker}
	W_A(\vb x) = \tr(AK(\vb x))
\end{equation}
for $K\in \mathcal B(\mathcal H_{\vb*p})$, its operator kernel, of the form
\begin{equation}\label{op-ker-eq}
	\begin{aligned}
		K &= \sum_{n=0}^p c_n \sqrt{\dfrac{(n+1)^3}{\dim(\vb*p)}}\vb*e(n; \vb*0_n, 0)\\
		& =\dfrac{1}{\dim(\vb*p)}\mathds 1 + \sum_{n=1}^p c_n \sqrt{\dfrac{(n+1)^3}{\dim(\vb*p)}}\vb*e(n; \vb*0_n, 0) \, ,
	\end{aligned}
\end{equation}
cf. (\ref{dimpure}), 
with $(c_1,...,c_{p})\in (\mathbb R^\times)^p$ and $c_0 = 1$. In particular, 
\begin{equation}\label{sc-b-eq}
	W: \ \ \sqrt{\dim(\vb*p)}\vb*e(n; \vb*\nu, J) \ \ \mapsto \ \ c_n X^n_{\vb*\nu, J}\, .
\end{equation}
\end{theorem}
\begin{proof}
Suppose $W$ is a symbol correspondence. By Proposition \ref{Qnn} and the fact that each $\vb*e(n;\vb*0_n,0)$ is real and diagonal, we must have
\begin{equation}\label{K-k_n}
    K = \sum_{n=0}^p
    k_n\vb*e(n; \vb*0_n, 0)
\end{equation}
with each $k_n$ real. Injectivity implies $k_n\ne 0$. If $A = \vb*e(n;\vb*\nu, J)$, then
\begin{equation}
	\tr(AK(\vb x)) = \tr(AK(\vb x)^\dagger) = \dfrac{k_n}{(n+1)^{3/2}}X^n_{\vb*\nu, J}\, ,
\end{equation}
so that
\begin{equation}\label{kncn}
	k_n = c_n \sqrt{\dfrac{(n+1)^3}{\dim(\vb*p)}}
\end{equation}
gives what we want. On the other hand, for $K$ given by (\ref{op-ker-eq}), we have that equations \eqref{tr_equiv} and \eqref{tr_real}, the hypothesis that $c_n\in \mathbb R^\times$ and the fact that $\tr(K)=1$ imply that (\ref{sc-ker}) defines a symbol correspondence given by (\ref{sc-b-eq}).
\end{proof}

\begin{corollary}\label{moduli-space-pq}
	The moduli space $\mathcal M_{\vb*p}(\mathbb CP^2)$ of symbol correspondences for a pure-quark system $\vb*p$ with $|\vb*p| = p$ is $(\mathbb R^\times)^p$.
\end{corollary}

\begin{corollary}
The images of all symbol correspondences for $\vb*p$ satisfying $|\vb*p| = p$ are the same space, namely, the space $\mathcal X_p$ spanned by the $\mathbb CP^2$ harmonics $X^n_{\vb*\nu, J}$ with $0\le n\le p$.
\end{corollary}

\begin{remark}
Theorem \ref{symb_corresp-ker} shows that taking real CG coefficients already provides
\begin{equation}
	\vb*e^\dagger(n;\vb*\nu, J) = (-1)^{2(t+u)}\vb*e(n;\widecheck{\vb*\nu}, J)
\end{equation}
without appealing to further symmetries of CG coefficients.
\end{remark}

\begin{definition}
	The numbers $c_n$ of Theorem \ref{symb_corresp-ker} are referred to as \emph{characteristic numbers} of both the correspondence and the operator kernel.
\end{definition}

If $K\in \mathcal B(\mathcal H_{\vb*p})$ is an operator kernel, it is diagonal with real entries, thus it is a linear combination of projections of the form 
\begin{equation}\label{K-conv-comb}
    K = \sum_{\vb*\nu}a_{\vb*\nu}\, \Pi_{\vb*\nu} \ , 
\end{equation}
for real coefficients $a_{\vb*\nu}$, where $\Pi_{\vb*\nu}$ is an orthogonal projector onto the weight space of $\vb*\nu$. We can separate the summation as
\begin{equation}\label{I-sep}
    K = \sum_{j = 0}^{p} K_j \ , \ \ K_j = \sum_{\vb*\nu \in j/2} a_{\vb*\nu}\, \Pi_{\vb*\nu} \ ,
\end{equation}
where $\vb*\nu \in j/2$ means $\vb*e(\vb*p;\vb*\nu) \equiv \vb*e(\vb*p;\vb*\nu,j/2)$, cf. (\ref{const-2})-(\ref{basis-simpl-pq}).

\begin{proposition}\label{op-ker-proj}
If $K\in\mathcal B(\mathcal H_{\vb*p})$ is an operator kernel, then
\begin{equation}\label{op-ker-proj-eq}
    K = \sum_{j=0}^{p}a_j\sum_{\vb*\nu \in j/2}\Pi_{\vb*\nu} \, ,
\end{equation}
where the coefficients $a_j$ are real numbers satisfying
\begin{equation}\label{sum-ar}
    \sum_{j=0}^p a_j(j+1) = 1 \ .
\end{equation}
\end{proposition}
\begin{proof}
Every operator kernel is fixed by $H$, so $K$ must be fixed also by the $u$-standard $SU(2)$, cf. \eqref{u2-ut} and \eqref{su2}. Decomposing $K$ as in (\ref{I-sep}), we have that each component $K_j$ is an operator on the subrepresentation $j/2$ of that $u$-standard $SU(2)$ and it must commute with $\{U_3, U_\pm\}$, which implies $\displaystyle{K_j = a_j\sum_{\vb*\nu \in j/2} \Pi_{\vb*\nu}}$ , where $a_j$ is real. To finish, $\tr(K) = 1$ implies (\ref{sum-ar}).
\end{proof}

We can deduce uniqueness of dual correspondences for pure-quark systems from Proposition \ref{prop:dual_int}.

\begin{proposition}\label{dual-pquark-prop}
	Let $K$ be an operator kernel with characteristic numbers $(c_n)$. The equation
	\begin{equation}\label{cont-corresp}
		A = \dim(\vb*p)\int_{\mathbb CP^2}\widetilde W_A(\vb x)K(\vb x)d\vb x \ .
	\end{equation}
	defines implicitly a unique symbol correspondence $\widetilde W$ with characteristic numbers
	\begin{equation}\label{tildecc}
		\widetilde c_n = 1/c_n \ .   
	\end{equation}
\end{proposition}
\begin{proof}
	We have
	\begin{equation}
		\begin{aligned}
			&\int_{\mathbb CP^2} X^n_{\vb*\nu, J}(\vb x)K(\vb x)d\vb x \  = \  \int_{SU(3)} X^n_{\vb*\nu, J}(g\vb x_0)K^g d g\\
			& = \sum_{n',\vb*\mu, L}\dfrac{k_{n'}}{(n'+1)^{3/2}}\ip{X^{n'}_{\vb*\mu, L}}{X^n_{\vb*\nu, J}}\vb*e(n;\vb*\mu, L)  = \dfrac{k_n}{(n+1)^{3/2}}\vb*e(n;\vb*\nu, J) \ ,
		\end{aligned}
	\end{equation}
	where $k_n$ is given by (\ref{kncn}). Thus,
	\begin{equation}
		\dim(\vb*p) \int_{\mathbb CP^2} \dfrac{1}{c_n}X^n_{\vb*\nu, J}(\vb x)K(\vb x)d\vb x = \sqrt{\dim(\vb*p)} \vb*e(n; \vb*\nu, J) \ .
	\end{equation}
	By Theorem \ref{symb_corresp-ker}, $\widetilde W$ is a symbol correspondence with characteristic numbers $(\widetilde{c}_n)$ satisfying (\ref{tildecc}).
\end{proof}

\begin{theorem}\label{dual-prop}
	For any symbol correspondence with characteristic numbers $(c_n)$, its unique dual correspondence has characteristic numbers $(1/c_n)$.
\end{theorem}
\begin{proof}
It is a direct consequence of Propositions \ref{prop:dual_int} and \ref{dual-pquark-prop}.
\end{proof}

\begin{corollary}\label{str-weyl-cn}
	A symbol correspondence is a Stratonovich-Weyl correspondence if and only if its characteristic numbers $(c_n)$ satisfy $|c_n|=1$.
\end{corollary}

Corollary above shows the existence of Stratonovich-Weyl correspondences. To verify the existence of a mapping-positive correspondence, we follow very closely the construction of \cite{RS} for spin systems. Let $\rho_1$ be the defining representation of $SU(3)$ on $\mathbb C^3$, and consider the canonical projection $S^5\to \mathbb CP^2:\vb n \mapsto [\vb n]$. Then,
\begin{equation}
	\Phi_p  : \mathbb C^3 \to \mathbb C^{(p+1)(p+2)/2} : \ 
	(z_1,z_2,z_3) \mapsto (z_1^p,...,\sqrt{{p \choose{j,k,l}}}z_1^j z_2^k z_3^l,..., z_3^p)
\end{equation}
is equivariant for the representation $\rho_p$ of class $(p,0)$ on $\mathbb C^{(p+1)(p+2)/2}\simeq \mathcal H_{(p,0)}$ described in Appendix \ref{sec:def_PQ}.

\begin{proposition}\label{ber-prop}
	The map $B: \mathcal B(\mathcal H_{(p,0)})\to C^\infty_\mathbb{C}(\mathbb CP^2): A\mapsto B_A$, with
	\begin{equation}
		B_A(\vb x) = \ip{\Phi_p(\vb n)}{A\Phi_p(\vb n)}
	\end{equation}
	for $\vb x \in \mathbb CP^2$ and $\vb n \in S^5$ related by $[\vb n] = \vb x$, is a mapping-positive symbol correspondence whose operator kernel is the projection $\Pi_{(p,0,0)}$ onto the highest weight space of $(p,0)$ and whose characteristic numbers are
	\begin{equation}\label{b+}
		b_n = (-1)^p\sqrt{\frac{\dim(\vb*p)}{(n+1)^3}}\cg{(p,0)}{(0,p)}{n}{(p,0,0)}{(0,p,p)}{(n,n,n),0} \ .
	\end{equation}
\end{proposition}
\begin{proof}
	First of all, for $\vb n, \vb n' \in S^5$, we have $[\vb n] = [\vb n']$ if and only if $\vb n' = e^{i\theta} \vb n$, but $\Phi_p(e^{i\theta }\vb n) = e^{ip\theta}\Phi_p(\vb n)$, so
	\begin{equation*}
		\ip{\Phi_p(e^{i\theta}\vb n)}{A\Phi_p(e^{i\theta}\vb n)} = \ip{e^{ip\theta}\Phi_p(\vb n)}{e^{ip\theta}A\Phi_p(\vb n)} = \ip{\Phi_p(\vb n)}{A\Phi_p(\vb n)} \ ,
	\end{equation*}
	hence $B_A$ is a well defined function on $\mathbb CP^2$ for any $A \in \mathcal B(\mathcal H_{(p,0)})$. It is also smooth, since the projection $\vb n\mapsto [\vb n]$ is a surjective submersion and its composition with $B_A$ is a smooth function.
	
	The linearity of $B$ is trivial. The equivariance follows straightforwardly from the equivariance of $\Phi_p$. For any $g \in SU(3)$,
	\begin{equation*}\label{equiv-berez}
		\begin{aligned}
			B_{A^g}(\vb x) & = \ip{\Phi_p(\vb n)}{A^g\Phi_p(\vb n)} = \ip{\Phi_p(\vb n)}{\rho_p(g)A\rho_p(g)^{-1}\Phi_p(\vb n)}\\
			& = \ip{\rho_p(g)^{-1}\Phi_p(\vb n)}{A\rho_p(g)^{-1}\Phi_p(\vb n)} = \ip{\Phi_p(\rho_1(g)^{-1}\vb n)}{A\Phi_p(\rho_1(g)^{-1}\vb n)}\\
			& = B_A(g^{-1}\vb x) = (B_A)^g(\vb x) \,.
		\end{aligned}
	\end{equation*}
	Equivariance implies that $\ker(B)$ is an invariant subspace, and we use this to prove $B$ is injective by means of contradiction. Suppose $B$ is not injective, then $\ker(B)$ contains an irreducible representation of the form $(n,n)$, so the diagonal matrix $\vb* e(n, \vb*0_n, J)$ lies in $\ker(B)$, that is, there exists a non-zero diagonal matrix $D = diag(d_{p,0,0},..., d_{j,k,l},..., d_{0,0,p})\in \ker (B)$. Thus, we have
	\begin{equation*}
		B_D(\vb x) = \ip{\Phi_p(\vb n)}{D\Phi_p(\vb n)} = \sum_{j+k+l=1}{p\choose{j,k,l}}d_{j,k,l}\,|z_1|^{2j}|z_2|^{2k}|z_3|^{2l} = 0
	\end{equation*}
	for every $\vb n = (z_1,z_2,z_3) \in S^5$, so $D$ must be a zero matrix, and this is the desired contradiction, therefore $B$ is injective. Since
	\begin{equation}\label{proj-eq}
			\Pi_{(p,0,0)}  = \dfrac{1}{\dim(\vb*p)}\mathds 1
             + (-1)^p\sum_{n=1}^p \cg{(p,0)}{(0,p)}{n}{(p,0,0)}{(0,p,p)}{(n,n,n),0} \,\vb*e(n;(n,n,n),0) \ ,
	\end{equation}
	it remains to show $B_A(g\vb x_0) = \tr(A\Pi_{(p,0,0)}^g)$. 	We have that $\vb n_0 = (1,0,0)\in S^5$ satisfies $[\vb n_0] = \vb x_0$ and $\Phi_p(\vb n_0) = (1,0,...,0)$, so
	\begin{equation*}
		B_A(\vb x_0) = \ip{\Phi_p(\vb n_0)}{A\Phi_p(\vb n_0)} = \tr(A\Pi_{(p,0,0)})
	\end{equation*}
	and therefore
	\begin{equation*}
		B_A(g\vb x_0) = (B_A)^{g^{-1}}(\vb x_0) = B_{A^{g^{-1}}}(\vb x_0) = \tr(A^{g^{-1}}\Pi_{(p,0,0)}) = \tr(A\Pi_{(p,0,0)}^g) \ .
	\end{equation*}
	This concludes the proof that $B$ is a symbol correspondence  with operator kernel $\Pi_{(p,0,0)}$. Equation (\ref{b+}) for characteristic numbers follows from (\ref{proj-eq}). Finally, Proposition \ref{prop:mp_corresp_state} implies that $B$ is a mapping-positive correspondence.
\end{proof}

A minor adaptation of the argument of Proposition \ref{ber-prop} shows that $\Pi_{(0,p,p)}\in \mathcal B(\mathcal H_{(0,p)})$, the projection onto the lowest weight space of $(0,p)$, is an operator kernel for a symbol correspondence $\mathcal B(\mathcal H_{(0,p)})\to C^\infty_\mathbb{C}(\mathbb CP^2)$ according to (\ref{sc-ker}). Consider $\widecheck \rho_1$ and $\widecheck \rho_p$ the dual representations of $\rho_1$ and $\rho_p$, respectively, so that 
\begin{equation}
	\theta: \mathbb C^3\to \mathbb C^3:(z_1,z_2,z_3)\mapsto (-\overline{z_3}, \overline{z_2}, -\overline{z_1})
\end{equation}
and $\Phi_p$ are both equivariant maps in the following sense: 
\begin{equation}
	\begin{aligned}
		\widecheck{\rho}_p(g)\circ \theta &= \theta \circ \rho_p(g) \ , \\
		\widecheck{\rho}_p(g)\circ \Phi_p &= \Phi_p\circ \widecheck{\rho}_p(g) \ .
	\end{aligned}
\end{equation}

\begin{proposition}\label{ber-prop2}
	The map $B^-: \mathcal B(\mathcal H_{(0,p)})\to C^\infty_\mathbb{C}(\mathbb CP^2): A\mapsto B^-_A$, with
	\begin{equation*}
		B^-_A(\vb x) = \ip{\Phi_p\circ \theta(\vb n)}{A\Phi_p\circ \theta(\vb n)}
	\end{equation*}
	for $\vb x \in \mathbb CP^2$ and $\vb n \in S^5\subset \mathbb C^3$ related by $[\vb n] = \vb x$, is a mapping-positive symbol correspondence whose operator kernel is the projection $\Pi_{(0,p,p)}$ onto the lowest weight space of $(0,p)$ and whose characteristic numbers are
	\begin{equation}\label{b-}
		b_{n-} = (-1)^p\sqrt{\frac{\dim(\vb*p)}{(n+1)^3}}\cg{(0,p)}{(p,0)}{n}{(0,p,p)}{(p,0,0)}{(n,n,n),0} \ .
	\end{equation}
\end{proposition}
\begin{proof}
	The proof goes just as the proof of Proposition \ref{ber-prop}, we just highlight that the following holds: $\Phi_p\circ\sigma(e^{i\theta}\vb n) = e^{-ip\theta}\Phi_p\circ\sigma(\vb n)$,
	\begin{equation}\label{proj-eq2}
			\Pi_{(0,p,p)}  = \dfrac{1}{\dim(\vb*p)}\mathds 1
            + (-1)^p\sum_{n=1}^p \cg{(0,p)}{(p,0)}{n}{(0,p,p)}{(p,0,0)}{(n,n,n),0} \vb*e(n;(n,n,n),0)
	\end{equation}
	and $\vb n_0 = (1,0,0)\in S^5$ satisfies $\Phi_p\circ \theta(\vb n_0) = (0,...,0,(-1)^p)$.
\end{proof}

So far, practically all results obtained on symbol correspondences for pure-quark systems have analogous results for spin systems,  cf. \cite{RS}. However, the next proposition, extending Remark \ref{rem-H-invariant-proj}, sets an important distinction between symbol correspondences for pure-quark systems and for spin systems.

\begin{proposition}\label{uniquePi}
    Projector  $\Pi_{(p,0,0)}\in \mathcal B(\mathcal H_{(p,0)})$ is the unique projector onto a weight space of $(p,0)$ that is an operator kernel. Likewise, $\Pi_{(0,p,p)}\in \mathcal B(\mathcal H_{(0,p)})$ is the unique projector onto a weight space of $(0,p)$ that is an operator kernel.
\end{proposition}
\begin{proof}
	If $\Pi_{\vb*\nu}$ is an operator kernel for $\vb*p$, from Proposition \ref{op-ker-proj} we get $j=0$, that is, $\vb*e(\vb*p;\vb*\nu) = \vb*e(\vb*p;\vb*\nu,0)$, cf. (\ref{basis-simpl-pq}). Using (\ref{const-2})-(\ref{const-3}), we get that $\vb*\nu = (p,0,0)$ for $\vb*p = (p,0)$ and $\vb*\nu= (0,p,p)$ for $\vb*p = (0,p)$.
\end{proof}

\begin{remark}\label{rem-H-invariant-proj}
	One could expect that  $\Pi_{(0,0,p)}\in \mathcal B(\mathcal H_{(p,0)})$, the projection onto the lowest weight space of $(p,0)$, is also an operator kernel. What the last proposition tells us is that this fails because $\Pi_{(0,0,p)}$ is not $H$-invariant. However, this is just a matter of convention, a choice of $U(2)$ subgroup by which we impose invariance. Invoking (\ref{v_p-basis}), we have $\rho_p(\widecheck\delta)e_{j,k,l} = (-1)^p e_{l,k,j}$ for $\widecheck\delta$ as in (\ref{g0}), so that $\Pi_{(0,0,p)} = \Pi_{(p,0,0)}^{\widecheck\delta}$ and $\Pi_{(0,0,p)}$ is fixed by $\widecheck{H}=\widecheck\delta H\widecheck\delta$, cf. (\ref{u2-u}). If we set 
	\begin{equation}
		\widecheck{\vb x}_0 = \widecheck\delta\vb x_0 = [0:0:1] \in \mathbb CP^2 \ ,    
	\end{equation}
	then we can construct a symbol correspondence $A\mapsto B'_A$ using $\Pi_{(0,0,p)}$ as operator kernel  via the modified rule (compare with (\ref{sc-ker})) given by 
	\begin{equation}
		B'_A(g \widecheck{\vb x}_0) = \tr(A\Pi_{(0,0,p)}^g) \ .
	\end{equation}
	But in fact, $B'$ and $B$ are the same map, that is, 
	\begin{equation}
		B_A'(g\widecheck{\vb x}_0) = \tr(A\Pi_{(0,0,p)}^g) = \tr(A\Pi_{(p,0,0)}^{g\widecheck\delta}) = B_A(g\widecheck\delta\vb x_0) = B_A(g\widecheck{\vb x}_0) \ \ .
	\end{equation}
	In the same vein, using the modified rule 
	\begin{equation}
		{B_A^-}'(g \widecheck{\vb x}_0) = \tr(A\Pi_{(p,p,0)}^g) \ ,
	\end{equation}
	we can identify the highest weight projector of $(0,p)$ as the operator kernel of the symbol correspondence $A\mapsto {B^-_A}'$. But again, ${B^-_A}'\equiv B^-_A$. 
\end{remark}

\begin{definition}\label{st-al-ber-def}
	The symbol correspondences $B: \mathcal B(\mathcal H_{(p,0)})\to C^\infty_\mathbb{C}(\mathbb CP^2)$, with operator kernel $\Pi_{(p,0,0)}$, and $B^-: \mathcal B(\mathcal H_{(0,p)})\to C^\infty_\mathbb{C}(\mathbb CP^2)$, with operator kernel $\Pi_{(0,p,p)}$, are the \emph{symmetric Berezin correspondences} of $(p,0)$ and $(0,p)$, respectively, with each unique dual being the respective \emph{symmetric Toeplitz correspondence}.  
\end{definition}

By Corollary \ref{str-weyl-cn}, the moduli space of Stratonovich-Weyl correspondences for a pure-quark system $\vb*p$ is $(\mathbb Z_2)^p$, with different Stratonovich-Weyl correspondences lying in different connected components of the moduli space $(\mathbb R^\times)^p$ of all correspondences. Thus, there is an unique Stratonovich-Weyl correspondence that can be continuously deformed from the symmetric Berezin correspondence for $\vb*p$. 

\begin{definition}\label{sym-ber-sw-def}
	For a pure-quark system $\vb*p$, the Stratonovich-Weyl correspondence in the same connected component of the symmetric Berezin correspondence is called the \emph{symmetric Stratonovich-Weyl correspondence}.
\end{definition}

\begin{theorem}\label{pq_ber_sw_ant}
	The symmetric Berezin (resp. Stratonovich-Weyl) correspondence for a pure-quark system $\vb*p$ is antipodal to the symmetric Berezin (resp. Stratonovich-Weyl) for $\widecheck{\vb*p}$.
\end{theorem}
\begin{proof}
Since $\Pi_{(p,0,0)} = \Pi_{(0,p,p)}^\ast$, it follows from Propositions \ref{prop:k_ant} and \ref{prop:conex_antip}.
\end{proof}

\begin{remark}
	Assuming symmetries of Theorem \ref{cg-sym-theo}, two symbol correspondences $W_1:\mathcal B(\mathcal H_{\vb*p})\to C^\infty_{\mathbb C}(\mathbb CP^2)$ and $W_2:\mathcal B(\mathcal H_{\widecheck{\vb*p}})\to C^\infty_{\mathbb C}(\mathbb CP^2)$ are antipodal to each other if and only if their characteristic numbers are equal. As expected, these symmetries provide that (\ref{b+}) is equal to (\ref{b-}).
\end{remark}

Twisted products of $\mathbb CP^2$ harmonics can be easily computed and determine the twisted product for all functions in $\mathcal X_p$ by bilinearity of the product.

\begin{theorem}\label{twist-prod-hs}
If $W:\mathcal B(\mathcal H_{\vb*p})\to C^\infty_{\mathbb C}(\mathbb CP^2)$ is a symbol correspondence with characteristic numbers $(c_n)$, then the induced twisted product is given by
\begin{equation}
\begin{aligned}
    X^{n_1}_{\vb*\nu_1, J_1}\star X^{n_2}_{\vb*\nu_2, J_2} & = \sqrt{\dim(\vb*p)}\sum_{n=0}^p\sum_{\vb*\nu,J}(-1)^{p+2(t_{\vb*\nu}+u_{\vb*\nu})}\\
    & \hspace{3 em}\times \begin{bmatrix}
    n_1 & n_2 & n\\
    \vb*\nu_1, J_1 & \vb*\nu_2, J_2 & \widecheck{\vb*\nu}, J
    \end{bmatrix}\!\![\vb* p]\dfrac{c_n}{c_{n_1}c_{n_2}}X^n_{\vb*\nu, J}
\end{aligned}
\end{equation}
for $0\le n_1,n_2\le p$, where summations over $\vb*\nu$ and $J$ are effectively  restricted to $\nabla_{\vb*\nu_1+\vb*\nu_2,\vb*\nu}=1$ and $\delta(J_1,J_2,J)=1$.
\end{theorem}
\begin{proof}
The result follows directly from Theorems \ref{prod-sym} and \ref{symb_corresp-ker}.
\end{proof}

Finally, we have the specific form of Proposition \ref{prop:int_trik} for pure-quark systems:

\begin{corollary}\label{trik-hs-theo}
For a correspondence $W:\mathcal B(\mathcal H_{\vb*p})\to C^\infty_{\mathbb C}(\mathbb CP^2)$ with operator kernel $K$ and characteristic numbers $(c_n)$, the integral trikernel $\mathbb L^{\!W}_{\vb*p}$ is given by
\begin{equation}\label{tri-k-eq}
    \begin{aligned}
        & \mathbb L^{\!W}_{\vb*p}(\vb x_1,\vb x_2,\vb x_3) = \big(\dim(\vb*p)\big)^2\tr(\widetilde K(\vb x_1)\widetilde K(\vb x_2) K(\vb x_3))\\
        & \hspace{2 em}= (-1)^{p}\sqrt{\dim(\vb*p)}\sum_{n_k=0}^p\sum_{\vb*\nu_k,J_k}\begin{bmatrix}
    n_1 & n_2 & n_3\\
    \vb*\nu_1, J_1 & \vb*\nu_2, J_2 & \vb*\nu_3, J_3
    \end{bmatrix}\!\![\vb* p]\\
    & \hspace{10 em}\times \dfrac{c_{n_3}}{c_{n_1}c_{n_2}}\overline{X^{n_1}_{\vb*\nu_1,J_1}}(\vb x_1)\overline{X^{n_2}_{\vb*\nu_2,J_2}}(\vb x_2)\overline{X^{n_3}_{\vb*\nu_3, J_3}}(\vb x_3)
    \end{aligned}
\end{equation}
where $\widetilde K$ is the operator kernel dual to $K$, with characteristic numbers $\widetilde c_n=(1/c_n)$, so that the associated reproducing kernel is
\begin{equation}\label{rk_pq}
	\mathcal R_{\vb*p}^{W}(\vb x_1, \vb x_2) := \sum_{n=0}^p\sum_{\vb*\nu, J}\overline{X^{n}_{\vb*\nu, J}}(\vb x_1)X^n_{\vb*\nu, J}(\vb x_2) = \mathcal R_{\vb*p}^W(\vb x_2, \vb x_1)\, .
\end{equation}
\end{corollary}

\section{Mixed-quark systems}
\label{sec:gen-sys}

In this section, we begin a study of correspondences over the classical mixed-quark systems, that is, representations of generic class $(p,q)$ and coadjoint orbit  $\mathcal E$. Although we proceed by basically reproducing what we have done for $(p,0)$ (or $(0,p)$) and $\mathbb CP^2$, some new phenomena shall appear.

\subsection{Classical mixed-quark system}

\begin{definition}
The \emph{classical mixed-quark system} consists of $\mathcal E$ equipped with a $SU(3)$-invariant symplectic form, together with its Poisson algebra on $C^\infty_{\mathbb C}(\mathcal E)$.
\end{definition}

We have $\mathcal E\simeq SU(3)/T$, where $T$ is the maximal torus (\ref{torus}) of $SU(3)$. So we look for representations with weights satisfying $t = u = 0$, cf. (\ref{const})-(\ref{tuv-nu}).

\begin{proposition}\label{sph-rep-T}
The representations of $SU(3)$ with non null vectors fixed by $T$ are of the form $(a,b)$ for $a\equiv b \pmod 3$. For $$k = |a-b|/3 \ ,$$ the space fixed by $T$ is spanned by the set
\begin{equation}\label{multgamma}
    \{\vb*e((a,b);\vb*0_{(a,b)}, J_{(a,b)}^\gamma): \gamma = 1,...,\min\{a,b\}+1\} \, ,
\end{equation}
 where
\begin{eqnarray}
    \vb*0_{(a,b)} &=& \begin{cases}
    (a+2k,a+2k,a+2k) \  ,  \ \  \mbox{if} \  \ \ \min\{a,b\} = a\\
    (b+k,b+k,b+k) \  ,  \ \  \mbox{if}  \ \ \ \min\{a,b\} = b 
    \end{cases} \label{nu0} \\ 
 J_{(a,b)}^\gamma &=& \gamma -1 + k \ .   
\end{eqnarray}
\end{proposition}

\begin{proof}
Let $\vb*e((a,b);(\nu_1,\nu_2,\nu_3), J)$ be such that $t = u = 0$. From (\ref{const})-(\ref{tuv-nu}), we get that $\nu_1=\nu_2=\nu_3=\nu\in\mathbb N_0$, with 
\begin{equation}\label{mod3}
\begin{aligned}
    2\nu = r_++r_-  \ &, \ \ 3\nu =  a+2b \ , \ \ \mbox{for} \\
0\le r_-\le b\le r_+&\le a+b \ , \ \ r_-\le \nu\le r_+ \ .   
\end{aligned}
\end{equation}
From (\ref{mod3}), $a+2b\equiv 0 \pmod 3$, which implies $a\equiv b \pmod 3$. 

For representations of class $(a, a+3k)$, with $a,k \in \mathbb N_0$, we have that 
\begin{equation}
    \nu=a+2k
\end{equation}
and the GT states fixed by $T$ are given by $r_+$ and $r_-$ satisfying
\begin{equation}
    \begin{cases}
    r_+ + r_- = 2a+4k\\
    0\le r_- \le a+3k \le r_+ \le 2a+3k
    \end{cases} .
\end{equation}
The system has $a+1$ solutions:
\begin{equation}
    \begin{cases}
    r_+ = 2a+3k \ , \ \ r_- = k\\
    \vdots\\ 
    r_+ = a+3k \ , \ \ r_- = a+k 
    \end{cases} , 
\end{equation}
so that the subspace fixed by $T$ is spanned by 
\begin{equation}\label{sp1}
   \{\vb*e((a,a+3k); (a+2k, a+2k, a+2k),J): J = k,...,a+k\} \, .
\end{equation}

For a representation of class $(b+3k,b)$, since it is dual to $(b,b+3k)$, the subspace fixed by $T$ is spanned by
\begin{equation}\label{sp2}
 \{e((b+3k,b);(b+k,b+k,b+k), J): J = k,...,b+k\}   \, .
\end{equation}

To finish, we order the $J$-multiplicities by crescent $J$ in both cases (\ref{sp1})-(\ref{sp2}) by setting $J_\gamma^{(a,b)} = \gamma + k-1$, 
where $1\le\gamma\le \min\{a,b\}+1$.
\end{proof}

\begin{definition}\label{genharm-def}
The \emph{mixed-quark harmonics}, or \emph{$\mathcal E$-harmonics} are the functions on $\mathcal E$ given by 
\begin{equation}\label{genharm-eq}
    Z^{(a,b,\gamma)}_{\vb*\nu, J}(g\vb z_0) = \sqrt{\dim(a,b)}\,\overline{D^{(a,b)}_{\vb*\nu J,\, \vb*0_{(a,b)} J_{(a,b)}^\gamma}}(g) \ , 
\end{equation}
for every $g\in SU(3)$, with $a\equiv b \pmod 3$ and $(\vb*0_{(a,b)},J_{(a,b)}^\gamma)$ as in Proposition \ref{sph-rep-T}.
\end{definition}

Just as in definition of $\mathbb CP^2$-harmonics, the factor multiplying the $D$-function is the square root of the dimension of the representation $\vb*a=(a,b)$, 
\begin{equation*}
   \dim(\vb*a)=\dim(a,b)=\dfrac{(a+1)(b+1)(a+b+2)}{2} \ , 
\end{equation*}
cf. (\ref{dim}), 
and is used to normalize according to Schur's Orthogonality Relations,  
\begin{equation}
\ip{Z^{(a,b,\gamma)}_{\vb*\nu, J}}{Z^{(c,d,\zeta)}_{\vb*\mu, I}} = \delta_{a,c}\delta_{b,d}\delta_{\gamma,\zeta}\delta_{\vb*\nu,\vb*\mu}\delta_{J,I} \ .
\end{equation}

\begin{remark}\label{gen-harm-orbit}
Analogously to Remark \ref{cp2-harm-orbit}, for any $x,y>0$, we can take the generic harmonics as functions on $\mathcal O_{(x,y)}$ via the compositions $Z^{(a,b,\gamma)}_{\vb*\nu, J}\circ \psi_{(x,y)}$ so that the harmonic functions on $\mathcal O_{(x,y)}$ are related to the ones on $\mathcal O_{(y,x)}$ by $\alpha_{(x,y)}\circ \iota$, cf. (\ref{iso-comp-orbit2}). Besides that, the involution $\alpha_{(x,y)}$ generates another, but somewhat equivalent, set of harmonic functions on $\mathcal O_{(x,y)}$, just as $\alpha$ does to $\mathcal E$ (cf. (\ref{iso-comp-orbit1})):
\begin{equation}
    \widetilde Z^{(a,b,\gamma)}_{\vb*\nu, J}(g\vb z_0) = Z^{(a,b,\gamma)}_{\vb*\nu, J}(g\widecheck\delta\vb z_0) \ .
\end{equation}
\end{remark}

As expected, we have 
\begin{equation}\label{Z0=1}
  Z^{(0,0)}_{(0,0,0),0} \equiv 1   
\end{equation}
and, cf. (\ref{conj-wfun}), 
\begin{equation}\label{Zbar}
 \overline{Z^{(a,b,\gamma)}_{\vb*\nu, J}} = (-1)^{2(t+u)}Z^{(b,a,\gamma)}_{\widecheck{\vb*\nu}, J} \ , \  \mbox{for} \ \ \Delta^{a+b}_{\vb*\nu, \widecheck{\vb*\nu}}=1 \ .
\end{equation}

\begin{theorem}\label{e-harm-point-prod}
The pointwise product of $\mathcal E$-harmonics decomposes as
\begin{equation}\label{e-harm-point-prod-eq}
\begin{aligned}
    Z^{(\vb*a_1,\gamma_1)}_{\vb*\nu_1,J_1}Z^{(\vb*a_2,\gamma_2)}_{\vb*\nu_2, J_2} = \sum_{\substack{(\vb*a;\sigma)\\ \vb*\nu, J, \gamma}}&\sqrt{\dfrac{\dim (\vb*a_1)\dim (\vb*a_2)}{\dim (\vb*a)}}\cg{\vb*a_1}{\vb*a_2}{(\vb*a;\sigma)}{\vb*\nu_1 J_1}{\vb*\nu_2 J_2}{\,\,\vb*\nu J}\\ & \times \cg{\vb*a_1}{\vb*a_2}{(\vb*a;\sigma)}{\vb*0_{\vb*a_1} J_{\vb*a_1}^{\gamma_1}}{\vb*0_{\vb*a_2} J_{\vb*a_2}^{\gamma_2}}{\,\vb*0_{\vb*a} J_{\vb*a}^\gamma}Z^{(\vb*a,\gamma)}_{\vb*\nu, J} \ ,
\end{aligned}
\end{equation}
for $(\vb*0_{\vb*a_k}, J_{\vb*a_k}^{\gamma_k})$ and $(\vb*0_{\vb*a}, J_{\vb*a}^\gamma)$ as in Proposition \ref{sph-rep-T}, and summation restricted to $\nabla_{\vb*\nu_1+\vb*\nu_2,\vb*\nu}=\delta(J_1, J_2, J) = \delta(J_{\vb*a_1}^{\gamma_1}, J_{\vb*a_2}^{\gamma_2}, J_{\vb*a}^{\gamma})=1$ and $(\vb*a;\sigma)$ in the Clebsch-Gordan series of $\vb*a_1\otimes \vb*a_2$.
\end{theorem}
\begin{proof}
With a little abuse of notation, again,
\begin{equation}
   Z^{(\vb*a_k,\gamma_k)}_{\vb*\nu_k,J_k} = \sqrt{\dim (\vb*a_k)}\ \overline{D^{\vb*a_k}_{\vb*\nu_k J_k,\,\vb*0_{\vb*a_k} J_{\vb*a_k}^{\gamma_k}}} 
\end{equation}
and Lemma \ref{wig-prod} give us
\begin{equation*}
\begin{aligned}
     Z^{(\vb*a_1,\gamma_1)}_{\vb*\nu_1,J_1}Z^{(\vb*a_2,\gamma_2)}_{\vb*\nu_2, J_2} = \sum_{(\vb*a;\sigma)}\sum_{\substack{\vb*\nu, J\\\vb*\mu, L}}&\sqrt{\dim (\vb*a_1)\dim (\vb*a_2)}\cg{\vb*a_1}{\vb*a_2}{(\vb*a;\sigma)}{\vb*\nu_1 J_1}{\vb*\nu_2 J_2}{\,\,\vb*\nu J}\\
     & \times \cg{\vb*a_1}{\vb*a_2}{(\vb*a;\sigma)}{\vb*0_{\vb*a_1} J_{\vb*a_1}^{\gamma_1}}{\vb*0_{\vb*a_2} J_{\vb*a_2}^{\gamma_2}}{\vb*\mu I}\overline{D^{\vb*a}_{\vb*\nu J,\vb*\mu I}} \ ,
\end{aligned}
\end{equation*}
where $\nabla_{\vb*\nu_1+\vb*\nu_2, \vb*\nu}=\nabla_{\vb*\nu_{\vb*a_1}+\vb*\nu_{\vb*a_2},\vb*\mu} =1$ and $\delta(J_1,J_2,J)=\delta(J_{\vb*a_1}^{\gamma_1}, J_{\vb*a_2}^{\gamma_2}, I)=1$, so $\vb*\mu = (\mu,\mu,\mu)$. But $\vb*e(\vb*a;(\mu,\mu,\mu), I)$ only exists if $\vb*a$ and $(\vb*\mu, I)$ are as in Proposition \ref{sph-rep-T}. Thus, we set $\vb*\mu = \vb*0_{\vb*a}$ and $I = J_{\vb*a}^{\gamma}$.
\end{proof}

\subsection{Quantum mixed-quark system}

Now, we want representations $(p,q)$ such that $(p,q)\otimes (q,p)$ splits only into representations of the form $(a,b)$, $a\equiv b\pmod 3$, with multiplicity less than or equal to $\min\{a,b\}+1$. From Corollary \ref{cg-series-corol}, if we suppose, without loss of generality, that $\min\{a,b\} = a$, then the occurrences of $(a,b)\oplus (b,a)$ are given by the solutions of
\begin{equation}
    a = p+q-n-m-2k\ , \ \ 
    0 \le n \le p-k\ , \ \ 
    0 \le m \le q-k\ , 
\end{equation}
where $b = a+3k$. Of course, we can also assume without loss of generality that $p\ge q$. If $a+k\le q$, then we have $a+1$ solutions:
\begin{equation}
    \begin{cases}
    n = p-k-a \ , \ \ m = q-k\\
    n = p-k-a+1 \ , \ \ m = q-k-1\\
    \vdots\\
    n = p-k \ , \ \ m = q-a-k\\
    \end{cases} \ .
\end{equation}
Otherwise, we need to eliminate some lines of the above solutions, which means $(a,b)\oplus (b, a)$ have multiplicity less than $\min\{a,b\}+1$. Then, we have:

\begin{definition}\label{mat-def}
Let $(p,q)\in (\mathbb N\times \mathbb N_0)\cup (\mathbb N_0\times \mathbb N)$. A \emph{general quantum quark system} is a complex Hilbert space $\mathcal H_{(p,q)}\simeq \mathbb C^d$, where $d=\dim (p,q)$ is given by (\ref{dim}), with an irreducible unitary $SU(3)$-representation of class $(p,q)$ together with its operator algebra $\mathcal B(\mathcal H_{p,q})$. If $(p,q)\in \mathbb N\times\mathbb N$, we have a \emph{quantum mixed-quark system}, and in the particular case when $p=q$ the mixed-quark system is \emph{mesonic}. 
\end{definition}

\begin{remark}\label{material-rem}
The last name in Definition \ref{mat-def} refers to the number of quarks vs.~antiquarks for a general quark system, cf. Appendix \ref{sec:def_PQ}. From Theorem \ref{cg-s},
\begin{equation}
    (p,0)\otimes (0,q) = \bigoplus_{n=0}^{\min\{p,q\}}(p-n,q-n) \ ,
\end{equation}
so a general representation $(p,q)$ is the invariant space of $(p,0)\otimes (0,q)$ where the product of the highest weight vectors lives in. In Physics, a system of of $p$ quarks and $q$ antiquarks has baryon number $B = (p-q)/3$. If $B>0$, the system is a baryon (resp.~antibaryon if $B<0$) and if $B = 0$, it is a meson (if $B\geq 0$ it is a material system, antimaterial otherwise). But from our perspective of symbol correspondences, the relevance of mesonic systems stems from the adjoint orbit $\mathcal O_{(x,x)}\subset\mathfrak{su}(3)$ being invariant by  involution $\iota=-id:\mathfrak{su}(3)\to\mathfrak{su}(3)$, cf. (\ref{iota-def}).
\end{remark}

In particular, general quantum quark systems encompass quantum pure-quark systems as special cases (but now, we cannot further simplify Theorem \ref{op-prod-gen-s} as we did for pure-quark systems). However, since we have already studied pure-quark systems, for general quark systems we shall mostly focus on  mixed-quark systems. 

\subsection{Correspondences over mixed-quark systems}

For a quantum mixed-quark system,  $\vb*p\in \mathbb N\times \mathbb N$, the possible multiplicities on the CG series of $\vb*p\otimes \widecheck{\vb*p}$ forces us to fix some level of convention for the decomposition of degenerate representations in order to get a characterization somewhat similar to the one using characteristic numbers for pure-quark systems. We assume specifically \eqref{hce}.

\begin{notation} 
From now on, we shall use the notation
\begin{equation}
  m(\vb*a)=  m(a,b)= \min\{a,b\}+1 \, . 
\end{equation}
and simplify the notation $\textswab m(\vb*p,\widecheck{\vb*p};\vb*a)$ for the multiplicity of $\vb*a$ in the Clebsch-Gordan series of $\vb*p\otimes \widecheck{\vb*p}$ by setting
\begin{equation}
    \textswab{m}(\vb*p;\vb*a):= \textswab{m}(\vb*p,\widecheck{\vb*p};\vb*a)\, . 
\end{equation}
Finally, we set
\begin{equation}\label{Bpa} 
    \mathcal B(\vb*p;\vb*a) = \bigoplus_{\sigma = 1}^{\textswab{m}(\vb*p;\vb*a)}(\vb*a;\sigma) \subset \mathcal B(\mathcal H_{\vb*p}) \ .
\end{equation}
\end{notation}

Henceforth, unless otherwise stated, we are considering only symbol correspondence over the mixed classical quark system $\mathcal E$, which includes general quantum quark systems $\vb*p$. As usual, for an operator $K\in \mathcal B(\mathcal H_{\vb*p})$ fixed by $T$, let 
\begin{equation}
	\mathcal E\to \mathcal B(\mathcal H_{\vb*p}): g\vb z_0 \mapsto K(\vb z) = K(g\vb z_0) = K^g\, .
\end{equation}

\begin{theorem}\label{op-ker-gen}
	A map $W: \mathcal B(\mathcal H_{\vb*p})\to C^{\infty}_{\mathbb C}(\mathcal E): A\mapsto W_A$ is a symbol correspondence if and only if
	\begin{equation}\label{sc-ker-g}
		W_A(\vb z) = \tr(AK(\vb z))
	\end{equation}
	for $K\in \mathcal B(\mathcal H_{\vb*p})$, its operator kernel, of the form
	\begin{equation}\label{op-ker-eq-g}
		K = \dfrac{1}{\dim (\vb*p)}\mathds{1}+\sum_{\substack{(\vb*a;\sigma)}}\sum_{\gamma=1}^{m(\vb*a)}\overline{c^{\sigma}_{\gamma}({\vb*a})}\sqrt{\dfrac{\dim (\vb*a)}{\dim (\vb*p)}}\,\vb*e((\vb*a;\sigma);\vb*\nu_{\vb*a}, J_{\gamma}^{\vb*a}) \, , 
	\end{equation}
	with $c^{\sigma}_\gamma\!(\vb*a)$ being the $\gamma\times \sigma$ entry of a complex full rank matrix $\mathbf{C}(\vb*a)$ of order $m(\vb*a)\times \textswab{m}(\vb*p;\vb*a)$ satisfying $\overline{\mathbf{C}(\vb*a)} = \mathbf{C}(\widecheck{\vb*a})$ and $\mathbf{C}(0,0) = 1$.  In particular,
	\begin{equation}\label{sc-gen-eq}
		W: \sqrt{\dim (\vb*p)}\, \vb*e((\vb*a;\sigma); \vb*\nu, J) \mapsto \sum_{\gamma=1}^{m(\vb*a)}c^\sigma_{\gamma}(\vb*a)Z^{(\vb*a,\gamma)}_{\vb*\nu, J} =: Z\mathbf{C}(\vb*a)^{\sigma}_{\vb*\nu, J} \, . 
	\end{equation}
\end{theorem}
\begin{proof}
	Analogously to Theorem \ref{symb_corresp-ker}, assume $W$ is a symbol correspondence. Proposition \ref{sph-rep-T} and \eqref{hce} implies
	\begin{equation}\label{coef-K}
		K = \sum_{\substack{(\vb*a;\sigma)\\ \gamma}}k^{(\vb*a;\sigma)}_{\gamma}\vb* e((\vb*a;\sigma);\vb*\nu_{\vb*a}, J_{\gamma}^{\vb*a}) \, ,
	\end{equation}
	with $\overline{k^{(\vb*a;\sigma)}_\gamma} = k^{(\widecheck{\vb*a};\sigma)}_\gamma$. For	$A = \vb*e((\vb*a;\sigma); \vb*\nu, J)$ we get
	\begin{equation}
		\begin{aligned}
			\tr(AK(\vb z)) = \tr(AK(\vb z)^\dagger) = \sum_{\gamma=1}^{m(\vb*a)}k^{(\widecheck{\vb*a};\sigma)}_{\gamma}\overline{D^{\vb*a}_{\vb*\nu J, \vb*\nu_{\vb*a} J_{\gamma}^{\vb*a}}} \, .
		\end{aligned}
	\end{equation}
	Taking
	\begin{equation}\label{k-value}
		k^{(\widecheck{\vb*a};\sigma)}_{\gamma} = c^{\sigma}_{\gamma}(\vb*a)\sqrt{\dfrac{\dim (\vb*a)}{\dim (\vb*p)}} \, ,
	\end{equation}
	injectivity of $W$ imposes that the matrix $\mathbf{C}(\vb*a)$ with entries $c^\sigma_\gamma(\vb*a)$ is full rank.
	
	The converse is again analogous to Theorem \ref{symb_corresp-ker}, and it is rather straightforward to verify that, for $K$ given by (\ref{op-ker-eq-g}), equations \eqref{tr_equiv} and \eqref{tr_real}, plus the hypothesis of $\mathbf{C}(\vb*a)$ and the fact that $\tr(K) = 1$, imply that \eqref{sc-ker-g} defines a symbol correspondence given by (\ref{sc-gen-eq}).
\end{proof}

\begin{remark}\label{trivialift}
    If $\vb*p=(p,0)$ or $(0,p)$, we can still consider a pure-quark symbol correspondence $W:\mathcal B(\mathcal H_{\vb*p})\to C^{\infty}_{\mathbb C}(\mathbb CP^2)$ as a map $W:\mathcal B(\mathcal H_{\vb*p})\to C^{\infty}_{\mathbb C}(\mathcal E)$, where now every symbol $W_{\!A}$ is a constant function along the fibers of 
    \begin{equation*}
        \mathbb CP^1\to\mathcal E\to\mathbb CP^2 \ ,  
    \end{equation*}
    so that it is a trivial lift of a pure-quark symbol $W_{\!A}\in C^\infty_{\mathbb C}(\mathbb CP^2)$, as studied before. 
\end{remark}

Taking Remark \ref{trivialift} into consideration, cf. also Remark \ref{nontrivialift} below, we have: 

\begin{corollary}\label{moduli-space-gq}
The moduli space $\mathcal M_{\vb*p}(\mathcal E)$ of correspondences for a general quantum quark system $\vb*p$ is
\begin{equation}\label{Vgen}
   \left(\prod_{a=0}^{|\vb*p|} V_{\textswab{m}(\vb*p;a,a)}(\mathbb R^{a+1})\right)\times \left(\prod_{a<b} V_{\textswab{m}(\vb*p;a,b)}(\mathbb C^{a+1})\right) \ ,
\end{equation}
where $V_k(\mathbb K^n) = GL_n(\mathbb K)/GL_{n,k}(\mathbb K)$, for $GL_{n,k}(\mathbb K)\subset GL_n(\mathbb K)$ a maximal subgroup that fixes a $k$-dimensional subspace, is a non compact Stiefel manifold.
\end{corollary}

\begin{definition}
The numbers $(c^{\sigma}_{\gamma}(\vb*a))$ and the matrices $\mathbf{C}(\vb*a)$ of Theorem \ref{op-ker-gen} are called, respectively, the \emph{characteristic parameters} and the \emph{characteristic matrices} of both the correspondence and the operator kernel.
\end{definition}

\begin{remark}
   The matrices $\mathbf{C}(\vb*a)$ are matrix representations of the map $W$ restricted to a weight space of $\mathcal B(\vb*p;\vb*a)$ with respect to a coupled basis of $\mathcal B(\mathcal H_{\vb*p})$ and the $\mathcal E$ harmonics. They are analogous to characteristic numbers of symbol correspondences for pure-quark system: in the latter case, the domain and codomain of a symbol correspondence are multiplicity free and have only representations $(n,n)$, so it provides a $1\times 1$ real matrix indexed by $n$, so that, if a pure-quark correspondence has characteristic numbers $(c_n)$, its characteristic parameters would be $c^{1}_{\gamma}(n,n) = c_n\delta_{\gamma,1}$ and the moduli space, in this case, would be a product of $V_{1}(\mathbb R) = \mathbb R^\times$. 
\end{remark}

Characteristic matrices encode all the information about symbol correspondences for mixed-quark systems in the same vein of characteristic numbers for pure-quark systems. Thus, the existence of irreps $\vb*a$ with higher degeneracy within $C^\infty_{\mathbb C}(\mathcal E)$ than within $\mathcal B(\mathcal H_{\vb*p})$ are reflected by $\mathbf{C}(\vb*a)$ having more lines than columns. For instance, since
\begin{equation}\label{mult_B_C}
	\textswab{m}(\vb*p;|\vb*p|,|\vb*p|) = 1 \ , \ \  m(|\vb*p|,|\vb*p|) = |\vb*p|+1>1
\end{equation}
for any non trivial $\vb*p$, the characteristic matrix $\mathbf{C}(|\vb*p|,|\vb*p|)$ of any symbol correspondences for $\vb*p$ is a column vector. Such observation leads us to the of existence symbol correspondences for the same $\vb*p$ with different image sets.

\begin{lemma}\label{WneqW}
	For any general quantum quark system $\vb*p$, there exist symbol correspondences $W_1,W_2:\mathcal B(\mathcal H_{\vb*p})\to C^\infty_{\mathbb C}(\mathcal E)$ such that $\mathcal S_{\vb*p}(W_1) \ne \mathcal S_{\vb*p}(W_2)$. 
\end{lemma}
\begin{proof}
	From \eqref{mult_B_C}, we can take
	\begin{equation}\label{W(|p|,|p|)}
		\begin{aligned}
			& W_1(\mathcal B(\vb*p;|\vb*p|,|\vb*p|)) = \operatorname{span}\left\{
			Z\mathbf{C}[1](|\vb*p|,|\vb*p|)_{\vb*\nu, J}\right\} \ , \\ 
			& W_2(\mathcal B(\vb*p;|\vb*p|,|\vb*p|)) = \operatorname{span}\left\{
			Z\mathbf{C}[2](|\vb*p|,|\vb*p|)_{\vb*\nu, J}\right\} \ , 
		\end{aligned}    
	\end{equation}
	to be orthogonal subspaces of $C^\infty_{\mathbb C}(\mathcal E)$. Indeed, for $k=1,2$, let $W_k$ have characteristic matrices $\mathbf{C}[k](\vb*a)$ with
	characteristic parameters $c[k]^{\sigma}_{\gamma}(\vb*a)$. We can drop the index $\sigma$ for the characteristic parameters $c[k]_{\gamma}(|\vb*p|,|\vb*p|)$. Hence, by choosing
	\begin{equation}
		\sum_{\gamma = 1}^{|\vb*p|+1}c[1]_{\gamma}(|\vb*p|,|\vb*p|)c[2]_{\gamma}(|\vb*p|,|\vb*p|) = 0 \, , 
	\end{equation}
	we get $W_1(\mathcal B(\vb*p;|\vb*p|,|\vb*p|))\perp W_2(\mathcal B(\vb*p;|\vb*p|,|\vb*p|))$, so $\mathcal S_{\vb*p}(W_1) \ne \mathcal S_{\vb*p}(W_2)$.
\end{proof}

\begin{proposition}\label{contc-g-prop}
Let $K\in \mathcal B(\mathcal H_{\vb*p})$ be an operator kernel with characteristic matrices $\mathbf{C}(\vb*a)$.
A symbol correspondence $\widetilde W:\mathcal B(\mathcal H_{\vb*p})\to C^\infty_{\mathbb C} (\mathcal E)$ satisfies
\begin{equation}\label{cont-corresp-g}
    A = \dim (\vb*p)\int_{\mathcal E}\widetilde W_A(\vb z)K(\vb z) d\vb z
\end{equation}
if and only if it has characteristic matrices $\widetilde{\mathbf C}(\vb*a)$ such that $(\widetilde{\mathbf{C}}(\vb*a))^\dagger \mathbf{C}(\vb*a) = \mathds{1}$.
\end{proposition}
\begin{proof}
Reproducing what we did in Proposition \ref{dual-pquark-prop}, we get
\begin{equation}
    \begin{aligned}
        \int_{\mathcal E}Z^{(\vb*a,\gamma)}_{\vb*\nu, J}(\vb z)K(\vb z) d\vb z & = \!\!\sum_{\substack{(\vb*a';\sigma')\\\gamma'\\\vb*\mu, L}} \!\dfrac{k^{(\vb*a';\sigma')}_{\gamma'}}{\sqrt{\dim (\vb*a')}}\!\ip{Z^{(\vb*a',\gamma')}_{\vb*\mu, L}}{Z^{(\vb*a,\gamma)}_{\vb*\nu, J}}\vb*e((\vb*a';\sigma');\vb*\mu, L)\\
        & = \!\!\!\sum_{\sigma'=1}^{\textswab{m}(\vb*p;\vb*a)} \dfrac{k^{(\vb*a;\sigma')}_{\gamma}\vb*e((\vb*a;\sigma'); \vb*\nu, J)}{\sqrt{\dim (\vb*a)}} 
    \end{aligned}
\end{equation}
where $k^{(\vb*a';\sigma')}_{\gamma'}$ is given by (\ref{k-value}). So, for $c^{\sigma}_{\gamma}(\vb*a)$ and $\widetilde c^{\sigma}_{\gamma}(\vb*a)$ being the characteristic parameters of $\mathbf{C}(\vb*a)$ and $\widetilde{\mathbf{C}}(\vb*a)$, respectively, we have
\begin{equation}
    \begin{aligned}
        \dim (\vb*p)\int_{\mathcal E} &
        Z\widetilde{\mathbf{C}}(\vb*a)^{\sigma}_{\vb*\nu, J}(\vb z)K(\vb z) d\vb z\\
        & = \sqrt{\dim (\vb*p)}\sum_{\gamma,\sigma'} \widetilde c^{\sigma}_{\gamma}(\vb*a)\overline{c^{\sigma'}_{\gamma}(\vb*a)}\, \vb* e((\vb*a;\sigma');\vb*\nu, J) \ .
    \end{aligned}
\end{equation}
Hence, (\ref{cont-corresp-g}) holds if and only if
\begin{equation}
    \sum_{\gamma=1}^{m(\vb*a)} \widetilde c^{\sigma}_{\gamma}(\vb*a)\overline{c^{\sigma'}_{\gamma}(\vb*a)} = \delta_{\sigma,\sigma'} \ ,
\end{equation}
which means $(\mathbf{C}(\vb*a))^\dagger\widetilde{\mathbf{C}}(\vb*a) = \mathds 1$, or equivalently $(\widetilde{\mathbf{C}}(\vb*a))^\dagger \mathbf{C}(\vb*a) = \mathds 1$.
\end{proof}

\begin{theorem}\label{dual-prop-g}
For any symbol correspondence $\mathcal B(\mathcal H_{\vb*p})\to C^\infty_{\mathbb C}(\mathcal E)$ with characteristic matrices $\mathbf{C}(\vb*a)$, there are multiple dual correspondences, that is, there are multiple characteristic matrices $\widetilde{\mathbf C}(\vb*a)$ satisfying
\begin{equation}\label{char_mat_dual}
    (\widetilde{\mathbf{C}}(\vb*a))^\dagger \mathbf{C}(\vb*a) = \mathds{1} \ .
\end{equation}
Such a dual correspondence is the canonical dual iff, in addition to \eqref{char_mat_dual}, the linear span of the columns of $\widetilde{\mathbf{C}}(\vb*a)$ equals the linear span of the columns of ${\mathbf{C}}(\vb*a)$.
\end{theorem}
\begin{proof}
The characterization of duality is a direct consequence of Propositions \ref{prop:dual_int} and \ref{contc-g-prop}. The claim about multiplicity follows from the fact that there is $\vb*a$ such that $\mathbf{C}(\vb*a)$ has multiple left inverse and, if the matrix $(\widetilde{\mathbf C}(\vb*a))^\dagger$ is such a left inverse, then $\widetilde{\mathbf C}(\vb*a)$ is full rank and the conjugate of $(\widetilde{\mathbf C}(\vb*a))^\dagger$ is a left inverse of $\overline{\mathbf{C}(\vb*a)}=\mathbf{C}(\widecheck{\vb*a})$. The characterization of canonical duality is just the additional hypothesis that the symbol correspondences have the same images, cf. Definition \ref{canndual}.
\end{proof}
\begin{corollary}
A symbol correspondence is Stratonovich-Weyl  iff its characteristic matrices are semi-unitary matrices, that is, they satisfy $(\mathbf{C}(\vb*a))^\dagger\mathbf{C}(\vb*a) = \mathds 1$.	
\end{corollary}

In addition to the special cases of isometric correspondences, now we also have correspondences given by a direct sum of conformal maps.

\begin{definition}
If a symbol correspondence $W:\mathcal B(\mathcal H_{\vb*p})\to C^\infty_{\mathbb C}(\mathcal E)$ preserves angles for each $\mathcal B(\vb*p;\vb*a)$, that is,
\begin{equation}
    \dfrac{\ip{A_1}{A_2}_{\vb*p}}{\norm{A_1}_{\vb*p}\norm{A_2}_{\vb*p}} = \dfrac{\ip{W_{A_1}}{W_{A_2}}}{\norm{W_{A_1}}\norm{W_{A_2}}}
\end{equation}
for all non null $A_1,A_2\in \mathcal B(\vb*p;\vb*a)$ and every $\mathcal B(\vb*p;\vb*a)\subset \mathcal B(\mathcal H_{\vb*p})$, then $W$ shall be called a \emph{semi-conformal correspondence}.
\end{definition}

\begin{corollary}\label{st-conf-prop}
A symbol correspondence is semi-conformal if and only if its characteristic matrices are semi-conformal matrices, that is, $(\mathbf{C}(\vb*a))^\dagger \mathbf{C}(\vb*a) = \alpha(\vb*a)\, \mathds 1$ for $\alpha(\vb*a) > 0$, where $\alpha(\vb*a) = \alpha(\widecheck{\vb*a})$ and $\alpha(0,0) = 1$.
\end{corollary}

\begin{corollary}
	The canonical dual of a semi-conformal correspondence with characteristic matrices $\mathbf C(\vb*a)$ satisfying $(\mathbf C(\vb*a))^\dagger \mathbf C(\vb*a) = \alpha(\vb*a)\, \mathds 1$ has characteristic matrices
	\begin{equation}
		\widetilde{\mathbf C}(\vb*a) = \dfrac{1}{\alpha(\vb*a)}\mathbf C(\vb*a) \ .
	\end{equation}
\end{corollary}

\begin{remark}\label{semiconfremark}
A symbol correspondence $W$ is an actual conformal map if and only if $W = \sqrt{\alpha}\, W'$ for $\alpha>0$ and some Stratonovich-Weyl correspondence $W'$. Since $W_{\mathds 1} = W'_{\mathds 1}$, we must have $\alpha = 1$, so the only actual conformal correspondences are the isometric ones. For pure-quark systems (likewise for spin systems), every symbol correspondence is  semi-conformal, with $\alpha(\vb*a)=\alpha(n,n)=c_n^2$. 
\end{remark}

\begin{proposition}
	A symbol correspondence is a semi-conformal correspondence if and only if its antipodal correspondence is also a semi-conformal correspondence.
\end{proposition}
\begin{proof}
	A symbol correspondence $W:\mathcal B(\mathcal H_{\vb*p})\to C^\infty_{\mathbb C}(\mathcal E)$ is semi-conformal if and only if there are $\alpha(\vb*a)>0$ such that
	\begin{equation}
		W' = \bigoplus_{B(\vb*p;\vb*a)}\alpha(\vb*a)^{-1/2}W|_{\mathcal B(\vb*p;\vb*a)}
	\end{equation}
	is a Stratonovich-Weyl correspondence, so the result is just an application of Proposition \ref{prop:ant_dual}.
\end{proof}

 Recall that each of the projectors $\Pi_{(p,0,0)}\in \mathcal B(\mathcal H_{(p,0)})$ or $\Pi_{(0,q,q)}\in \mathcal B(\mathcal H_{(0,q)})$ is an operator kernel of a mapping-positive correspondence $\mathcal B(\mathcal H_{(p,0)})\to C^\infty_\mathbb{C}(\mathcal E)$ or $\mathcal B(\mathcal H_{(0,q)})\to C^\infty_\mathbb{C}(\mathcal E)$, respectively, with respective symbols on $\mathcal E$ being constant extensions of functions on $\mathbb CP^2$, cf. Remark \ref{trivialift}. 

\begin{remark}\label{nontrivialift}
From Remark \ref{rem-H-invariant-proj}, the only impediment for $\Pi_{(0,0,p)}\in \mathcal B(\mathcal H_{(p,0)})$ and $\Pi_{(q,q,0)}\in \mathcal B(\mathcal H_{(0,q)})$ to define symbol correspondences for pure-quark systems is the lack of $H\simeq U(2)$ invariance. But these operators are invariant by the torus $T$, so each of these projectors is also an operator kernel for a general quark system, defining a mapping-positive symbol correspondence from a quantum pure-quark system $\mathcal H_{\vb*p}$, for $\vb*p=(p,0)$ or $(0,q)$, to the classical mixed-quark system defined on $\mathcal E$, now with corresponding symbols no longer  constant along the fibers of $\mathcal E$, in general (so these could also be characterized as mixed-quark correspondences).
\end{remark}

In respect to the previous remark, we have the more general fact below:

\begin{theorem}\label{ber-gen-prop}
For any general quark system $\vb*p$, the projector onto the highest weight space and onto the lowest weight space, $\Pi_>, \Pi_<\in\mathcal B(\mathcal H_{\vb*p})$, are each one an operator kernel for a mapping-positive symbol correspondence $W:\mathcal B(\mathcal H_{\vb*p})\to C^\infty_{\mathbb C}(\mathcal E)$.
\end{theorem}

In Appendix \ref{sec:ber_ker}, we present a detailed proof of the above theorem, by specializing to $SU(3)$ the main argument in  \cite{figueroa, wild} for general compact semisimple Lie groups.

\begin{definition}\label{ber-def-gen}
A mapping-positive correspondence whose operator kernel is $\Pi_>$ shall be called the \emph{highest Berezin correspondence}. Likewise for the \emph{lowest Berezin correspondence} in the case of $\Pi_<$.
\end{definition}

\begin{theorem}\label{ber-gen-cor}
For any general quark system $\vb*p = (p,q)$, the highest Berezin correspondence is antipodal to the lowest Berezin correspondence, and their characteristic parameters are, respectively,
\begin{equation}\label{b>}
    (b_>)^\sigma_\gamma(\vb*a) = (-1)^{|\vb*p|}\sqrt{\dfrac{\dim (\vb*p)}{\dim (\vb*a)}} \cg{\ \vb*p}{\ \widecheck{\vb*p}}{(\vb*a;\sigma)}{(p+q,q,0) q/2}{(0,p,p+q) q/2}{\vb*\nu_{\vb*a} J_\gamma^{\vb*a}} \, ,
\end{equation}
\begin{equation}\label{b<}
    (b_<)^\sigma_\gamma(\vb*a) = (-1)^{|\vb*p|}\sqrt{\dfrac{\dim (\vb*p)}{\dim (\vb*a)}} \cg{\ \vb*p}{\ \widecheck{\vb*p}}{(\vb*a;\sigma)}{(0,q,p+q) p/2}{(p+q,p,0) p/2}{\vb*\nu_{\vb*a} J_\gamma^{\vb*a}} \, .
\end{equation}
\end{theorem}
\begin{proof}
The first statement follows straightforwardly from Proposition \ref{prop:k_ant}. The characteristic parameters of $\Pi_>$ and $\Pi_<$ can be obtained by expanding them in the uncoupled basis, just as in Propositions \ref{ber-prop} and \ref{ber-prop2}. 
\end{proof}

For pure-quark systems, Proposition \ref{uniquePi} asserts there is only one projector that defines a (symmetric)  symbol correspondence, and for general quark systems Theorem \ref{ber-gen-prop} asserts that the highest and lowest projectors define symbol correspondences for every $(p,q)$. However, given a representation of class $(p,q)$, we don't know which other projectors can define symbol correspondences. Also, we still don't have explicit examples of mapping-positive correspondences for every $(p,q)$, other than the highest and the lowest Berezin correspondences.

\begin{remark}
As in Definition \ref{sym-ber-sw-def}, one could expect to define highest and lowest Stratonovich-Weyl correspondences via continuous deformations from the highest and lowest Berezin correspondences. There are, however, infinite Stratonovich-Weyl correspondences connected via continuous deformation from either one, so this can not be done unambiguously. To see this, consider a Berezin correspondence with characteristic matrices $\mathbf C(\vb*a)$. By continuous application of Gram-Schmidt process on the columns of $\mathbf C(\vb*a)$, we obtain a semi-unitary matrix. Then, if we apply any rotation to the columns of this semi-unitary matrix, it remains semi-unitary.
\end{remark}

For the characterization of antipodal correspondences, Theorem \ref{cg-sym-theo} has the following corollary.

\begin{corollary}\label{ch-par-alt}
The symbol correspondence $\widecheck{W}:\mathcal B(\mathcal H_{\widecheck{\vb*p}})\to C^\infty_{\mathbb C}(\mathcal E)$ with characteristic parameters $(\widecheck{c}^{\sigma}_{\gamma}(\vb*a))$ is antipodal to the symbol correspondence $W:\mathcal B(\mathcal H_{\vb*p})\to C^\infty_{\mathbb C}(\mathcal E)$ with characteristic parameters $(c^{\sigma}_{\gamma}(\vb*a))$ if and only if
\begin{equation}\label{cac}
    c^\sigma_\gamma(\vb*a) = (-1)^{|\vb*a|}\widecheck{c}^{\,\widecheck\sigma}_{\gamma}(\vb*a) \ .
\end{equation}
\end{corollary}

Now, twisted products can be explicitly obtained from Theorems \ref{op-prod-gen-s} and \ref{op-ker-gen}:

\begin{theorem}\label{twist-prod-gen}
If $W:\mathcal B(\mathcal H_{\vb*p})\to C^\infty_{\mathbb C}(\mathcal E)$ is a symbol correspondence with characteristic matrices $\mathbf C(\vb*a)$, then the induced twisted product $\star$ is given by
\begin{equation}
\begin{aligned}
    &
    Z\mathbf{C}(\vb*a_1)^{\sigma_1}_{\vb*\nu_1, J_1} \star
    Z\mathbf{C}(\vb*a_2)^{\sigma_2}_{\vb*\nu_2, J_2}
    \\ & \ \ \ \ \ = \sqrt{\dim (\vb*p)}\sum_{\substack{(\vb* a;\sigma)\\\vb*\nu, J}}
    (-1)^{|\vb*p|+2(t_{\vb*\nu}+u_{\vb*\nu})} \begin{bmatrix}
    (\vb*a_1;\sigma_1) & (\vb*a_2;\sigma_2) & (\vb*a;\sigma)\\
    \vb*\nu_1, J_1 & \vb*\nu_2, J_2 & \widecheck{\vb*\nu}, J
    \end{bmatrix}\!\![\vb*p]\,
    Z\mathbf{C}(\vb*a)^{\sigma}_{\vb*\nu, J}\ ,
\end{aligned}
\end{equation}
cf. (\ref{sc-gen-eq}), with summations over $\vb*\nu$ and $J$ effectively restricted by (\ref{prod-s-ne-0}).
\end{theorem}

\begin{corollary}\label{trik-gen-prop}
For a correspondence $W:\mathcal B(\mathcal H_{\vb*p})\to C^\infty_{\mathbb C}(\mathcal E)$ with operator kernel $K$ and characteristic matrices $\mathbf C(\vb*a)$, the integral trikernel is of the form
\begin{equation}\label{tri-k-gen-eq}
    \begin{aligned}
        \mathbb L_{\vb*p}^{\!W}&(\vb z_1,\vb z_2,\vb z_3) = \left(\dim (\vb*p)\right)^2\tr(\widetilde K(\vb z_1)\widetilde K(\vb z_2) K(\vb z_3))\\
        & = (-1)^{|\vb*p|}\sqrt{\dim (\vb*p)}\sum_{\substack{(\vb* a_k;\sigma_k)\\\vb*\nu_k, J_k}} \begin{bmatrix}
    (\vb*a_1;\sigma_1) & (\vb*a_2;\sigma_2) & (\vb*a_3;\sigma_3)\\
    \vb*\nu_1, J_1 & \vb*\nu_2, J_2 & \vb*\nu_3, J_3
    \end{bmatrix}\!\![\vb*p]\\
    & \hspace{5 em} \times\overline{Z\widetilde{\mathbf{C}}(\vb*a_1)^{\sigma_1}_{\vb*\nu_1, J_1}}(\vb z_1)\,\overline{Z\widetilde{\mathbf{C}}(\vb*a_2)^{\sigma_2}_{\vb*\nu_2, J_2}}(\vb z_2)\,\overline{Z\mathbf{C}(\widecheck{\vb*a}_3)^{\sigma_3}_{\vb*\nu_3, J_3}}(\vb z_3)
    \end{aligned}
\end{equation}
cf. Proposition \ref{prop:int_trik}, where $\widetilde{\mathbf C}(\vb*a)$ are the characteristic matrices of the operator kernel $\widetilde K$ canonically dual to $K$, so that the associated reproducing kernel is
\begin{equation}
	\mathcal R^{\!W}_{\vb*p}(\vb z_1, \vb z_2) = \sum_{\substack{(\vb*a;\sigma)\\\vb*\nu, J}}Z\widetilde{\mathbf{C}}(\vb*a)^{\sigma}_{\vb*\nu, J}(\vb z_1)\,Z\mathbf{C}(\vb*a)^{\sigma}_{\vb*\nu, J}(\vb z_2)\, .
\end{equation}
\end{corollary}

\section{Concluding remarks}\label{sec:conc}

The main problem studied in this Paper I, the characterization of symbol correspondences between a quantum mechanical system and a classical mechanical system on a symplectic manifold, which are both symmetric under $SU(3)$, here referred to as  quark systems, is often settled on  facts pertaining to systems symmetric under more general compact Lie groups, thus some of the features presented here are common to the case of spin systems ($SU(2)$-symmetric systems). For example, the realization of any symbol correspondence as expectation values over an \emph{operator kernel}, which is a special ``pseudo-state'', that is, a special Hermitian operator with unitary trace. Then, the more restricted case of a mapping-positive correspondence is a correspondence generated as expectation values over an operator kernel that is also an ``actual state'', thus being also a positive operator.

In particular, symbol correspondences for pure-quark systems show little formal distinction to what is known for spin systems, being also determined by ordered $n$-tuples of non zero real numbers, the \emph{characteristic numbers}. Nonetheless, we highlight two important differences. First, because representations $(p,0)$ and $(0,p)$ are not self-dual, antipodal correspondences are defined for pure-quark systems dual to each other and have the same characteristic numbers.\footnote{The antipodal relation stems from the action of the longest element of the Weyl group, begging the question of possible relations associated to the action of other elements of the Weyl group.} Second, for any pure-quark system, there exists only one Berezin (and thus only one Toeplitz) correspondence from quantum operators to functions on $\mathbb CP^2$.  

However, correspondences for general quark systems present new features from the degeneracy of representations within both the quantum operator space $\mathcal B(\mathcal H_{(p,q)})$ and the classical function space $C^\infty_{\mathbb C}(\mathcal E)$.  Then, the characterization of correspondences for general quark systems, particularly mixed-quark systems, are given not in terms of characteristic numbers, but in terms of \emph{characteristic matrices}. As consequence, there are multiple correspondences linked by a dual relation and, in addition to isometric (\emph{Stratonovich-Weyl}) correspondences, we have the more general definition of \emph{semi-conformal correspondences} as special cases of symbol correspondences, alongside the special cases of mapping-positive.

Our subsequent Paper II is dedicated to the problem of asymptotic behavior of symbol correspondences, studying the conditions under which the Poisson algebra of a classical system emerges as an asymptotic limit of operator algebras of quantum systems, via twisted products. In addition, while in this Paper I we only dealt with correspondences to classical quark systems defined on symplectic manifolds, in Paper II we shall also address the question of how the various symbol correspondences and their twisted algebras  of functions on (co)adjoint orbits can be glued along the Poisson manifold which is the unitary sphere $\mathcal S^7\subset \mathfrak{su}(3)$.

\bibliographystyle{plain}
\bibliography{main.bib}

\appendix

\section{Explanation of Definition \ref{gt-basis}}\label{sec:def_GT}
In this appendix, we  explain the Gelfand-Tsetlin method applied to the case of $SU(3)$, which is used in Definition \ref{gt-basis}, cf. (\ref{const})-(\ref{GThighestweight}) and (\ref{d-1}). For a general description of the Gelfand-Tsetlin method, see \cite{louck, zhel}.

We can take the matrices $E_{jk}$, with $j,k\in\{1,2,3\}$, given by $(E_{j,k})_{l,m} = \delta_{j,l}\delta_{k,m}$ as generators of the unitary group $U(3)$, so that $E_{jk}$ is a raising operator if $j<k$, it is a lower operator if $j>k$ and it is a Cartan operator if $j=k$. Those operators satisfy the commutation relations
\begin{equation}\label{e-commut}
    [E_{jk}, E_{lm}] = \delta_{l,k}E_{jm}-\delta_{j,m}E_{lk}.
\end{equation}
The triples $(\nu_1,\nu_2,\nu_3)$ widely used in this work are weights of representations of $U(3)$, nonnegative eigenvalues of $E_{11}$, $E_{22}$ and $E_{33}$, respectively.

One obtains generators of $SU(3)$ by maintaining the ladder operators and taking $(E_{11}-E_{22})/2$ and $(E_{22}-E_{33})/2$ as Cartan operators. Thus, an irreducible representation $(p,q)$ of $SU(3)$ gives rise to an irreducible representation of $U(3)$ with highest weight $(p+q+m,q+m,m)$, for any nonnegative integer $m$, and vice-versa\footnote{An irreducible representation of $U(3)$ with highest weight $(a_1,a_2,a_3)$ corresponds to the Young tableau with $a_j$ boxes in the $j$-th row.}. We conveniently choose $m=0$, so that we can identify an irrep $(p,q)$ of $SU(3)$ with the irrep of $U(3)$ with highest weight $(p+q,q,0)$.

Now, we want to unambiguously index an orthonormal basis of the representation consisting only of weight vectors. To do so, we consider the subrepresentations of the $U(2)$ related to the generators $E_{jk}$ with $j,k\in \{2,3\}$. Then we decompose the subrepresentations of this $U(2)$ into irreducible subrepresentations of the $U(1)$ generated by $E_{33}$. Since $U(1)$ is abelian, its irreducible representations are unidimensional, so we can get an orthonormal basis for $(p,q)$ by the restriction of $(p+q,q,0)$ to the chain of subgroups $U(1)\subset U(2)\subset U(3)$.

The classification of subrepresentations within a given representation for a chain of groups is usually called \emph{branching rule}, and the branching rule for $U(n-1)\subset U(n)$ is well known: they are multiplicity free and given by the so called \emph{betweenness condition} \cite{holman}. In our case of interest, with the choice $m=0$ as explained above, this means that the subrepresentations of $U(2)$ are determined by all pairs of integers $(r_+,r_-)$ such that $r_+$ is between $q$ and $p+q$ and $r_-$ is between $0$ and $q$, that is, 
\begin{equation}
    0\le r_-\le q \le r_+ \le p+q \ .
\end{equation}
Then, for each subrepresentation $(r_+,r_-)$ of $U(2)$, the subrepresentations of $U(1)$ associated to generator $E_{33}$ are given by integers $\nu_3$ satisfying
\begin{equation}
    r_-\le \nu_3 \le r_+ \ .
\end{equation}

It is straightforward to verify from (\ref{e-commut}) that $E_{11}+E_{22}+E_{33}$ is invariant by $U(3)$. By applying it to the vector with highest weight $(p+q,q,0)$, we get
\begin{equation}
    E_{11}+E_{22}+E_{33} = (p+2q)\mathds 1 \implies E_{11} = (p+2q)\mathds 1 - (E_{22}+E_{33})\ .
\end{equation}
Analogously, in a subrepresentation $(r_+,r_-)$ of $U(2)$, we have
\begin{equation}
    E_{22}+E_{33} = (r_++r_-)\mathds 1 \implies E_{22} = (r_++r_-)\mathds 1 - E_{33}\ .
\end{equation}
Therefore, $\nu_1$ and $\nu_2$ are given by
\begin{equation}
    \nu_1 = p+2q - (r_++r_-) \ \ \ \ \mbox{and} \ \ \ \ \nu_2 = r_++r_--\nu_3 \ .
\end{equation}

Just as  for $SU(3)\subset U(3)$, the operators $\{E_{23}, E_{32},(E_{22}-E_{33})/2\}$ are generators of $SU(2)$ among the generators of the chosen $U(2)$, so that we can identify the representation $(r_+,r_-)$ of $U(2)$ with the representation of $SU(2)$ with spin number
\begin{equation}
    J = (r_+-r_-)/2 \ . 
\end{equation}
Then, from $U_-=E_{32}$ and $T_-=E_{21}$, the coefficients in (\ref{d-1}) can be explicitly carried out by straightforward calculations as done in \cite{baird}.

\section{A method for Theorem \ref{cg-sym-theo}}\label{sec:mixed-casimirs}Let $\mathcal H_{\vb*p}$ be a Hilbert space carrying an irrep of class $\vb*p=(p,q)$ and let $A_{jk}$, for $j,k\in\{1,2,3\}$, be the generators satisfying:
\begin{equation}\label{Aeq1}
\begin{aligned} 
    &A_{12} = T_+ \ , \ \ A_{23} = U_+ \ , \ \ T_3 = \dfrac{1}{2}(A_{11}- A_{22}) \ , \ \ U_3 = \dfrac{1}{2}(A_{22}-A_{33}) \ , \\
    &A_{jk}^\dagger = A_{kj} \ , \ \ A_{11}+A_{22}+A_{33} = 0 \ , \ \
    [A_{jk}, A_{lm}] = \delta_{l,k}A_{jm}-\delta_{j,m}A_{lk} \ ,
\end{aligned}
\end{equation}
so that $\widecheck{\vb*p}=(q,p)$ is generated by the operators 
\begin{equation}\label{checkA}
  \widecheck A_{jk} = - A_{kj}  \ .
\end{equation}

Then, for $\vb*p=(p,q)$, the \emph{quadratic and cubic Casimir operators} are (cf. \cite{sharp})
\begin{equation}
\begin{aligned} 
    C_2 &:= \dfrac{1}{2}\sum_{j,k=1}^3A_{jk}A_{kj}=\dfrac{1}{3}[(p+q)(p+q+3)-pq]\mathds 1 \ , \\
    C_3 &:= \sum_{j,k,l=1}^3A_{jk}A_{kl}A_{lj}=\dfrac{1}{9}(p-q)(p+2q+3)(2p+q+3)\mathds 1 +3C_2 \ .
\end{aligned}     
\end{equation}

Now, for $x\in\{1,2,3\}$, consider $\mathcal H_x$ carrying the representation $\vb*p_x = (p_x,q_x)$ and the triple tensor product $\mathcal H_1\otimes \mathcal H_2\otimes \mathcal H_3$. Let $A_{jk}^{(x)}$, $C_2^{(x)}$ and $C_3^{(x)}$ be operators relative to $\vb*p_x$ and, for simplicity, given an operator $A^{(x)}$ on $\mathcal H_x$, we just write $A^{(x)}$ to denote its tensor product with the identity operator. Then, we have:
\begin{equation}\label{comut-E-cd}
    A_{jk}^{(x)}A_{lm}^{(y)} = A_{lm}^{(y)}A_{jk}^{(x)} \ , \ \ \mbox{for} \ x\neq y \ .
\end{equation}

\begin{definition}\label{defmixcas}
The \emph{mixed Casimir operators} are defined by
\begin{equation}
        C_2^{xy} := \sum_{j,k}A_{jk}^{(x)} A_{kj}^{(y)} \ , \ \ 
    C_3^{xyz} :=\sum_{j,k,l}A_{jk}^{(x)}A_{kl}^{(y)}A_{lj}^{(z)} \ , \label{mixedC3}
\end{equation}
for superindices not all equal, so that, from (\ref{comut-E-cd}) we also have 
\begin{equation}\label{c3-xyy}
C_2^{xy}=C_2^{yx} \ , \ \ 
C_3^{xyz} = C_3^{yzx} = C_3^{zxy} \ , \ \ C_3^{xyy} = C_3^{yyx} \ , \ \ C_3^{xyx} = C_3^{yxx}-C_{2}^{yx} \ .
\end{equation}
\end{definition}

Now, from the mixed cubic Casimir operators, we define the operators 
\begin{eqnarray}
    \mathbf{S}_{xy} &:=& \dfrac{1}{2}(C_3^{xyy} - C_3^{yxx}) = \dfrac{1}{2}(C_3^{yyx} - C_3^{yxx}) \ , \label{Sxydef} \\ 
    \mathbf S_{xyz} &:=& \dfrac{1}{3}(\mathbf S_{xy}+\mathbf S_{yz}+\mathbf S_{zx}) \ , \label{Sxyzdef}
\end{eqnarray}
and likewise for $\widecheck{\mathbf{S}}_{xy}$ and $\widecheck{\mathbf{S}}_{xyz}$.
Then, it is straightforward to check that these operators are antisymmetric under odd permutation of indices and dualization, 
\begin{eqnarray}
\mathbf S_{xy} &=& -\mathbf S_{yx} \ = \  -\widecheck{\mathbf S}_{xy} \ ,  \label{xy-yx}\\
    \mathbf S_{xyz} &=& - \mathbf S_{yxz} \ = \  - \mathbf S_{xzy} \ = \ -\widecheck{\mathbf S}_{xyz} \ . \label{xyz-yxz}
\end{eqnarray}

Now, let $\mathcal H^0\equiv \mathcal H^0_{123}$ be a maximal subspace of $\mathcal H_1\otimes \mathcal H_2\otimes \mathcal H_3$ where $SU(3)$ acts trivially. That is, for all $j,k\in\{1,2,3\}$,
\begin{equation}\label{1+2+3=0}
    A_{jk}^{(1)}+A_{jk}^{(2)}+A_{jk}^{(3)} = 0 \quad \mbox{on} \ \  \mathcal H^0 \ . 
\end{equation}

\begin{lemma}\label{p3check}
$\mathcal H^0$ is not null if and only if there is a representation of class $\widecheck{\vb*p}_3$ in the Clebsch-Gordan series of $\vb*p_1\otimes \vb*p_2$.
\end{lemma}
\begin{proof}
From the Clebsch-Gordan series
\begin{equation}\label{cgs-p1p2}
    \vb*p_1\otimes \vb*p_2 = \bigoplus_{(\vb*a;\sigma)}(\vb*a;\sigma) \  ,  \ \
    \vb*p_1\otimes \vb*p_2\otimes \vb*p_3 = \bigoplus_{(\vb*a;\sigma)}(\vb*a;\sigma)\otimes \vb*p_3 \ ,
\end{equation}
note that $k\ge 1$ in (\ref{cg-s-eq2}), so there exists a factor of class $(0,0)$ in the CG series of $(a,b)\otimes (p_3,q_3)$ iff there is $0\le n\le \min(a,q_3)$ and $0\le m\le \min(p_3,b)$ satisfying
\begin{equation*}
    a-n+p_3-m = b-m+q_3-n = 0 \ .
\end{equation*}
The only possible solution is $n = a = q_3$ and $m = p_3 = b$. Thus, $\mathcal H^0$ is not null iff there is a representation of class $\vb*a=\widecheck{\vb*p}_3$ in the r.h.s. of (\ref{cgs-p1p2}).
\end{proof}

And by a straightforward computation using (\ref{mixedC3})-(\ref{Sxyzdef}) and (\ref{1+2+3=0}), we obtain:
\begin{equation}\label{S0123}
\mathbf S_{123}^0:= \mathbf S_{123}|_{\mathcal H^0}    =\mathbf S_{12}-\dfrac{1}{3}C_3^{(1)}+\dfrac{1}{3}C_3^{(2)} \ .
\end{equation}

Then, the following is also straightforward. 
\begin{lemma}\label{s-hermit}
The operators $\mathbf S_{12}$ and $\mathbf S_{123}^0$ are Hermitian and $SU(3)$-invariant.
\end{lemma}

In this way, we decompose degenerate irreducible representations in the Clebsch-Gordan series of $\vb*p_1\otimes \vb*p_2$ via diagonalization of the operator 
\begin{equation}\label{Sdef}
    \mathbf S := \mathbf S_{12}-\dfrac{1}{3}C_3^{(1)}+\dfrac{1}{3}C_3^{(2)} \ , \ \ \ \mbox{satisfying} \ \ \ \mathbf S|_{\mathcal H^0_{123}}=\mathbf S^0_{123} \ , 
\end{equation}
cf. (\ref{S0123}). Because $\mathbf S$ is built from Casimir operators,  the eigenvalues $s_{123}$ of $\mathbf S^0_{123}$ depend only on the subrepresentations comprising $\mathcal H^0_{123}\subset \mathcal H_1\otimes\mathcal H_2\otimes\mathcal H_3$.

\begin{theorem}[\cite{chew1,pluh1}]\label{sdistinct}
The eigenvalues of $\mathbf S^0_{123}$ are distinct, for distinct irreducible subrepresentations in $\mathcal H^0_{123}$.
\end{theorem}

\begin{remark}
The proof of the above theorem  shall not be done here.
In \cite{chew1}, the authors use a polynomial basis for $SU(3)$-representations, as constructed e.g. in \cite{sharp}, to compute the matrix entries of an operator equivalent to $\mathbf S$ (reduced to an operator on the space spanned by the highest-weight vectors of each irreducible representation in $\mathcal H^0_{123}$),   in order to prove Theorem \ref{sdistinct}. But reproducing these computations in detail is a rather long and tedious exercise. 
\end{remark}

Thus, for (\ref{3tensorprod}), 
with $\widecheck{\vb*a}=\vb*p_3$, cf. Lemma \ref{p3check}, 
we can define a function 
\begin{equation}\label{s12a}
    s_{\vb*p_1, \vb*p_2;\vb*a}: \{1,...,\textswab{m}(\vb*p_1, \vb*p_2;\vb*a)\}\to\mathbb R \ , \ \sigma\mapsto s_{\vb*p_1, \vb*p_2;\vb*a}(\sigma) \ , 
\end{equation}
which indexes the eigenvalues of $\mathbf S^0_{123}$ in this case. Denoting $s_{\vb*p_1, \vb*p_2;\vb*a}\equiv s_{\vb*a}$, for simplicity, we adopt the convention that $s_{\vb*a}$ is an increasing function of the multiplicity counting index, that is, 
\begin{equation}\label{s_a(sigma)-1}
    \{1,...,\textswab{m}(\vb*p_1, \vb*p_2;\vb*a)\}\ni\sigma\mapsto s_{\vb*a}(\sigma) \in \mathbb R \ , \ \ s_{\vb*a}(\sigma)< s_{\vb*a}(\sigma+1) \ .
\end{equation}
In other words, the eigenvalues of $\mathbf S^0_{123}$ are indexed by increasing order.

Then, the involution (\ref{sigma-check}) is at par with the involution $s\mapsto -s$ in the set of eigenvalues for $\mathbf S$ under $\mathbf S\mapsto \widecheck{\mathbf S}$, which is a consequence of (\ref{xy-yx})-(\ref{xyz-yxz}), taking into account the convention (\ref{s_a(sigma)-1}).  And in this way, taking into account all the symmetries of $\mathbf S$, we obtain Theorem (\ref{cg-sym-theo}).

\section{Wigner symbols for $SU(3)$}\label{sec:wig_symb}
In this appendix, we fix the convention of Theorem \ref{cg-sym-theo}.

For $SU(2)$, CG coefficients can be substituted by other coefficients with neater symmetry properties, the so-called \emph{Wigner $3jm$-symbols} (cf. e.g. \cite{RS, varsh, wig}).\footnote{References \cite{butl, der-sh} define the generalized Wigner $3jm$-symbols for general compact groups. Here, we shall follow the conventions in \cite{pluh2} for $SU(3)$.} We shall use such variations of the CG coefficients and the recoupling coefficients (defined below), in order to rewrite the decomposition of the operator product in the coupled basis in a way that shows explicitly all symmetries of the product.

\begin{definition}\label{Wigner-coupling}
	The \emph{Wigner coupling symbol} is the coefficient denoted by the round brackets below:
	\begin{equation}\label{Wcoupling}
		\begin{pmatrix}
			\vb*p_1 & \vb*p_2 & (\vb*a;\sigma)\\
			\vb*\nu_1, J_1 & \vb*\nu_2, J_2 & \vb*\nu, J
		\end{pmatrix} = \dfrac{(-1)^{|\vb*a|+2(t_{\vb*\nu}+u_{\vb*\nu})}}{\sqrt{\dim (\widecheck{\vb*a})}}\cg{\vb*p_1}{\vb*p_2}{(\widecheck{\vb*a};\sigma)}{\vb*\nu_1 J_1}{\vb*\nu_2 J_2}{\,\,\widecheck{\vb*\nu} J} \ .
	\end{equation}
\end{definition}

Thus, from Theorem \ref{cg-sym-theo}, we have the symmetries 
\begin{equation}\label{wig-3-sym}
	\begin{aligned}
		\begin{pmatrix}
			\vb*p_1 & \vb*p_2 & (\vb*a;\sigma)\\
			\vb*\nu_1, J_1 & \vb*\nu_2, J_2 & \vb*\nu, J
		\end{pmatrix} & = (-1)^{|\vb*p_1|+|\vb*p_2|+|\vb*a|}\begin{pmatrix}
			\vb*p_2 & \vb*p_1 & (\vb*a;\widecheck\sigma)\\
			\vb*\nu_2, J_2 & \vb*\nu_1, J_1 & \vb*\nu, J
		\end{pmatrix}\\
		&  =  (-1)^{|\vb*p_1|+|\vb*p_2|+|\vb*a|}\begin{pmatrix}
			\vb*p_1 & \vb*a & (\vb*p_2; \widecheck\sigma)\\
			\vb*\nu_1, J_1 & \vb*\nu, J & \vb*\nu_2, J_2
		\end{pmatrix} \\ 
		& = (-1)^{|\vb*p_1|+|\vb*p_2|+|\vb*a|}\begin{pmatrix}
			\widecheck{\vb*p}_1 & \widecheck{\vb*p}_2 & (\widecheck{\vb*a}; \widecheck\sigma)\\
			\widecheck{\vb*\nu}_1, J_1 & \widecheck{\vb*\nu}_2, J_2 & \widecheck{\vb*\nu}, J
		\end{pmatrix}  .
	\end{aligned}
\end{equation}

Now, consider the two Clebsch-Gordan sub-series for $\vb*p_1\otimes \vb*p_2\otimes \vb*p_3$, 
\begin{equation*}
	\begin{aligned}
		&\vb*p_1\otimes \vb*p_2 = \bigoplus_{(\vb*a_{12};\sigma_{12})}(\vb*a_{12}; \sigma_{12}) \ , \ \ (\vb*a_{12};\sigma_{12})\otimes \vb*p_3 = \bigoplus_{(\vb*a;\sigma)}(\vb*a;\sigma) \ , \\ 
		&\vb*p_2\otimes \vb*p_3 = \bigoplus_{(\vb*a_{23};\sigma_{23})}(\vb*a_{23};\sigma_{23}) \ , \ \ \vb*p_1\otimes (\vb*a_{23};\sigma_{23}) = \bigoplus_{(\vb*a';\sigma')}(\vb*a';\sigma') \ ,
	\end{aligned}     
\end{equation*}
defining two basis,  $\{\vb* e((\vb*a_{12}; \sigma_{12}),(\vb*a;\sigma); \vb*\mu, L)\}$ and $\{e((\vb*a_{23};\sigma_{23}),(\vb*a';\sigma'); \vb*\mu', L')\}$, for $\vb*p_1\otimes \vb*p_2\otimes \vb*p_3$ satisfying (cf. (\ref{coup-unc}) and (\ref{Wcoupling})):
\begin{eqnarray}
	&\vb* e((\vb*a_{12}; \sigma_{12}),(\vb*a;\sigma); \vb*\mu, L) \ = \label{b12-3} \\
	&\sqrt{\dim (\vb*a)\dim (\vb*a_{12})}\sum_{\substack{\vb*\mu_{12}, L_{12}\\ \vb*\nu_3, J_3}}\sum_{\substack{\vb*\nu_1, J_1\\\vb*\nu_2, J_2}}(-1)^{|\vb*a|+|\vb*a_{12}|+2(t_{\vb*\mu}+u_{\vb*\mu}+t_{\vb*\mu_{12}}+u_{\vb*\mu_{12}})} \nonumber \\
	& \ \ \ \times \begin{pmatrix}
		\vb*p_1 & \vb*p_2 & (\widecheck{\vb*a}_{12};\sigma_{12})\\
		\vb*\nu_1, J_1 & \vb*\nu_2, J_2 & \widecheck{\vb*\mu}_{12}, L_{12}
	\end{pmatrix}\begin{pmatrix}
		\vb*a_{12} & \vb*p_3 & (\widecheck{\vb*a};\sigma)\\
		\vb*\mu_{12}, L_{12} & \vb*\nu_3, J_3 & \widecheck{\vb*\mu}, L
	\end{pmatrix} \times \hat{\vb*e}_{123} \ , \nonumber \\
	& \vb* e((\vb*a_{23};\sigma_{23}),(\vb*a';\sigma'); \vb*\mu', L') \ = \label{b1-23} \\
	& \sqrt{\dim (\vb*a')\dim (\vb*a_{23})}\sum_{\substack{\vb*\mu_{23}, L_{23}\\ \vb*\nu_1, J_1}}\sum_{\substack{\vb*\nu_2, J_2\\\vb*\nu_3, J_3}}(-1)^{|\vb*a'|+|\vb*a_{23}|+2(t_{\vb*\mu'}+u_{\vb*\mu'}+t_{\vb*\mu_{23}}+u_{\vb*\mu_{23}})} \nonumber \\
	& \ \ \ \ \ \times \begin{pmatrix}
		\vb*p_1 & \vb*a_{23} & (\widecheck{\vb*a}';\sigma')\\
		\vb*\nu_1, J_1 & \vb*\mu_{23}, L_{23} &  \widecheck{\vb*\mu}', L'
	\end{pmatrix}\begin{pmatrix}
		\vb*p_2 & \vb*p_3 & (\widecheck{\vb*a}_{23};\sigma_{23})\\
		\vb*\nu_2, J_2 & \vb*\nu_3, J_3 & \widecheck{\vb*\mu}_{23}, L_{23}
	\end{pmatrix} \times \hat{\vb*e}_{123} \ , \nonumber 
\end{eqnarray}
where \ $\hat{\vb*e}_{123}\equiv\vb*e(\vb*p_1; \vb*\nu_1, J_1)\otimes \vb*e(\vb*p_2; \vb*\nu_2, J_2)\otimes \vb*e(\vb*p_3; \vb*\nu_3, J_3)$. 
Of course, there is an unitary transformation relating these two basis and
\begin{equation*}
	\ip{\vb* e((\vb*a_{12};\sigma_{12}),(\vb*a;\sigma); \vb*\mu, L)}{\vb* e((\vb*a_{23};\sigma_{23}),(\vb*a';\sigma'); \vb*\mu', L')} \ne 0 \!\!\implies\!\!\!\begin{cases}
		\vb*a = \vb*a'\\ (\vb*\mu, L) = (\vb*\mu', L')
	\end{cases}
\end{equation*}
Also, the coefficients $\ip{\vb* e((\vb*a_{12};\sigma_{12}),(\vb*a;\sigma); \vb*\mu, L)}{\vb* e((\vb*a_{23};\sigma_{23}),(\vb*a;\sigma'); \vb*\mu, L)}$ don't depend on weight and spin number $L$ since these vectors can be generated from the highest weight vectors of their respective representations by applying ladder operators and we can write a highest weight vector of one basis as a linear combination of highest weight vectors of the other basis for equivalent representations.

\begin{definition}\label{Wigner-recoupling}
	The \emph{Wigner recoupling symbol} is the coefficient denoted by the curly brackets below:
	\begin{equation}
		\begin{aligned}
			& \begin{Bmatrix}
				\vb*p_1 & \vb*p_2 & (\widecheck{\vb*a}_{12};\sigma_{12})\\
				\vb*p_3 & (\vb*a;\sigma,\sigma') & (\vb*a_{23};\sigma_{23})
			\end{Bmatrix}\\
			&\hspace{3 em} = \ \sum_{\substack{\vb*\nu_1, J_1\\\vb*\nu_2, J_2\\\vb*\nu_3, J_3}}\sum_{\substack{\vb*\mu_{12}, L_{12}\\\vb*\mu_{23}, L_{23}\\ \vb*\mu, L}}(-1)^{|\vb*p_2|+|\vb*p_3|+|\vb*a_{12}|+2(t_{\vb*\mu_{12}}+u_{\vb*\mu_{12}}+t_{\vb*\mu_{23}}+u_{\vb*\mu_{23}})}\\
			& \hspace{5 em}\times\begin{pmatrix}
				\vb*a_{12} & \vb*p_3 & (\widecheck{\vb*a};\sigma)\\
				\vb*\mu_{12}, L_{12} & \vb*\nu_3, J_3 & \widecheck{\vb*\mu}, L
			\end{pmatrix}\begin{pmatrix}
				\vb*p_1 & \vb*p_2 & (\widecheck{\vb*a}_{12};\sigma_{12})\\
				\vb*\nu_1, J_1 & \vb*\nu_2, J_2 & \widecheck{\vb*\mu}_{12}, L_{12}
			\end{pmatrix}\\
			& \hspace{5 em}\times \begin{pmatrix}
				\vb*a_{23} & \vb*p_1 & (\widecheck{\vb*a};\widecheck\sigma')\\
				\vb*\mu_{23}, L_{23} & \vb*\nu_1, J_1 & \widecheck{\vb*\mu}, L
			\end{pmatrix}\begin{pmatrix}
				\vb*p_2 & \vb*p_3 & (\widecheck{\vb*a}_{23};\sigma_{23})\\
				\vb*\nu_2, J_2 & \vb*\nu_3, J_3 & \widecheck{\vb*\mu}_{23}, L_{23}
			\end{pmatrix}\ .
		\end{aligned}
	\end{equation}
\end{definition}

The name \emph{recoupling symbol} is justified by the following equation: \begin{equation}\label{wig_rec_lc}
	\begin{aligned}
		& \vb* e((\vb*a_{23};\sigma_{23}),(\vb*a;\widecheck\sigma'); \vb*\mu, L) \\
		& \hspace{5 em }= \sum_{\substack{(\vb*a_{12},\sigma_{12})\\\sigma}}(-1)^{|\vb*a_{23}|+|\vb*p_2|+|\vb*p_3|}\sqrt{\dim (\vb*a_{12}) \dim (\vb*a_{23})}\\
		& \hspace{7 em}\times \begin{Bmatrix}
			\vb*p_1 & \vb*p_2 & (\widecheck{\vb*a}_{12};\sigma_{12})\\
			\vb*p_3 & (\vb*a;\sigma,\sigma') & (\vb*a_{23};\sigma_{23})
		\end{Bmatrix}\, \vb* e((\vb*a_{12};\sigma_{12}),(\vb*a;\sigma); \vb*\mu, L) \ .
	\end{aligned}
\end{equation}

Then, from (\ref{b12-3}) and (\ref{b1-23}), one obtains  \emph{Wigner's identity}\footnote{An equivalent formula is deduced for $SU(2)$ in \cite{wig} and for a general compact group in \cite{butl}.}: 
\begin{equation}\label{wigner-id}
	\begin{aligned}
		& \sum_{\substack{\vb*\mu_{23}, L_{23}}}(-1)^{2(t_{\vb*\nu_1}+u_{\vb*\nu_1})}\begin{pmatrix}
			\vb*a_{23} & \vb*p_1 & (\widecheck{\vb*a};\widecheck\sigma')\\
			\vb*\mu_{23}, L_{23} & \vb*\nu_1, J_1 & \widecheck{\vb*\mu}, L
		\end{pmatrix}\begin{pmatrix}
			\vb*p_2 & \vb*p_3 & (\widecheck{\vb*a}_{23};\sigma_{23})\\
			\vb*\nu_2, J_2 & \vb*\nu_3, J_3 & \widecheck{\vb*\mu}_{23}, L_{23}
		\end{pmatrix}\\
		&\\
		& = \sum_{\substack{\vb*\mu_{12}, L_{12}\\(\vb*a_{12};\sigma_{12})\\\sigma}}\!\!(-1)^{|\vb*a_{12}|+|\vb*p_2|+|\vb*p_3|+2(t_{\vb*\nu_3}+u_{\vb*\nu_3})}\dim (\vb*a_{12})\begin{Bmatrix}
			\vb*p_1 & \vb*p_2 & (\widecheck{\vb*a}_{12};\sigma_{12})\\
			\vb*p_3 & (\vb*a;\sigma,\sigma') & (\vb*a_{23};\sigma_{23})
		\end{Bmatrix}\\
		& \ \ \ \ \ \ \ \ \ \ \ \ \ \ \times \begin{pmatrix}
			\vb*a_{12} & \vb*p_3 & (\widecheck{\vb*a};\sigma)\\
			\vb*\mu_{12}, L_{12} & \vb*\nu_3, J_3 & \widecheck{\vb*\mu}, L
		\end{pmatrix} \begin{pmatrix}
			\vb*p_1 & \vb*p_2 & (\widecheck{\vb*a}_{12};\sigma_{12})\\
			\vb*\nu_1, J_1 & \vb*\nu_2, J_2 & \widecheck{\vb*\mu}_{12}, L_{12}
		\end{pmatrix} .
	\end{aligned}
\end{equation}
And using (\ref{wig-3-sym}), we obtain the symmetries 
\begin{equation}\label{wig-6-sym}
	\begin{aligned}
		& \begin{Bmatrix}
			\vb*p_1 & \vb*p_2 & (\widecheck{\vb*a}_{12};\sigma_{12})\\
			\vb*p_3 & (\vb*a;\sigma,\sigma') & (\vb*a_{23};\sigma_{23})
		\end{Bmatrix} \\ = & \begin{Bmatrix}
			\widecheck{\vb*p}_2 & \widecheck{\vb*p}_1 & (\vb*a_{12};\sigma_{12})\\
			\vb*a & (\vb*p_3;\sigma, \sigma_{23}) & (\vb*a_{23};\sigma')
		\end{Bmatrix}
		=  \begin{Bmatrix}
			\widecheck{\vb*p}_1 & \vb*a_{12} & (\widecheck{\vb*p}_2;\sigma_{12})\\
			\vb*p_3 & (\vb*a_{23};\sigma_{23}, \sigma') & (\vb*a;\sigma)
		\end{Bmatrix} \\ 
		= & \begin{Bmatrix}
			\vb*p_3 & \widecheck{\vb*a} & (\vb*a_{12};\sigma')\\
			\vb*p_1 & (\widecheck{\vb*p}_2;\sigma_{12}, \sigma_{23}) & (\widecheck{\vb*a}_{23};\sigma)
		\end{Bmatrix}
		=  \begin{Bmatrix}
			\widecheck{\vb*p}_1 & \widecheck{\vb*p}_2 & (\vb*a_{12};\widecheck\sigma_{12})\\
			\widecheck{\vb*p}_3 & (\widecheck{\vb*a};\widecheck\sigma,\widecheck\sigma') & (\widecheck{\vb*a}_{23};\widecheck\sigma_{23})
		\end{Bmatrix} \ .
	\end{aligned}
\end{equation}

\begin{theorem}
	The Wigner product symbol is given by
	\begin{equation}
		\begin{aligned}
			&\begin{bmatrix}
				(\vb*a_1;\sigma_1) & (\vb*a_2;\sigma_2) & (\vb*a;\sigma)\\
				\vb*\nu_1, J_1 & \vb*\nu_2, J_2 & \vb*\nu, J
			\end{bmatrix}\![\vb*p] \ = \ \sqrt{\dim (\vb*a_1)\dim (\vb*a_2)\dim (\vb*a)}\\
			&\hspace{3 em} \times  \sum_{\sigma'}\begin{Bmatrix}
				\vb*a_1 & \vb*a_2 & (\widecheck{\vb*a};\sigma')\\
				\vb*p & (\vb*p;\sigma, \sigma_1) & (\vb*p;\sigma_2)
			\end{Bmatrix}\begin{pmatrix}
				\vb*a_1 & \vb*a_2 & (\widecheck{\vb*a};\sigma')\\
				\vb*\nu_1, J_1 & \vb*\nu_2, J_2 & \vb*\nu, J
			\end{pmatrix} \ .
		\end{aligned}
	\end{equation}
\end{theorem}
\begin{proof}
	From (\ref{wigner-id}), we have
	\begin{equation}
		\begin{aligned}
			& \sum_{\substack{\vb*\mu_{3}, L_{3}}}(-1)^{2(t_{\vb*\nu_1}+u_{\vb*\nu_1})}\begin{pmatrix}
				\vb*p & \vb*a_1 & (\widecheck{\vb*p};\widecheck\sigma_1)\\
				\vb*\mu_{3}, L_{3} & \vb*\nu_1, J_1 & \widecheck{\vb*\mu}_2, L_2
			\end{pmatrix}\begin{pmatrix}
				\vb*a_2 & \vb*p & (\widecheck{\vb*p};\sigma_2)\\
				\vb*\nu_2, J_2 & \vb*\mu_1, L_1 & \widecheck{\vb*\mu}_{3}, L_{3}
			\end{pmatrix}\\
			&  = \sum_{\substack{\vb*\nu', \  L'\\(\vb*a';\sigma'), \ \sigma''}}(-1)^{|\vb*a'|+|\vb*a_2|+|\vb*p|+2(t_{\vb*\mu_1}+u_{\vb*\mu_1})}\dim (\vb*a')\begin{Bmatrix}
				\vb*a_1 & \vb*a_2 & (\widecheck{\vb*a}';\sigma')\\
				\vb*p & (\vb*p;\sigma'',\sigma_1) & (\vb*p;\sigma_2)
			\end{Bmatrix}\\
			& \ \ \ \ \ \ \ \ \ \ \ \ \ \ \ \  \ \ \ \times \begin{pmatrix}
				\vb*a' & \vb*p & (\widecheck{\vb*p};\sigma'')\\
				\vb*\nu', L' & \vb*\mu_1, L_1 & \widecheck{\vb*\mu}_2, L_2
			\end{pmatrix} \begin{pmatrix}
				\vb*a_1 & \vb*a_2 & (\widecheck{\vb*a}';\sigma')\\
				\vb*\nu_1, J_1 & \vb*\nu_2, J_2 & \widecheck{\vb*\nu}', L'
			\end{pmatrix} \ .
		\end{aligned}
	\end{equation}
	Multiplying both sides by \ 
	$(-1)^{2(t_{\vb*\nu}+u_{\vb*\nu}+t_{\vb*\mu_2}+u_{\vb*\mu_2})}\begin{pmatrix}
		\vb*p & \widecheck{\vb*p} & (\vb*a;\sigma)\\
		\vb*\mu_1,L_1 & \widecheck{\vb*\mu}_2,L_2 & \vb*\nu, J
	\end{pmatrix}$
	and summing over $\vb*\mu_1$, $L_1$, $\vb*\mu_2$, $L_2$, using (\ref{wig-3-sym}) and (\ref{ortho-cg}), we obtain what we want.
\end{proof}

Then, from (\ref{cg-ne-0}), plus symmetries (\ref{wig-3-sym}) and (\ref{wig-6-sym}), recalling (\ref{sigma-check}), we have
\begin{theorem}\label{symWP}
	The Wigner product symbol satisfies: 
	\begin{equation}\label{prod-s-ne-0}
		\begin{bmatrix}
			(\vb*a_1;\sigma_1) & (\vb*a_2;\sigma_2) & (\vb*a;\sigma)\\
			\vb*\nu_1, J_1 & \vb*\nu_2, J_2 & \vb*\nu, J
		\end{bmatrix}\!\![\vb*p] \ \neq 0  \ \ \ \ \implies \ \ \ \ \begin{cases}
			\ \nabla_{\vb*\nu_1+\vb*\nu_2, \widecheck{\vb*\nu}} = 1 \\
			\ \delta(J_1,J_2,J) = 1
		\end{cases}  ,
	\end{equation}
	\begin{equation}\label{prod-s-sym}
		\begin{aligned}
			\begin{bmatrix}
				(\vb*a_1;\sigma_1) & (\vb*a_2;\sigma_2) & (\vb*a_3;\sigma_3)\\
				\vb*\nu_1, J_1 & \vb*\nu_2, J_2 & \vb*\nu_3, J_3
			\end{bmatrix}\!\![\vb*p] \  = &\begin{bmatrix}
				(\vb*a_3;\sigma_3) & (\vb*a_1;\sigma_1) & (\vb*a_2;\sigma_2)\\
				\vb*\nu_3, J_3 & \vb*\nu_1, J_1 & \vb*\nu_2, J_2
			\end{bmatrix}\!\![\vb*p]\\
			= &\begin{bmatrix}
				(\vb*a_2;\sigma_2) & (\vb*a_3;\sigma_3) & (\vb*a_1;\sigma_1)\\
				\vb*\nu_2, J_2 & \vb*\nu_3, J_3 & \vb*\nu_1, J_1
			\end{bmatrix}\!\![\vb*p]\\
			= &\begin{bmatrix}
				(\widecheck{\vb*a}_2;\sigma_2) & (\widecheck{\vb*a}_1;\sigma_1) & (\widecheck{\vb*a}_3;\sigma_3)\\
				\widecheck{\vb*\nu}_2, J_2 & \widecheck{\vb*\nu}_1, J_1 & \widecheck{\vb*\nu}_3, J_3
			\end{bmatrix}\!\![\vb*p]\\
			= (-1)^{\sum_{k=1}^3|\vb*a_k|}&\begin{bmatrix}
				(\widecheck{\vb*a}_1;\widecheck\sigma_1) & (\widecheck{\vb*a}_2;\widecheck\sigma_2) & (\widecheck{\vb*a}_3;\widecheck\sigma_3)\\
				\widecheck{\vb*\nu}_1, J_1 & \widecheck{\vb*\nu}_2, J_2 & \widecheck{\vb*\nu}_3, J_3
			\end{bmatrix}\!\![\widecheck{\vb*p}] \\
			= (-1)^{\sum_{k=1}^3|\vb*a_k|}&\begin{bmatrix}
				(\vb*a_2;\widecheck\sigma_2) & (\vb*a_1;\widecheck\sigma_1) & (\vb*a_3;\widecheck\sigma_3)\\
				\vb*\nu_2, J_2 & \vb*\nu_1, J_1 & \vb*\nu_3, J_3
			\end{bmatrix}\!\![\widecheck{\vb*p}]\ .
		\end{aligned}
	\end{equation}
\end{theorem}

\section{Justification of Definition \ref{QPQsystem}}\label{sec:def_PQ}

Consider the defining representation $\rho_1$ of class $(1,0)$ on $\mathcal H_{(1,0)} \simeq \mathbb C^3$, and let the canonical basis $\{e_1, e_2, e_3\}$ match a GT basis, with each vector associated to a weight on the diagram of Figure \ref{fig:basis-diag}(A), $e_1$ with the highest weight and $e_2,e_3$ ordered counterclockwise. Then, the tensor product space $\mathcal H = \mathcal H_{(1,0)}\otimes...\otimes \mathcal H_{(1,0)}$ with $p$ copies of $\mathcal H_{(1,0)}$ carries the induced representation $\rho=\rho_1\otimes...\otimes\rho_1$. 

Let $\mathcal H_{(p,0)} = Sym^p(\mathcal H_{(1,0)}) \subset \mathcal H$ be the subspace of totally symmetric tensors, $\sum_{i_1,...,i_p = 1}^3 c_{i_1,...,i_p}\, e_{i_1}\otimes...\otimes e_{i_p} \in \mathcal H_{p,0} \iff c_{i_{f(1)},...,i_{f(p)}} = c_{i_1,...,i_p}$, for every permutation $f\in S_p$. It is immediate that $\mathcal H_{(p,0)}$ is an invariant subspace.
We get a basis of $\mathcal H_{(p,0)}$ by means of symmetrization. For $e_{i_1}\otimes...\otimes e_{i_p}\in \mathcal H$, let $j$, $k$ and $l$ be the numbers of occurrence of index $1$, $2$ and $3$, respectively, and take
\begin{equation}\label{v_p-basis}
    e_{j,k,l} = \left(\dfrac{p!}{j!\,k!\,l!}\right)^{-1/2}\sum_{f\in S_p} e_{i_{f(1)}}\otimes...\otimes e_{i_{f(p)}} \in \mathcal H_{p,0} \ .
\end{equation}

The set $\{e_{j,k,l}: j+k+l = p\}$ is an orthonormal basis of $\mathcal H_{(p,0)}$ considering the inner product induced by $\mathcal H_{(1,0)}$ on $\mathcal H$. Starting with the element $e_{p,0,0}$, we  obtain the basis $\{e_{j,k,l}: j+k+l = p\}$ by recursively applying the ladder operators $T_-$ and $U_-$ and normalizing the result. As can be seen from the diagram of Figure \ref{fig:basis-diag}(A), $e_{j,k,l} = \mu_{j,k,l}(U_-)^l(T_-)^{k+l}e_{p,0,0}$, where $\mu_{j,k,l}>0$. Since $\dim \mathcal H_{p,0} = (p+1)(p+2)/2$ and $e_{p,0,0}$ is a highest weight vector with eigenvalues $p/2$ for $T_3$ and $0$ for $U_3$, we conclude that representation on $\mathcal H_{p,0}$ is an irreducible representation of class $(p,0)$. Likewise for $\mathcal H_{(0,p)}=\mathcal H_{(p,0)}^\ast$. In particular, we have the equivariant maps
\begin{equation}\label{equiv-p-form}
\begin{aligned}
    \mathcal H_{(1,0)} &\to \mathcal H_{(p,0)}: w = (z_1,z_2,z_3) \mapsto w\otimes...\otimes w =\!\!\sum_{j+k+l=1}\sqrt{\dfrac{p!}{j!\,k!\,l!}}z_1^jz_2^kz_3^l \, e_{j,k,l} \ , \\
    \mathcal H_{(0,1)} &\to \mathcal H_{(0,p)}: w^\ast = (z_1,z_2,z_3) \mapsto w^\ast\otimes...\otimes w^\ast = \!\sum_{j+k+l=1}\sqrt{\dfrac{p!}{j!\,k!\,l!}}z_1^jz_2^kz_3^l \, \widecheck e_{j,k,l} \ ,
\end{aligned}     
\end{equation}
where $\{\widecheck e_{j,k,l}: j+k+l=p\}$ is the basis induced by\footnote{Again, $\widecheck e_{p,0,0}$ is a highest weight vector and $\widecheck e_{0,0,p}$ is a lowest weight vector.} $\{\widecheck e_1 = -e_3^\ast, \widecheck e_2 = e_2^\ast, \widecheck e_3 = -e_1^\ast\}$ just like $\{e_{j,k,l}: j+k+l=p\}$ is induced by $\{e_1,e_2,e_3\}$, cf. Definition \ref{gt-dual}.

In Physics, the space of colors (resp. flavors) of a quark is precisely the representation $(1,0)$, with $e_1\equiv$ \emph{red} (resp. \emph{up quark}), $e_2\equiv$ \emph{blue} (resp. \emph{down quark}) and $e_3\equiv$ \emph{green} (resp. \emph{strange quark}). Thus, $(p,0)$ is the totally symmetric part of a system of $p$ quarks. Analogously for  $(0,q)$ and antiquarks. Therefore, such spaces arise in description of systems of $p$ identical quarks only (or $q$ identical antiquarks only) and we call them \emph{pure-quark systems} because the number of antiquarks (or quarks) is zero. On the other hand, because the highest or lowest weight space of $(p,0)$ or $(0,q)$ have the maximal isotropy subgroup $H\simeq U(2)$, these representations are also called \emph{symmetric representations}, in the Mathematics literature. So,  pure-quark systems could also be referred to as symmetric quark systems.

\begin{remark}\label{pq-isot-osc}
There is another interpretation for the representation $(p,0)$ as a quantum system whose classical phase space is $\mathbb CP^2$. A quantum three-dimensional isotropic harmonic oscillator is an affine quantum system with Hamiltonian
\begin{equation}\label{ham-q}
    H = \sum_{i=1}^3 H_i \ , \ \ H_i=\left(\dfrac{1}{2m} P_i^2 + \dfrac{1}{2}m\omega^2 X_i^2\right)\ , 
\end{equation}
where $P_i$ and $X_i$ are the component operators of momentum and position, for some positive parameters $m$ and $\omega$. It has degenerate energy levels\footnote{We are setting $\hbar = 1$ throughout this paper.}
\begin{equation}
    E_p = \left(p+\dfrac{3}{2}\right)\omega \ , \ \ p=n_1+n_2+n_3 \ , \ \ E_i=\left(n_i+\dfrac{1}{2}\right)\omega \ , 
\end{equation}
with $SU(3)$-symmetry given by representations $(p,0)$ \cite{murg}.
For the classical $3$-d isotropic harmonic oscillator, the phase space is $T^\ast \mathbb R^3 \simeq \mathbb R^6$, with Hamiltonian 
\begin{equation}\label{ham-c}
    h(\vec x,\vec p) = \sum_{i=1}^3\left(\dfrac{p_i^2}{2m}+\dfrac{1}{2}m\omega x_i^2\right)=\sum_{i=1}^3h(x_i,p_i) \ ,
\end{equation}
where $p_i$ and $x_i$ are the components of momentum and position, $\vec p,\vec x\in\mathbb R^3,$ in analogy to (\ref{ham-q}). By rescaling, a region of fixed energy is identified with $\mathcal S^5\subset \mathbb R^6$. The solution passing through a point of $\mathcal S^5$ is the orbit of the point under an $SO(2)$-action, where $SO(2)$ acts via rotations on each $\mathbb R^2$ of pairs $(x_i, p_i)$. But the action of $SO(2)$ on $\mathbb R^2$ is equivalent to the action of $U(1)$ on $\mathbb C$, so the set of solutions of a classical $3$-d isotropic harmonic oscillator is identified with $\mathcal S^5/\mathcal S^1=\mathbb CP^2$.
\end{remark}

\section{A proof of Theorem \ref{ber-gen-prop}}\label{sec:ber_ker}

We start by proving for $\Pi_>$. Let $\rho$ be an irreducible representation of class $\vb*p$ in $\mathcal H_{\vb*p}$.  Now, the really hard-to-check property that the map
\begin{equation}\label{mapPi>}
    \mathcal B(\mathcal H_{\vb*p})\ni A\mapsto \big(\,B_A:\mathcal E\to \mathbb C: \vb z\mapsto B_A(\vb z)=\tr(A\Pi_>(\vb z))\, \big)
\end{equation}
needs to satisfy in order to be a symbol correspondence is injectivity.  Here we reproduce in greater detail, for the specific case of $SU(3)$, the argument  for injectivity of the map (\ref{mapPi>}) as presented by Wildberger for general compact semisimple Lie groups in \cite{wild}, and reworked by Figueroa, Gracia-Bondía and Várilly in \cite{figueroa}.  

Because multiplicities of weights are not relevant for the argument, instead of the Gelfand-Tsetlin labeling used throughout the rest of the paper, here we use the representation of weights as linear combinations of the fundamental weights $\{\omega_1,\omega_2\}$, cf. (\ref{fund-wei}). Recall $T_\pm$ and $U_\pm$ act on the weights \eqref{action_Tpm_Upm}. Also, we consider the lexicographical order for pairs of weights induced by \eqref{weight-order}.

We denote by $\mathcal H_{\vb*p}^\omega$ the subspace of $\mathcal H_{\vb*p}$ spanned by vectors with weight $\omega$, and let $\mathcal B_{\omega, \tau} = Hom(\mathcal H_{\vb*p}^\tau, \mathcal H_{\vb*p}^\omega)$ so that
\begin{equation*}
    \mathcal H_{\vb*p} = \bigoplus_\omega \mathcal H_{\vb*p}^\omega \ , \ \ \mathcal B(\mathcal H_{\vb*p}) = \bigoplus_{\omega,\tau} \mathcal B_{\omega, \tau} \ .
\end{equation*}
Given $A\in \mathcal B(\mathcal H_{\vb*p})$, we have a decomposition $A = \sum_{\omega, \tau}A_{\omega, \tau}$ such that $A_{\omega, \tau} \in \mathcal B_{\omega, \tau}$.

We now introduce the sets  
\begin{equation*}
  \mathcal A = \{A\in \mathcal B(\mathcal H_{\vb*p}): A\ne 0, B_A = 0\} \ , \ \  \mathcal P = \{(\omega, \tau): A_{\omega, \tau}\ne 0 \mbox{ for some } A\in \mathcal A\} \, . 
\end{equation*}

\begin{lemma}
If $\mathcal A \ne \emptyset$, then $\max \mathcal P = (\vb*p,\vb*p)$.
\end{lemma}
\begin{proof}
Since the order on pair of weights is only a partial order, there might be more than one maximal element in $\mathcal P$. Let $(\omega, \tau)$ be a maximal element of $\mathcal P$ and take $A\in \mathcal A$ satisfying $A_{\omega,\tau}\ne 0$. Given basis $\{u_1,...,u_n\}$ of $\mathcal H_{\vb*p}^\omega$ and $\{v_1,...,v_m\}$ of $\mathcal H_{\vb*p}^\tau$, we have
\begin{equation*}
    A_{\omega,\tau} = \sum_{j=1}^n\sum_{k=1}^m a_{j,k}\, u_j\otimes v_k^\ast = \sum_{k=1}^m\left(\sum_{j=1}^na_{j,k}u_j\right)\otimes v_k^\ast = \sum_{k=1}^m w_k\otimes v_k^\ast \ ,
\end{equation*}
where $w_k = \sum_{j=1}^na_{j,k}u_j$. Since $A_{\omega, \tau}\ne 0$, there is some $k_0 \in \{1,...,m\}$ such that $w_{k_0}\ne 0$. If $\omega <\vb*p$, there is $E_j \in \{E_1 = T_+, E_2 = U_+\}$ such that $E_j(w_{k_0})\ne 0$. So $[E_j,A]$ has a non zero component $[E_j,A]_{\omega+\alpha_j,\tau}$. However, by equivariance of $B$, we have  $B_A=0\Longrightarrow B_{[E_j,A]} = 0$, which contradicts the maximality of $(\omega, \tau)$ in $\mathcal P$. Thus, $\omega = \vb*p$, and
\begin{equation*}
    A_{\vb*p,\tau} = \sum_{k=1}^m a_k \, e_0\otimes v_k^\ast = e_0\otimes \left(\sum_{k=1}^na_kv_k^\ast \right) = e_0\otimes v^\ast \ ,
\end{equation*}
where $e_0$ is some unit vector in $\mathcal H_{\vb*p}^{\vb*p}$ and $v = \sum_{k=1}^m\overline a_k v_k$ is non zero. Now, if $\tau< \vb*p$, there is, again, $E_k \in \{E_1, E_2\}$ such that $E_k(v)\ne 0$, so $[E_k^\dagger, A]$ has a non zero component $[E_k^\dagger, A]_{\vb*p, \tau+\alpha_k}$, but $B_{[E_k^\dagger,A]} = 0$. Thus, $\tau = \vb*p$ and $(\omega,\tau) = (\vb*p,\vb*p)$.
\end{proof}

From the  above lemma, if $\mathcal A\ne \emptyset$, then there exists $A\in \mathcal A$ such that $A_{\vb*p,\vb*p} \ne 0$. But, given an unit vector $e_0 \in \mathcal H_{\vb*p}^{\vb*p}$,
\begin{equation*}
B_A(\vb z_0) = \tr(A\Pi_>) = \ip{e_0}{Ae_0} = \ip{e_0}{A_{\vb*p,\vb*p}e_0} = A_{\vb*p,\vb*p}\ne 0 \ ,   
\end{equation*}
a contradiction. Therefore, $\mathcal A = \emptyset$, that is, $B_A = 0$ only if $A = 0$, hence the map (\ref{mapPi>}) is injective. One can easily check that (\ref{mapPi>}) also satisfies all other properties in Definition \ref{def:symb_corresp}, thus the highest weight projector $\Pi_>$ is an operator kernel. 

Finally, because $\Pi_< = \Pi_>^{\widecheck \delta}$, the projector onto the lowest weight space is an operator kernel as well.

\end{document}